\documentclass[superscriptaddress,aps,twocolumn,10pt,longbibliography,nofootinbib]{revtex4-2}
\usepackage[utf8]{inputenc}
\usepackage[T1]{fontenc}
\usepackage{amsmath}
\usepackage{amssymb}
\usepackage{graphicx}
\usepackage{bm}
\usepackage{float}
\usepackage{multirow}
\usepackage{tabu}
\usepackage{longtable} 
\usepackage[dvipsnames]{xcolor}
\usepackage{soul}
\sethlcolor{yellow}
\usepackage[normalem]{ulem}
\usepackage{ragged2e}
\usepackage{subfigure}
\usepackage{epstopdf}
\usepackage{etoolbox}
\usepackage{booktabs}
\usepackage{braket}
\usepackage{bbm}
\usepackage{amsthm}

\usepackage{natbib}
\usepackage{qcircuit}
\usepackage{wasysym}
\usepackage{array}

\newcolumntype{L}[1]{>{\raggedright\let\newline\\\arraybackslash\hspace{0pt}}m{#1}}
\newcolumntype{C}[1]{>{\centering\let\newline\\\arraybackslash\hspace{0pt}}m{#1}}
\newcolumntype{R}[1]{>{\raggedleft\let\newline\\\arraybackslash\hspace{0pt}}m{#1}}

\newcommand{\tr}{\mathrm{Tr}}
\newcommand{\rhosep}{\rho_{\mathrm{sep}}}
\newcommand{\Zn}{Z^{\otimes N}}
\newcommand{\Xn}{X^{\otimes N}}
\newcommand{\Yn}{Y^{\otimes N}}
\newcommand{\norm}[1]{\left\lVert #1 \right\rVert}

\newtheorem{theorem}{Theorem}
\newtheorem{definition}{Definition}
\newtheorem{lemma}{Lemma}

\begin{document}
\bibliographystyle{plainnat}

\title{Practical Verification of Quantum Properties\\ in Quantum Approximate Optimization Runs}

 \author{M. Sohaib Alam}
 \email{sohaib@rigetti.com}
 \affiliation{Rigetti Computing, Berkeley, CA, 94701, USA}
 \author{Filip A. Wudarski}
 \email{filip.a.wudarski@nasa.gov}
 \affiliation{Quantum Artificial Intelligence Laboratory (QuAIL), NASA Ames Research Center, Moffett Field, CA, 94035, USA}
 \affiliation{USRA Research Institute for Advanced Computer Science (RIACS), Mountain View, CA, 94043, USA}
 \author{Matthew J. Reagor}
 \email{matt@rigetti.com}
 \affiliation{Rigetti Computing, Berkeley, CA, 94701, USA}
\author{James Sud}
\affiliation{Quantum Artificial Intelligence Laboratory (QuAIL), NASA Ames Research Center, Moffett Field, CA, 94035, USA}
 \affiliation{USRA Research Institute for Advanced Computer Science (RIACS), Mountain View, CA, 94043, USA}
\author{Shon Grabbe}
\affiliation{Quantum Artificial Intelligence Laboratory (QuAIL), NASA Ames Research Center, Moffett Field, CA, 94035, USA}

\author{Zhihui Wang}
\affiliation{Quantum Artificial Intelligence Laboratory (QuAIL), NASA Ames Research Center, Moffett Field, CA, 94035, USA}
 \affiliation{USRA Research Institute for Advanced Computer Science (RIACS), Mountain View, CA, 94043, USA}
\author{Mark Hodson}
\affiliation{Rigetti Computing, Berkeley, CA, 94701, USA}

\author{P. Aaron Lott}
\affiliation{Quantum Artificial Intelligence Laboratory (QuAIL), NASA Ames Research Center, Moffett Field, CA, 94035, USA}
 \affiliation{USRA Research Institute for Advanced Computer Science (RIACS), Mountain View, CA, 94043, USA}
\author{Eleanor G. Rieffel}
\affiliation{Quantum Artificial Intelligence Laboratory (QuAIL), NASA Ames Research Center, Moffett Field, CA, 94035, USA}
\author{Davide Venturelli}
\email{dventurelli@usra.edu}
\affiliation{Quantum Artificial Intelligence Laboratory (QuAIL), NASA Ames Research Center, Moffett Field, CA, 94035, USA}
\affiliation{USRA Research Institute for Advanced Computer Science (RIACS), Mountain View, CA, 94043, USA}

\begin{abstract}

In order to assess whether quantum resources can provide an advantage over classical computation, it is necessary to characterize and benchmark the non-classical properties of quantum algorithms in a practical manner. In this paper, we show that using measurements in no more than 3 out of the possible $3^N$ bases, one can not only reconstruct the single-qubit reduced density matrices and measure the ability to create coherent superpositions, but also possibly verify entanglement across all $N$ qubits participating in the algorithm. We introduce a family of generalized Bell-type observables for which we establish an upper bound to the expectation values in fully separable states by proving a generalization of the Cauchy-Schwarz inequality, which may serve of independent interest. We demonstrate that a subset of such observables can serve as entanglement witnesses for QAOA-MaxCut states, and further argue that they are especially well tailored for this purpose by defining and computing an {\it entanglement potency} metric on witnesses. A subset of these observables also certify, in a weaker sense, the entanglement in GHZ states, which share the $\mathbb{Z}_2$ symmetry of QAOA-MaxCut. The construction of such witnesses follows directly from the cost Hamiltonian to be optimized, and not through the standard technique of using the projector of the state being certified. It may thus provide insights to construct similar witnesses for other variational algorithms prevalent in the NISQ era. We demonstrate our ideas with proof-of-concept experiments on the Rigetti Aspen-9 chip for ansatze containing up to 24 qubits.

\end{abstract}

\maketitle

\section{\label{sec:introduction}Introduction}

At present, there is great interest in designing and validating the implementation of solvers of computational problems that can achieve a \emph{quantum advantage}, i.e. superior performance with respect to other known classical methods. There is considerable research activity in the study of speedup metrics~\cite{ronnow2014defining} and in the certification and quantification of quantum properties (which we later also refer to as {\it quantumness}) in the context of variational quantum algorithms \cite{Woitzik_2020,Wiersema_2020,diezvalle2021quantum}.
It is however of utmost importance in the NISQ era to understand the role of coherence and entanglement, as well as to guarantee that the software-hardware systems that we use for the empirical exploration exploit these resources to outperform classical competitors.
This task is distinct from evaluating the fidelity of the experimental results with respect to a simulation, and instead focuses on verifying the key non-classical properties of a solver that may achieve quantum advantage on noisy hardware whose exact mechanics may not be easily described through analytical expressions or even numerical simulations.

\begin{figure}
    \centering
    \includegraphics[width=0.48\textwidth]{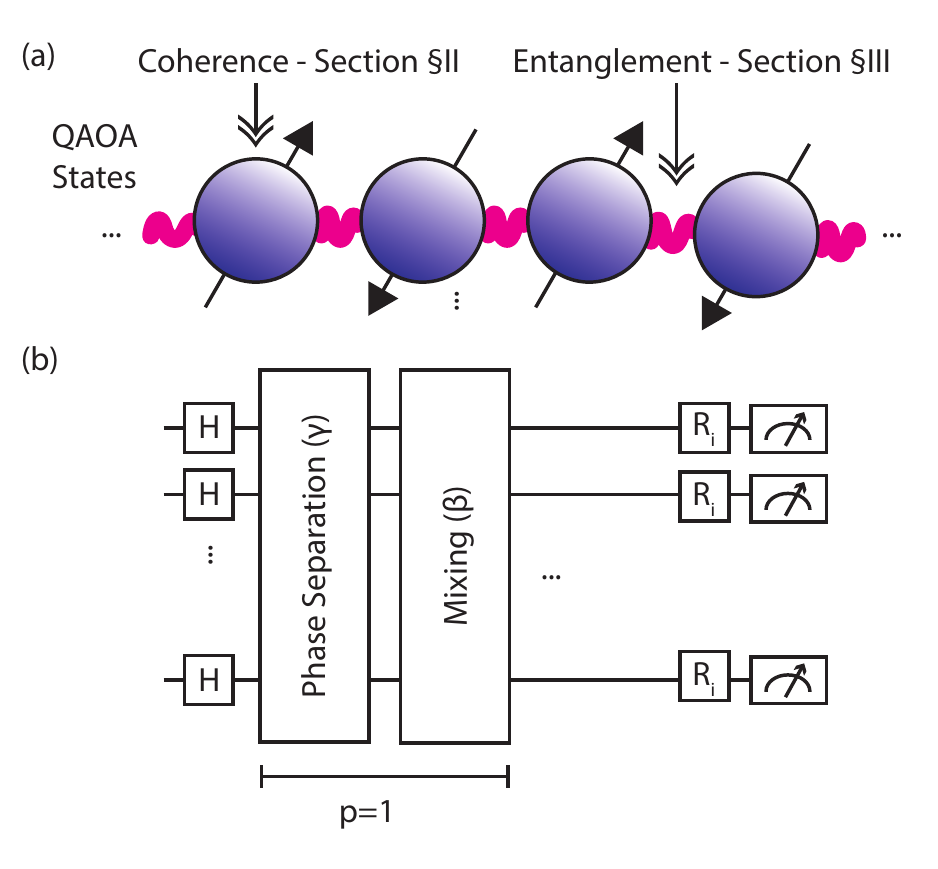}
    \caption{Practical Verification of Quantum Properties in QAOA: (a) QAOA leverages quantum coherent superposition and entanglement to perform combinatorial optimization. The goal of this work is to develop probes of Coherence (Sect.~\ref{sec:coherences}) and Entanglement (Sect.~\ref{sec:entanglement}) in states generated by the QAOA algorithm. (b) The specific ansatz studied here consists of an initial register state preparation via transversal Hadamard gates, followed by $p$ applications of phase separation and mixing Hamiltonians. To reveal coherence and entanglement properties, additional single-qubit gates rotate the measurement basis before reading out the state of the register. We analyze experimental results for linear chain problems up to twenty-four qubits and a single round ($p=1$) of QAOA.}
    \label{fig:overview}
\end{figure}

Bearing in mind the lively debate over approaches to verify the presence of quantum properties in experimental runs on quantum annealers~\cite{job2018test}, as well as the recent literature on quantum volume \cite{Moll_2018,PhysRevA.100.032328} and on entanglement characterization and detection in NISQ devices \cite{wang2020detecting,cornelissen2021scalable}, in this paper we take a pragmatic approach trying to answer the question of whether any quantum resources survive following the execution of a specific quantum algorithm, applied to a specific problem, run on a specific quantum processing device. Our proposed tests do not necessarily guarantee that the detected quantumness was exploited computationally prior to detection, however the analysis can be coupled with complementary simulations that can evaluate the computational value of the measured resource. While our algorithm of choice is the quantum approximate optimization algorithm \cite{farhi2014quantum} (QAOA), applied to the MaxCut problem, the operational character of the study will illuminate what needs to be done in practice for other choices of algorithms, problems and devices.

The QAOA for MaxCut problem (QAOA-MaxCut) can be regarded as the minimization of a 2-local $N$-qubit cost Hamiltonian $C = \sum_{\langle i, j \rangle} Z_i Z_j$, defined over some set of edges $\langle i,j \rangle$ in a circuit ansatz given by
\begin{equation}
    \ket{\psi_{QAOA}} = \left( \prod_{j=1}^{p} e^{-i\beta_j B} e^{-i\gamma_j C} \right) \ket{+}^{\otimes N},
\label{eqn:QAOA-ansatz}
\end{equation}
where $B = \sum_{i=1}^{N} X_i$ is the \textit{mixer}, and the ansatz parameters $\beta_1, \dots, \beta_p, \gamma_1, \dots, \gamma_p$ are to be optimized to minimize the expectation value of $C$ in the state $\vert \psi_{QAOA} \rangle$. The quantum device used is the Rigetti Apsen-9 system \cite{Caldwell2018,Reagor2018}, comprising 31 operational transmon qubits with median two-qubit gate fidelity of 97\%. Figure~\ref{fig:overview} provides an overview of our overall approach.

The $\mathbb{Z}_{2}$ symmetry of the MaxCut cost Hamiltonian ensures that QAOA-MaxCut states prepared via Eq. \eqref{eqn:QAOA-ansatz} inherit the symmetry so that the amplitude of some computational basis state is the same as that of its ones' complement, i.e. $\langle x \vert \psi_{QAOA}\rangle = \langle \bar{x} \vert \psi_{QAOA}\rangle$, where $\vert \bar{x}\rangle=\Xn\vert x \rangle$.
Although non-entangled states can exhibit such a symmetry as well, Eq. \eqref{eqn:QAOA-ansatz} with a $2$-local cost Hamiltonian $C$ ensures that the state $\vert \psi_{QAOA}\rangle$ can possess entanglement at every value of $p$. In the limit of very large $p$ and assuming a single MaxCut solution, QAOA would yield a solution of the form $\frac{1}{\sqrt{2}} \left( \ket{x} + \ket{\bar{x}} \right)$, where both the bitstrings $x$ and $\bar{x}$ specify the same `cut' through the set of edges, and just differ in how they label each partition. This kind of entanglement is strongly reminiscent of the GHZ state, and in Section \ref{sec:entanglement} the similarity between the witnesses constructed for GHZ states as well as QAOA-MaxCut states becomes obvious. Interestingly, we observe that there is a correlation between the value of $p$ for which QAOA-MaxCut saturates the maximum expectation value of the entanglement witness, and the value of $p$ for which it saturates the maximum expectation value of the cost Hamiltonian, suggesting that the entanglement certified by these observables is actively participating in the computational optimization. We note that earlier works \cite{Bravyi_2020,shaydulin2020classical} have considered the relationship of the underlying $\mathbb{Z}_{2}$ symmetry of the cost Hamiltonian and QAOA performance, and hope that our present work also sheds light on how entanglement is employed as a resource in QAOA.

By probing single-qubit reduced density matrices (SQRDMs), we are able to identify which qubits in the chip participate in the computation with the ability to create coherent superpositions of the classical bit values 0 and 1. Moreover, using the same measurement data as required to construct the SQRDMs, we are able to infer the expectation values of observables that serve as entanglement witnesses for some quantum state(s). The ease with which these witnesses can be measured make them a practical toolkit for a quantum programmer to quickly verify if any entanglement has been generated in the executed circuit. Although such a procedure could work for arbitrary circuits in principle, we demonstrate that the introduced family of witnesses is particularly well suited to detecting entanglement in QAOA circuits. We do so by explicitly showing that such observables do indeed serve as witnesses for QAOA-MaxCut states, and by introducing and computing an {\it entanglement potency} metric on witnesses, which quantifies the fractional volume of states whose entanglement is detected by the given witness. We compute this metric for QAOA and Haar random states.

A standard method to construct an entanglement witness \cite{entanglement-detection-review} for some state $\ket{\Psi}$ is to construct an operator out of a projector
\begin{equation}
    W = \alpha \mathbbm{1} - \ket{\Psi}\bra{\Psi}
\label{eqn:standard-EW-construction}
\end{equation}
where $\alpha$ is the maximum fidelity of the state $\ket{\Psi}$ with those in some set of states that we wish to certify against. However, for a parametric family of states such as QAOA-MaxCut states, the construction in Eq. \eqref{eqn:standard-EW-construction} not only depends on the choice of parameters, but may also require measurements in a large number of different bases as well as terms to estimate the expectation values of. In previous work, it has been shown that for certain $N$-qubit pure states, there exists a witness that requires $2N-1$ measurements \cite{PhysRevA.76.022330,PhysRevA.76.030305}. Here, we construct observables that require no more than 3 bases measurements for any value of $N$, and at most a polynomial number of terms to be estimated given the measurement data. While the theorems we prove only establish that bounds by fully separable states can be violated, we also provide numerical evidence that a subset of such observables can certify against bi-separable (and by extension $k$-separable, for $k \geq 2$) states, and therefore witness genuine $N$-partite entanglement.

The remainder of this article is organized as follows. In Section \ref{sec:coherences}, we discuss the measurement of coherences at the single-qubit level, and their relevance as a measure of quantumness. We also present experimental results from the Rigetti Aspen-9 quantum device. Section \ref{sec:entanglement} contains the main focus of our paper, wherein we discuss the construction of a family of entanglement witnesses, and establish several of their properties using a mix of analytical and numerical results. In particular, we analytically prove an upper bound to the absolute value of the expectation of a large class of observables in fully separable states using a generalized Cauchy-Schwarz inequality, which we also prove as a lemma, as well as a lower bound to the maximum eigenvalue for a subset of such observables, thus proving that they can serve as witnesses for some entangled states. We also identify a particular class of observables that are especially relevant for QAOA-MaxCut, and prove that they serve as witnesses for such states prepared via the ring cost Hamiltonian. In addition, we also provide numerical evidence for the view that similar observables constructed for other cost Hamiltonians also serve as witnesses for the corresponding QAOA-MaxCut states. We also demonstrate the relevance of the constructed observables to QAOA-MaxCut through the calculation of an entanglement potency metric, which we introduce. We finish this section by presenting proof-of-concept experimental results from the Rigetti Aspen-9 quantum chip. We conclude in Section \ref{sec:Conclusions}, discussing future avenues of exploration that our work opens up.

\section{Quantum Coherence}\label{sec:coherences}

The most iconic property of a quantum system is the ability to be in a \emph{coherent superposition} of multiple reference states. This basic distinctive property of quantum information processing and its direct manifestations (e.g. tunneling) is believed to be essential to achieving any kind of quantum advantage~\cite{DiVincenzo_1999}. Formally, a system is said to be in coherent superposition with respect to a basis of states, if its density matrix is non-diagonal when expressed in that basis~\cite{cohen1962theorie}. In combinatorial optimization, the specification of both the problem variables as well as the solution, or an algorithm's approximation to it, are efficiently representable with a polynomial number of classical bits. A quantum algorithm for combinatorial optimization might in general exploit the coherent superposition of all $2^N$ computational basis states. However, probing these coherences require a prohibitively exponential experimental cost.

Without requiring a full tomographic reconstruction of the many-body density matrix of a quantum state, we propose to use a method which allows us to extract information about off-diagonal elements of all single-qubit reduced density matrices (SQRDMs) expressed in the computational basis, which we later use as a threshold figure of merit to determine the extent of quantumness. If sizeable values of these elements are measured at any time during the execution of an algorithm, it can be concluded that the individual qubits at that time still hold the potential to exploit single-particle quantum effects for the purpose of information processing. A generalization of the arguments to multi-particle coherent effects (i.e. co-tunneling) would be straightforward in principle, though with an exponentially greater experimental cost.

\subsection{Efficient SQRDM Tomography}
Suppose that we have a single qubit quantum state, either pure or mixed, represented by the density matrix
\begin{equation}
    \rho = \left(\begin{array}{cc}\rho_{11} & \rho_{12} \\\rho_{21} & \rho_{22}\end{array}\right),
\end{equation}
for which we identify the off-diagonal elements $\rho_{12}=\rho_{21}^*$ as the ``coherences''.
Here, we choose to use the absolute value of coherence, i.e. $C_\rho = \vert\rho_{12}\vert$, as our metric of interest. Since $\rho\geq 0$ and $\tr{(\rho)}=1$, the maximum possible value attainable for $C_{\rho}$ is $\frac{1}{2}$, and is saturated for example by the $\ket{+}$ state. Each single qubit state can be experimentally reconstructed via state tomography \cite{D_Ariano_2002,Blume_Kohout_2010,d2003quantum} using measurements in the Pauli $X, Y$ and $Z$ bases, which allow us to determine expectation values of Pauli matrices from the outcome statistics and represent the state as 
\begin{equation}
    \rho  = \frac{1}{2}(\mathbbm{1}+\langle X\rangle X + \langle Y\rangle Y + \langle Z\rangle Z).
\label{eq:bloch_rep}
\end{equation}
 One could argue that only two measurements suffice to measure $C_\rho = \frac{1}{2}|\langle X\rangle + i  \langle Y\rangle|$. However, since the measurement protocol is inevitably noisy, it is more reasonable to first tomographically reconstruct the state using maximum likelihood estimation (MLE) \cite{vasconcelosextending}, which would guarantee that the measured state is a legitimate quantum object, i.e. being normalized $\tr(\rho) = 1$, and positive semi-definite $\rho\ge0$, and then compute the required expectation values in this state to measure $C_p$.

The above protocol easily generalizes\footnote{With the assumption that the reduced density matrices can be reliably estimated via single-qubit MLE.} - without additional measurement overhead - to determine the SQRDMs for an $N$-qubit state produced by a quantum algorithm, such as QAOA. In order to achieve this, we measure all (algorithmically relevant) qubits simultaneously in the $X, Y$ and $Z$ bases and determine $\rho_k$ as in Eq.~\eqref{eq:bloch_rep} where now $\langle Q_k\rangle\equiv\langle \mathbbm{1}_1\otimes\ldots Q_k\otimes\ldots\rangle$ denotes the expectation value of $Q = X,Y,Z$ acting on the $k$th qubit. In this, way we partially reconstruct the final state after the measurements and the MLE procedure as 
$$
\rho_{\mathrm{single-particle}}= \bigotimes_{k=1}^N \rho_k,  
$$
with all $\rho_k$ being legitimate states. It is important to note that 
$\rho_{\mathrm{single-particle}}$ 
is not a reconstruction of the density matrix of the state of the quantum processor before the measurements.
Most notably, 
it is a product state, and caries no information about other quantum features such as entanglement \cite{HorodeckiQuantumEntanglement} or discord \cite{Bera_2017}, which have been washed out in the course of local measurements.
Nevertheless, the detection of a non-zero coherence in a $\rho_k$ would indicate that superposition was still possible for qubit $k$ before measurement, so that some single-particle quantumness was exploitable.
A figure of merit for this kind of limited quantumness could be identified in the ``maximum observed single-particle coherence'', i.e. the largest measured value of $C_{\rho_k}$ over both the qubits and the parameters that define the circuit.

\subsection{SQRDMs in QAOA for MaxCut}
Now let us demonstrate how the above protocol works for a QAOA algorithm. Let $d$ denote the degree of the qubit (node) in the underlying MaxCut problem graph. The QAOA ansatz is $\vert \psi(\gamma, \beta) \rangle = e^{-i\beta B} e^{-i\gamma C} \ket{+}^{\otimes N}$ with the mixer $B = \sum_{j=1}^{N} X_i$ and the cost function $C = \sum_{\langle ij \rangle} Z_i Z_j$. Then, for a single layer ($p=1$) the SQRDMs can be analytically calculated to be
\begin{equation}
    \rho_{QAOA} = \frac{1}{2} \begin{pmatrix}
    1 & \cos^{d}{(2\gamma)} \\
    \cos^{d}{(2\gamma)} & 1
    \end{pmatrix},
\label{eqn:QAOA-SQRDM}
\end{equation}
Note the independence from $\beta$ and the total number of qubits $N$. Eq. \eqref{eqn:QAOA-SQRDM} tells us that $C_{\rho_{QAOA}} = \frac{1}{2} \vert \tr{\left(\rho_{QAOA} (X + iY) \right)} \vert = \frac{1}{2} \vert \cos^d(2\gamma) \vert$. Moreover, the fidelity \cite{JozsaFidelity} of this state with an arbitrary single-qubit state of the form Eq. (\ref{eq:bloch_rep}) is found to be
\begin{eqnarray}
    F(\rho_{QAOA}, \rho) = \frac{1}{2} \left( 1 + \langle X \rangle \cos^{d}{2\gamma} \right) + \nonumber \\
    \frac{1}{2} \left( 1 - \cos^{2d}{2\gamma}\right)^{1/2} \cdot \left( 1 - \vert \vec{P} \vert^{2} \right)^{1/2},
\label{eq:fidelity-qaoa-sqrdm}
\end{eqnarray}
where $\vec{P} = \left( \langle X \rangle, \langle Y \rangle, \langle Z \rangle \right)$. Note that this fidelity explicitly depends on $\gamma$. Of course, when $\vec{P} = \left(cos^{d}{2\gamma}, 0, 0 \right)$, the above expression equals one for all $\gamma$.

We can similarly compute the fidelity between the ideal QAOA state in Eq. \eqref{eqn:QAOA-SQRDM} and a classical probabilistic bit
\begin{equation}
\rho_{cl}(\theta) = \begin{pmatrix}
\cos^{2}{\theta} & 0 \\
0 & \sin^{2}{\theta}
\end{pmatrix},
\label{eq:prob-bit}
\end{equation}
In particular, we consider the maximum fidelity over all such probabilistic bits (max over $\theta$), given by
\begin{equation}
    \max_{\theta} F \left( \rho_{QAOA}, \rho_{cl} \right) = \frac{1}{2} + \frac{1}{2} \sqrt{1 - \cos^{2d}{2\gamma}}.
\label{eq:max-fid-classical-bit}
\end{equation}
States for which the difference between Eqs. (\ref{eq:fidelity-qaoa-sqrdm}) and (\ref{eq:max-fid-classical-bit}) is greater than zero represent resources that cannot be reproduced or employed by any classical device. We can therefore use this difference as a measure of whether quantum resources have been employed by a physical quantum device when running QAOA.

\subsection{Results - Coherences}\label{subsec:res_coherence}

We use the methods described above to report on the amount of coherences each single qubit can build during the circuit execution for a fixed $\beta=\frac{\pi}{8}$ angle \footnote{While the coherence for a single qubit does not depend on our choice of $\beta$, this particular choice was used to maximize the ideal expectation value of the entanglement witness we describe for the linear chain QAOA-MaxCut circuit, see Sections. \ref{subsectn:qaoa-witness} and \ref{sec:results}.} and across different $\gamma$ angles. Fig. \ref{fig:coherences} shows the coherences $C_{\rho_{QAOA}}$ for each SQRDM in 2, 8, 16 and 24-qubit linear chain QAOA-MaxCut circuits. In the ideal noiseless case, these should follow precisely the form $\frac{1}{2}\vert \cos^{d}{2\gamma}\vert$ where $d$ is the degree of the qubit whose coherence is being plotted. Experimentally, we observe a qualitatively similar pattern, except scaled down from its maximum value of $\frac{1}{2}$.

\begin{figure}
    \centering
    \includegraphics[width=0.48\textwidth]{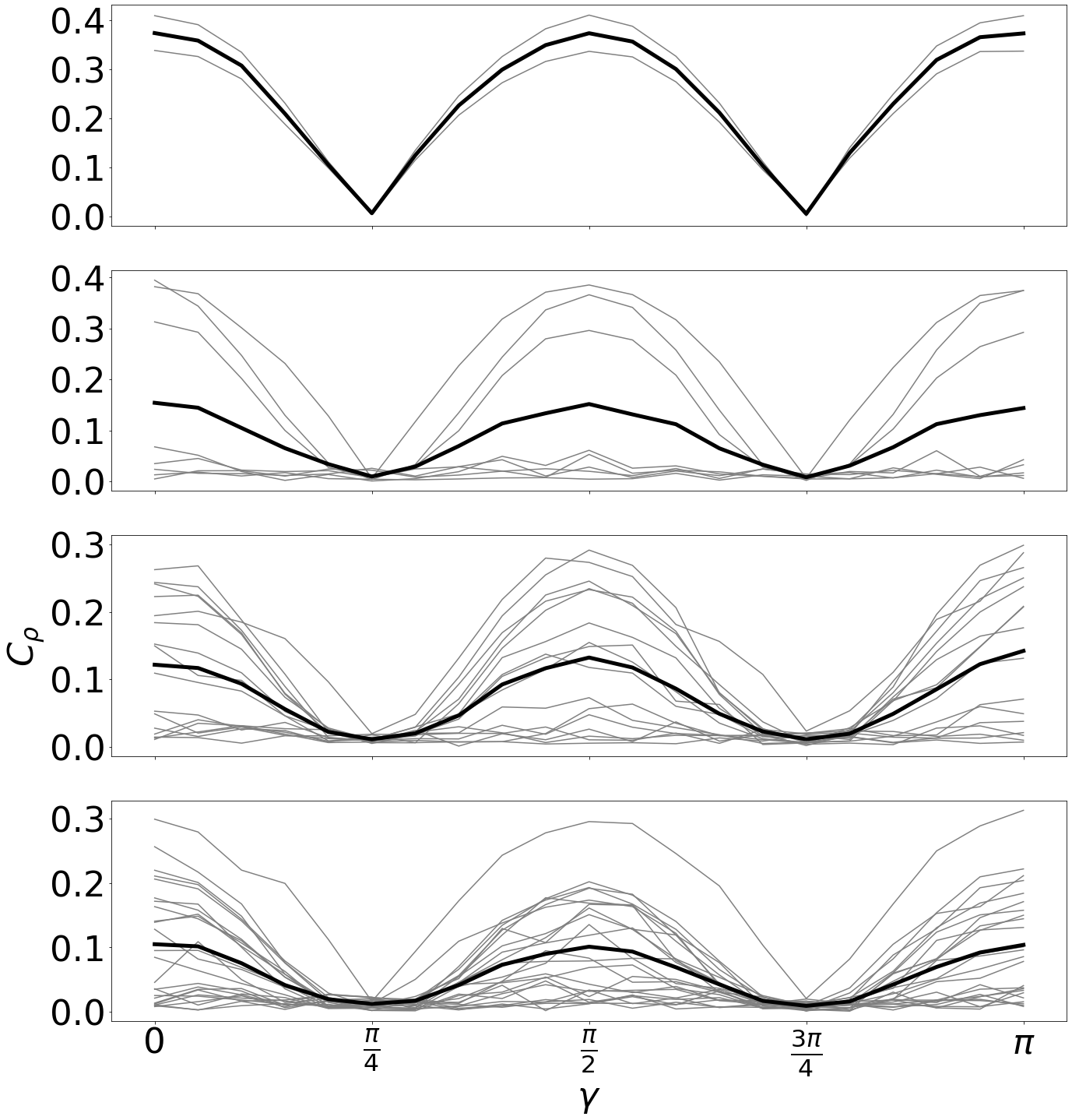}
    \caption{Coherences $C_{\rho}$ as a function of $\gamma$ for a fixed value of $\beta=\frac{\pi}{8}$ for qubits (thin, gray lines) in a linear chain QAOA-MaxCut circuit of (top to bottom) 2, 8, 16 and 24 qubits. The thick (black) line represents the mean coherence across all the qubits.
    }
    \label{fig:coherences}
\end{figure}

In Fig. \ref{fig:coherence_sqrdm_n24}, we compute the fidelity of the experimental SQRDM of a qubit (No. 0 - see Appendix~\ref{secn:experiment} for the chip's layout) after having ran a 24-qubit linear chain QAOA-MaxCut circuit with the ideal SQRDM for that qubit, calculated using Eq. (\ref{eq:fidelity-qaoa-sqrdm}).
As a comparison, we also plot the maximum fidelity achievable by a classical probabilistic bit. The difference between these two plots provides a measure of the non-classical coherent superposition effects that this qubit can experience during the course of a subsequent quantum computation.
If this difference is zero, then at the time of measurement, the qubit offered no more computational power than that offered by a probabilistic classical bit. Fig. \ref{fig:coherences-rainbow} in Appendix \ref{secn:experiment} provides more information on the coherences in each of the qubits in the executed circuits, as well as their physical locations on the Aspen-9 chip.

\begin{figure}
    \centering
    \includegraphics[width=0.48\textwidth]{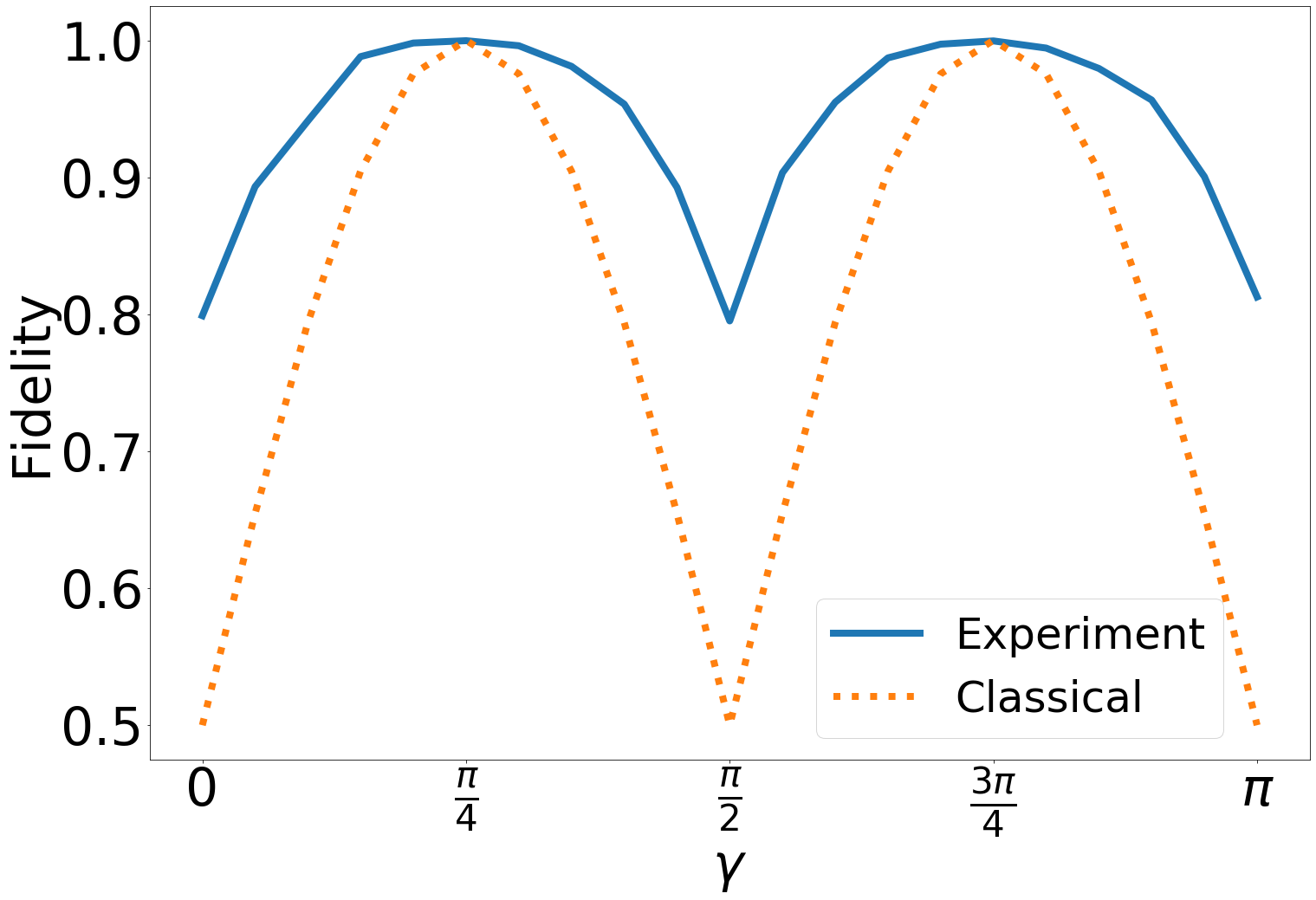}
    \caption{Experimental fidelity with the ideal SQRDM for a qubit (No. 0 - see Appendix~\ref{secn:experiment} for chip's layout) in a 24 linear chain QAOA-MaxCut circuit, compared with the maximum achievable fidelity by a classical probabilistic bit, or a fully decohered qubit.
    }
    \label{fig:coherence_sqrdm_n24}
\end{figure}

\section{\label{sec:entanglement}Entanglement}

\emph{Quantum entanglement}~\cite{HorodeckiQuantumEntanglement}  is a property of a physical system that relates to the intrinsic non-decomposable character of its state. It is famously manifested as a type of correlation that distinguishes classical physics, which is local and deterministic, from quantum mechanics. While it is still debated to what extent non-local correlations are necessary for quantum advantage~\cite{Jozsa_2003, datta2008quantum}, it is certainly a signature feature of the QAOA protocol and of its relationship with adiabatic quantum computation.

Methods to estimate, quantify, and detect entanglement are highly valuable assets for benchmarking quantum hardware. Measures such as bipartite and multipartite concurrence  \cite{PhysRevLett.78.5022,Wootters_1998, conc}, entanglement of formation \cite{Wootters_1998}, von Neumann entropy and many more exist \cite{HorodeckiQuantumEntanglement,plenio2014introduction,horodecki2001entanglement}. However they are typically difficult to compute and experimentally costly to estimate.

In this work, we construct a family of observables that act as \emph{entanglement witnesses} \cite{chruscinski2014entanglement,Bae_2020} for some quantum state(s). These are observables whose expectation value is bounded for separable states, so that a violation of such a bound provides a signature of entanglement.
Although we focus our attention on QAOA-MaxCut states, the observables we introduce below may also serve as witnesses for other types of states, such as parametric families of states found in variational quantum algorithms, e.g. variational quantum eigensolver (VQE) \cite{peruzzo2014variational,mcclean2016theory}, quantum machine learning \cite{biamonte2017quantum,farhi2018classification}, or variational Hamiltonian ansatz \cite{Wiersema_2020, Wecker_2015}.
From the experimentalist's perspective, it is a simple matter to run a circuit and obtain measurements in the $X$, $Y$ and $Z$ bases, and then to simply check whether an observable constructed as described below violates a bound (see Theorem \ref{thm:sep-W_M_k}). If it does, the experimentalist has verified the presence of entanglement in the circuit, regardless of whether that circuit runs QAOA or some completely different algorithm.\footnote{If it fails, this method cannot say anything about whether the state is entangled or not.} In this sense, the methods we describe below can be perceived as algorithm agnostic, and can be considered to be a general purpose benchmarking tool that can be easily and efficiently deployed for an arbitrary algorithm and system size. Additionally, this method can also be deployed to noisy circuits that are prevalent in the NISQ era, and is not just restricted to detecting entanglement in pure states.

In this section, we first describe general properties of a family of observables and show that they can serve as entanglement witnesses \cite{chruscinski2014entanglement} for some quantum states. We then show that a subset of such observables are capable of certifying entanglement in QAOA-MaxCut states. Additionally, we discuss the separability properties of these witnesses. Lastly, we introduce an {\it entanglement potency} metric on a witness that quantifies the relevance of a witness to a family of states. We conclude the section with experimental results from the Aspen-9 quantum chip.

\subsection{\label{subsec:bell_observables}Generalized Bell-type Observables}
Here, we introduce observables that are sums of $k$-local Pauli operators and establish bounds on the expectation values of these observables achievable by fully separable and general quantum states. Our aim is to demonstrate that a large family of such observables can serve as entanglement witnesses up to a spectral shift \cite{chruscinski2014entanglement,Bae_2020}.
Inspired by both Bell-type inequalities \cite{Bell2004-BELSAU, Aspect1982, Aspect1982_2}, we introduce a class of observables $W$ with the property
\begin{eqnarray}
\forall_{\rho_{\mathrm{sep}}}  &\qquad& \Lambda_{\min} \le \tr(W \rhosep) \le \Lambda_{\max}
\label{eq:W_upper_lower},
\end{eqnarray}
i.e. bounded by separable thresholds ($\Lambda_{\min}, \Lambda_{\max}$).\footnote{Note that both the upper and lower bounds can be easily transformed to form standard entanglement witness bounds through the addition/subtraction of a scaled identity operator.}
The key ingredient of our construction is the shared structure composed of $k$-local Pauli operators, and which for an exponentially large subset requires only three types of measurements ($X,Y$ and $Z$) in order to infer the expectation value and determine if a state violates any of the separable thresholds.
In the remainder of this paper, we will mostly focus on violations of the upper bound in Eq. \eqref{eq:W_upper_lower}, though one could also similarly analyze violations of the lower bound and the resulting signatures of entanglement.

To motivate the construction of our observables, we recall the usual Bell violation. In that scenario, given
\begin{eqnarray}
Q = Z \otimes I ,\; & S = \frac{1}{\sqrt{2}} \left( I \otimes Z + I \otimes X, \right),\nonumber \\
R = X \otimes I ,\;\;\; & T = \frac{1}{\sqrt{2}} \left( - I \otimes Z + I \otimes X \right), \nonumber
\label{eq:bell-2q-qrst}
\end{eqnarray}
the observable $W_B = (R+Q)S + (R-Q)T = \sqrt{2} \left( X \otimes X + Z \otimes Z\right)$ serves as a witness for the Bell state $\vert \psi \rangle = \frac{1}{\sqrt{2}} \left( \vert 00 \rangle + \vert 11 \rangle \right)$. Below, we describe successive generalizations of this observable.

A particularly simple and immediate one over $N$ qubits is defined as
\begin{eqnarray}
W_{XZ}^{(2,N,\mathcal{G})} &=& \sum_{\langle i,j \rangle} \left( X_i \otimes X_j + Z_i \otimes Z_j\right),
\label{eqn:W_XZ_2}
\end{eqnarray}
where the sum over $\langle i,j \rangle$ runs over some subset $E_2$ of all possible ${N \choose 2}$ edges in some graph $\mathcal{G}$ of $N$ nodes. We can generalize this further to similarly structured $k$-local observables of the following form
\begin{eqnarray}
W_{PQ}^{(k,N,\mathcal{G})} &=& \sum_{\langle i_1, \dots, i_k \rangle} \left( \bigotimes_{j=i_1}^{i_k}P_j + \bigotimes_{j=i_1}^{i_k} Q_j\right),
\label{eqn:W_PQ_k}
\end{eqnarray}
where $P,Q\in \{X,Y,Z\}$ and $P\neq Q$, the sum  over $\langle i_1, \dots, i_k \rangle$ similarly runs over some subset of all possible ${N \choose k}$ tuples $E_k$ (generalized edges) of some (generalized) graph $\mathcal{G}$ of $N$ nodes.\footnote{It is understood that in the local terms appearing in Eqs. (\ref{eqn:W_XZ_2}) and (\ref{eqn:W_PQ_k}), the identity operator acts on any of the $(N-k)$ qubits that do not appear in the tensor product. Thus, for example, the $2$-local $5$-qubit observable $X_1 \otimes X_4$ would really denote $X_1 \otimes I_2 \otimes I_3 \otimes X_4 \otimes I_5$.}

We can generalize this construction to observables of the form
\begin{equation}
    W_{XYZ}^{(k,N,\mathcal{G})} = \sum_{\langle i_1,\dots,i_k \rangle} \left( \bigotimes_{j=i_1}^{i_k}  X_j + \bigotimes_{j=i_1}^{i_k} Y_j  + \bigotimes_{j=i_1}^{i_k} Z_j \right)
\label{eqn:W_XYZ_k}
\end{equation}
where we have used the same notation as in Eq. (\ref{eqn:W_PQ_k}). 
Both these types of observables themselves belong to a much larger family, with variable coefficients as well as Pauli operators, as described in Appendix \ref{secn:general-obs}. This entire family of observables admits an upper and a lower bound in the expectation value achievable by any fully separable state, as shown in Theorem \ref{thm:sep-W_M_k}. In order to prove this theorem, we first introduce the following lemma.

\begin{lemma}\textbf{Generalized Cauchy-Schwarz inequality:}
Given a collection of vectors $\vec{x}^{(1)}, \dots, \vec{x}^{(k)}$, where $\vec{x}^{(i)} \in \mathbb{R}^{n}$ for $i \in \{1, \dots, k\}$, we have the following inequality
\begin{equation}
\left\vert \sum_{i=1}^{n} \left( \odot_{j=1}^{k} \vec{x}^{(j)} \right)_i \right\vert \leq \prod_{j=1}^{k} \norm{\vec{x}^{(j)}}
\label{eqn:lemma-cs-inequality}
\end{equation}
where $\vert \cdot \vert$ denotes the absolute value, $\norm{\cdot}$ denotes the Euclidean norm, and $\odot$ denotes the Hadamard product.
\label{thm:cauchy-schwarz}
\end{lemma}
\begin{proof}
See Appendix \ref{secn:proofs}.
\end{proof}
Note that Lemma \ref{thm:cauchy-schwarz} reduces to the usual Cauchy-Schwarz inequality \cite{cauchy-schwarz} for the Euclidean dot product in the base case of $k=2$. Based on this lemma, we show that the expectation value with respect to any fully separable state of any observable of the form Eq. (\ref{eqn:W_PQ_k}) or (\ref{eqn:W_XYZ_k}) or in fact a much larger family of observables described in Appendix \ref{secn:general-obs}, satisfies an upper and a lower bound given by (up to a sign factor) the number of generalized edges.
\begin{theorem}
The absolute value of the expectation value of any $N$-qubit $k$-local ($k \leq N$) observable $W_{M}^{(k,N,\mathcal{G})}$ of the form Eq. (\ref{eqn:W_M_kN}) in any fully separable quantum state is upper bounded as
\begin{eqnarray}
\left\vert \langle W_{M}^{(k,N,\mathcal{G})} \rangle_{\text{sep}} \right\vert \leq \vert E_k \vert
\end{eqnarray}
where $\vert E_k \vert$ denotes the number of $k$-tuples $\langle i_1, \dots, i_k\rangle$ (or generalized edges of the generalized graph $\mathcal{G}$) being summed over, and $M\in\{1,2,3\}$ denotes the number of Pauli operators defined on a single $k$-tuple (see Appendix~\ref{secn:general-obs} for notational convention).
\label{thm:sep-W_M_k}
\end{theorem}
\begin{proof}
See Appendix \ref{secn:proofs}.
\end{proof}

Theorem \ref{thm:sep-W_M_k} applies to all the observables that we have introduced so far, as well as a much larger set. It allows us to fulfill one of the criteria, Eq. \eqref{eq:W_upper_lower}, for detecting entanglement in the system. However, in order to establish these observables as genuine entanglement witnesses, one needs to show that the absolute expectation value can exceed the separable threshold, i.e. $\left\vert \langle W^{(k,N,\mathcal{G})}_M\rangle \right\vert > |E_k|$ for any $M \in \{2,3\}$ (see Appendix \ref{secn:general-obs} for notational conventions). In general, it is difficult to analytically establish lower bounds for an arbitrary observable in this family. Indeed, some may not even serve as entanglement witnesses.

A straightforward way to demonstrate that the bound from Theorem~\ref{thm:sep-W_M_k} can be violated is to provide an explicit construction of a state that does it. We consider an observable $W_{XYZ}^{(N,N,\mathcal{G})}$  of the form in Eq. (\ref{eqn:W_XYZ_k}) where the sum runs over a single tuple consisting of all $N$ qubits. Its expectation value in GHZ states $\vert \psi \rangle = \frac{1}{\sqrt{2}} \left( \vert 0 \rangle^{\otimes N} + \vert 1 \rangle^{\otimes N} \right)$ produces the maximum possible value of $3$ at $N = 4n$, and the minimum possible value of $-3$ at $N=4n-2$ for $n \in \mathbb{Z^{+}}$. Numerically, we observe that for $N=2n+1$ the largest and smallest eigenvalues are $\pm\sqrt{3}$.

We can similarly consider an observable of the form $W_{XZ}^{(N,N,\mathcal{G})}$ which for $N=2$ is the standard Bell observable $W_B$ up to a global factor.
Its expectation in GHZ states gives the maximum possible value of $2$ for $N=2n$, and whose expectation in (locally equivalent to GHZ) states of the form $\frac{1}{\sqrt{2}} \left( \vert 01\rangle^{\otimes N/2} - \vert 10 \rangle^{\otimes N/2} \right)$ gives the minimum possible value of $-2$ for $N=2n$ where $n \in \mathbb{Z^{+}}$. Similar statements could be constructed via local unitary operations about other observables of the form $W_{PQ}^{(N,N)}$, which are therefore seen to belong to the class of observables in Eq. \eqref{eq:W_upper_lower}, which are equivalent (up to a spectrum shift) to the entanglement witness 2.0 notion introduced in \cite{Bae_2020}. See Fig.~\ref{fig:Wxyz_witness} and Appendix \ref{secn:properties-W_XYZ_W_PQ} for more details.

\begin{figure}
    \centering
    \includegraphics[width=0.48\textwidth]{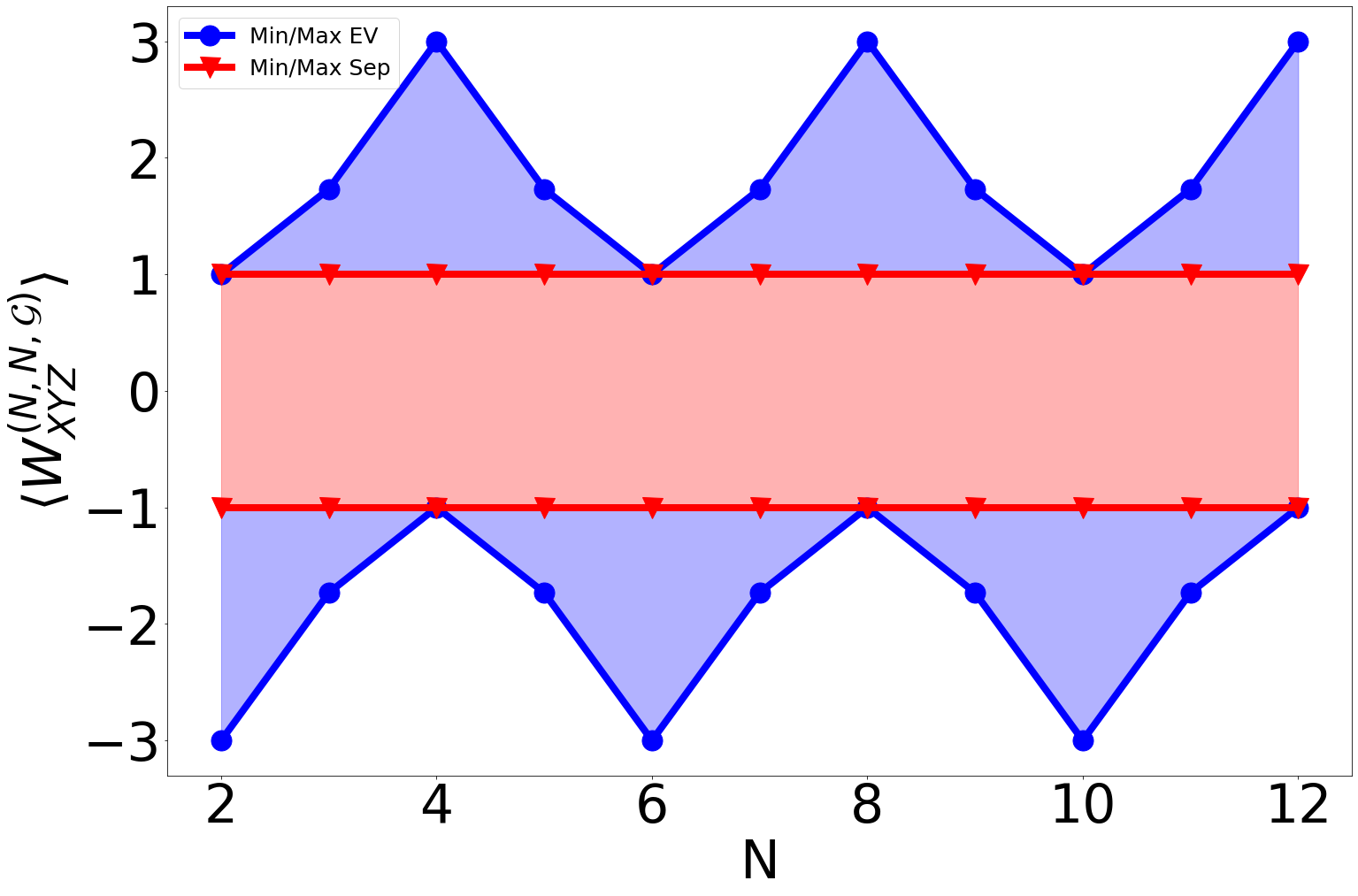}
    \includegraphics[width=0.48\textwidth]{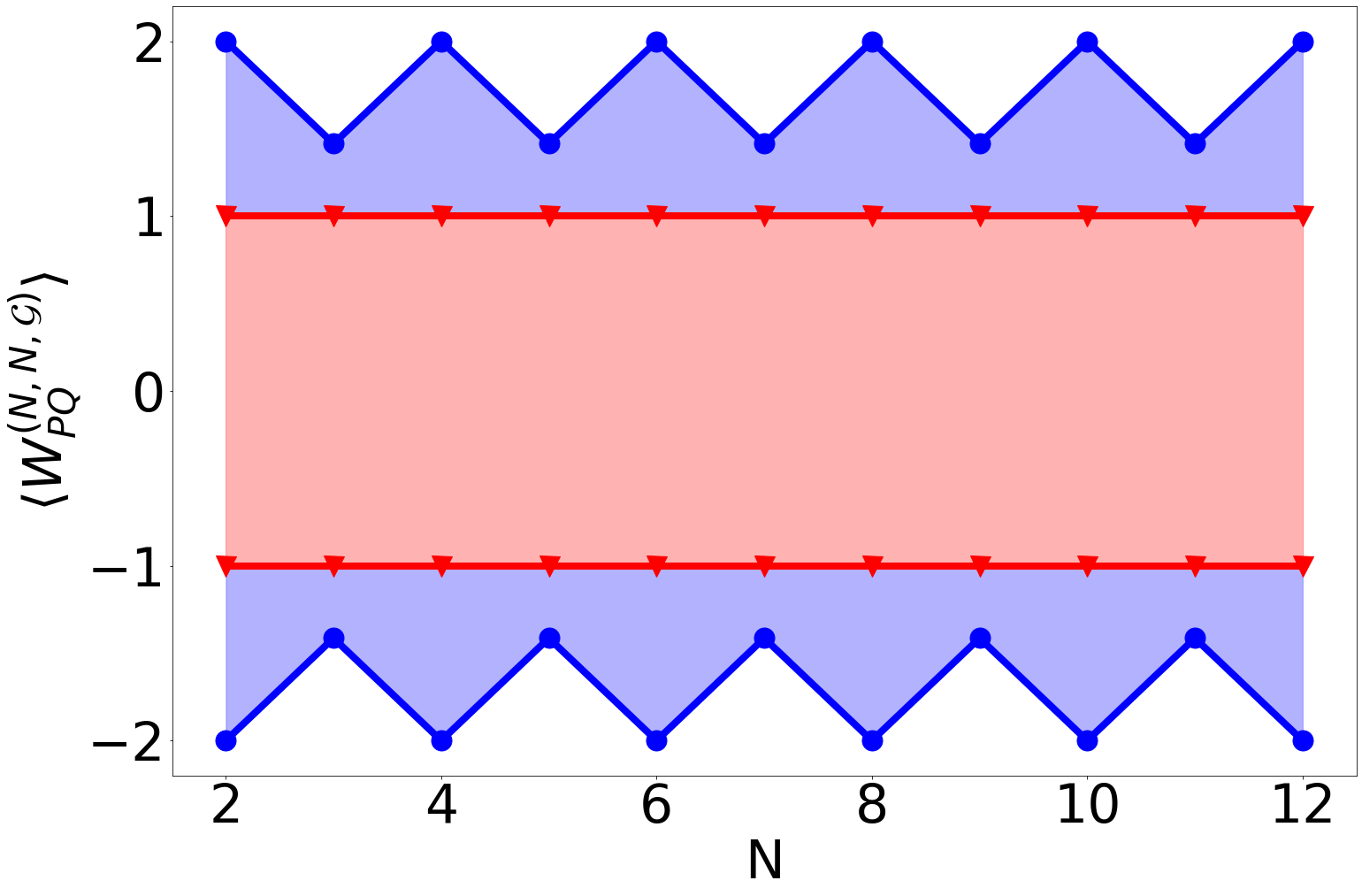}
    \caption{The maximum eigenvalue compared to the threshold for separable states for observables of the form  $W_{XYZ}^{(N,N,\mathcal{G})}$ (top) and $W_{PQ}^{(N,N,\mathcal{G})}$ (bottom), where $P \neq Q$. The blue region is where a violation of the separable threshold is possible, while the red denotes a region where these observables cannot certify, or witness, entanglement since the same expectation value could be produced by a separable state.}
    \label{fig:Wxyz_witness}
\end{figure}

In general for observables of the form $W_{PQ}^{(k,N,\mathcal{G})}$, we can demonstrate a lower bound for the maximal eigenvalue that is larger than the number of edges.
\begin{theorem}
The maximum eigenvalue $\lambda_{max}$ of any observable of the form $W_{PQ}^{(k,N,\mathcal{G})}$ (see Eq. \eqref{eqn:W_PQ_k}) is bounded below by $\lambda_{max} > \vert E_k \vert + O\left( \frac{1}{\vert E_k \vert^4} \right)$.
\label{thm:max_ev}
\end{theorem}
\begin{proof}
See Appendix \ref{secn:proofs}
\end{proof}
Theorem \ref{thm:max_ev} establishes that all observables of the form \eqref{eqn:W_PQ_k} can serve as entanglement witnesses, at least in the weak sense that they can distinguish fully separable states from (some) quantum states that have some entanglement, perhaps only over a subset of the entire $N$ qubits over which the observable is defined. Note however that this theorem does not necessarily apply to the larger family of operators consisting of two Pauli operators on each edge, $W_{2}^{(k,N,\mathcal{G})}$ (see Appendix \ref{secn:general-obs} for notational convention), with possibly different operators acting on each of the qubits in a $k$-local term, nor those consisting of three Pauli operators on each edge, $W_{XYZ}^{(k,N,\mathcal{G})}$. Indeed, we numerically observe that the largest eigenvalues of operators of the form $W_{XYZ}^{(2,N,\mathcal{G})}$ equal the number of edges, $\vert E_k \vert$, precisely the upper bound for separable states established in Theorem \ref{thm:sep-W_M_k}, and therefore such observables cannot serve as witnesses.

Nevertheless, the family of observables to which Theorem \ref{thm:max_ev} applies contains exponentially many elements, consisting of all non-trivial graph types, with different choices of Pauli operators, with each of these constituting a different entanglement witness. For a fixed $N$, the total number of such elements is $3 \sum_{k=2}^{N}\left( 2^{{N \choose k}} - 1 \right)$, where the exponent ${N \choose k}$ is the total possible number of $k$-tuples, the base 2 counts whether a given $k$-tuple/edge belongs to a (generalized) graph or not, the $-1$ subtracts the trivial graph with no edges (and which does not yield a witness), and 3 counts the possibilities $XZ$, $XY$ and $YZ$.

It is also instructive to look at the upper bound of the largest eigenvalue of these observables. Such an upper bound would tell us the extent of violation of the separable threshold provided by Theorem \ref{thm:sep-W_M_k} that is possible even in principle. The following lemma establishes a fairly generic upper bound applicable to all observables to which Theorem \ref{thm:sep-W_M_k} applies.

\begin{lemma}
The maximum eigenvalue $\lambda_{max}$ of any $N$-qubit $k$-local ($k \leq N$) observable of the form $W_{M}^{(k,N,\mathcal{G})}$ defined in Eq. \eqref{eqn:W_M_kN} (see Appendix \ref{secn:general-obs}) is bounded above by $\lambda_{max} \leq M \vert E_k \vert$.
\label{lemma:upper_bound}
\end{lemma}
\begin{proof}
Using successive applications of Weyl's inequality for the largest eigenvalue of the sum of two Hermitian matrices $A$ and $B$
\begin{eqnarray}
\lambda_{max}\left( A + B \right) &\leq& \lambda_{max} (A) + \lambda_{max} (B)
\end{eqnarray}
and the fact that $\lambda_{max}(P) = 1$ for any Pauli operator $P$, we readily obtain
\begin{eqnarray}
\lambda_{max}\left( W_{M}^{(k,N,\mathcal{G})}\right) &\leq& \sum_{\langle i_1, \dots, i_k \rangle} \sum_{m=1}^{M} \lambda_{max} \left( \alpha_m \bigotimes_{j=1}^{k}\sigma_{i_j}^{(a_{j,m})}\right) \nonumber \\
&\leq& \sum_{\langle i_1, \dots, i_k \rangle} \sum_{m=1}^{M} \lambda_{max} \left( \bigotimes_{j=1}^{k}\sigma_{i_j}^{(a_{j,m})}\right) \nonumber \\
&=& \sum_{\langle i_1, \dots, i_k \rangle} \sum_{m=1}^{M} (1) = M \vert E_k \vert
\end{eqnarray}
\end{proof}

The upper bound of $M \vert E_k \vert$ is saturated, though not uniquely and not always, by graphs with edges on disjoint pairs. This can be seen for $M=2$ by choosing the state $\vert \psi \rangle = \frac{1}{\sqrt{2}} \left( \vert 00 \rangle + \vert 11 \rangle \right)$ on each pair to give an expectation value of 2 to the observable $ZZ + XX$ defined on this pair, 2 being the maximum value of this observable, since both $\langle ZZ \rangle \leq 1$ and $\langle XX \rangle \leq 1$. For such graphs, the difference $\Delta W = \lambda_{max} - \max_{\rho_{sep}} \tr{\left(W\rho_{sep}\right)} = \vert E_2  \vert = \left\lfloor \frac{N}{2} \right\rfloor$.

Numerically, we find that this difference is the maximum value achievable by graphs of any connectivity for $W_{XZ}^{(2,N,\mathcal{G})}$.\footnote{If this is true, then we would have the sharper upper bound
$\lambda_{max} \leq \vert E \vert + \left\lfloor{\frac{N}{2}}\right\rfloor$
for any observable of the form $W_{PQ}^{(2,N,\mathcal{G})}$, since the observables $W_{XY}^{(2,N,\mathcal{G})}$ and $W_{YZ}^{(2,N,\mathcal{G})}$ have the same eigenspectrum as the corresponding observable $W_{XZ}^{(2,N,\mathcal{G})}$ once the graph is fixed.} In the next sections, we see that observables of the form $W_{XZ}^{(2,N,\mathcal{G})}$ have a special significance for QAOA-Maxcut states.

\subsection{\label{subsectn:qaoa-witness} Entanglement Witnesses for QAOA-MaxCut}
   In QAOA \cite{farhi2014quantum}, we seek to find the (approximate) ground state of diagonal Hamiltonians. For the particular case of MaxCut over a graph of equal weights, we seek to maximize the cost Hamiltonian $C^{\prime} = \frac{1}{2} \sum_{\langle i,j \rangle} \left( I - Z_i Z_j \right)$, or equivalently, minimize the cost Hamiltonian $C = \sum_{\langle i,j \rangle} Z_i Z_j$. In the basic formulation of QAOA, we start with an initial state of $H^{\otimes N} \vert 0 \rangle^{\otimes N}$, and then apply the unitary $e^{-i\beta_j B} e^{-i \gamma_j C}$ a total of $p$ times, where $B = \sum_{i=1}^{N} X_i$ and the parameters $\vec{\theta} = \left(\beta_1, \dots, \beta_p, \gamma_1, \dots, \gamma_p \right)$ are to be optimized to minimize the expectation value $\langle C \rangle$.
   
   For QAOA states built out of the 2-local MaxCut cost Hamiltonian, the 2-local observable $W_{XZ}^{(2,N,\mathcal{G})}$ defined over the same set of edges as the cost Hamiltonian is a natural candidate for an entanglement witness. Indeed, we can view this observable as being constructed directly out of the cost Hamiltonian $C$ as $W_{XZ}^{2,N,\mathcal{G}} = C + H^{\otimes N} C H^{\otimes N}$. Note that if $\rho (\vec{\theta}_p)$ represents the parametric QAOA state after $p$ rounds, and $A$ is any observable, then $\max_{\theta} \tr\left(A \, \rho(\vec{\theta}_{p+k})\right) \geq \max_{\theta} \tr \left( A \, \rho(\vec{\theta}_p) \right)$ for all $k \in \mathbb{Z}^{+}$. In other words, if $W$ is a witness for QAOA$_p$ states, then it is also a witness for QAOA$_{p+k}$ states. We now show that observables from the $W_{XZ}^{(2,N,\mathcal{G})}$ family can serve as entanglement witnesses for QAOA states, noting a lemma before we do so.

\begin{lemma}
Given $C = \sum_{\langle ij \rangle \in E_2} Z_i Z_j$ where $E_2$ is some edge set, and $B = \sum_i X_i$, the expectation values of the operators $X_u X_v$, $Y_u Y_v$ and $Z_u Z_v$ over some edge $\langle uv \rangle \in E_2$ in the $p=1$ QAOA-MaxCut state $\vert \psi(\gamma, \beta) \rangle = e^{-i\beta B} e^{-i \gamma C} H^{\otimes N} \vert 0 \rangle^{\otimes N}$ are given respectively as
\begin{eqnarray}
\langle X_{u}X_{v}\rangle &=& \frac{1}{2}  \cos^{d_u + d_v - 2f}{2\gamma} \left( 1 + \cos^{f}{4\gamma}\right) \nonumber \\
\langle Y_{u}Y_{v}\rangle &=& - \frac{1}{2}\sin{4\beta}\sin{2\gamma} \left( \cos^{d_u}{2\gamma} + \cos^{d_v}{2\gamma} \right) \nonumber \\
&& + \frac{1}{2} \cos^{2}{2\beta} \cos^{d_u + d_v - 2f}{2\gamma} \left( 1 - \cos^{f}{4\gamma}\right) \nonumber \\
\langle Z_{u}Z_{v} \rangle &=& \frac{1}{2} \sin{4\beta} \sin{2\gamma} \left( \cos^{d_u}{2\gamma} + \cos^{d_v}{2\gamma} \right) \nonumber \\
&& + \frac{1}{2} \sin^{2}{2\beta} \cos^{d_u + d_v - 2f}{2\gamma} \left( 1 - \cos^{f}{4\gamma} \right) \nonumber \\
\end{eqnarray}
where $d_u$ ($d_v$) is the number of neighbors of $u$ ($v$) excluding $v$ ($u$), and $f$ is the number of triangles in the graph that include the edge $\langle uv \rangle$.
\label{lemma:xx-zz-expects}
\end{lemma}
\begin{proof}
See Appendix \ref{secn:proofs}.
\end{proof}
Note that for triangle free graphs $f=0$, we have that $\langle Z_u Z_v \rangle (\beta, \gamma) = \langle Y_u Y_v \rangle (-\beta, \gamma) = \langle Y_u Y_v \rangle (\beta, -\gamma)$, so that if an observable of the form $W_{XZ}^{(2,N,\mathcal{G})}$ is a witness for some QAOA state, then so is $W_{XY}^{(2,N,\mathcal{G})}$ defined over the same graph.
We can now show that the constructed observables can serve as entanglement witnesses for QAOA states prepared through MaxCut cost Hamiltonians on different graph topologies $\mathcal{G}$.

\subsubsection*{Ring graph}

\begin{theorem}
Given a ring Hamiltonian $C = \sum_i Z_i Z_{i+1}$ and the standard mixer $B = \sum_i X_i$, the observable defined on the ring graph $\mathcal{G}_\mathrm{ring}$
\begin{eqnarray}
W \equiv W_{XZ}^{(2,N,\mathcal{G}_\mathrm{ring})} = \sum_i X_i X_{i+1} + Z_i Z_{i+1}
\end{eqnarray}
serves as an entanglement witness for QAOA states $\Pi_{k=1}^{p} \left( e^{-i \beta_k B} e^{-i \gamma_k C} \right) \vert + \rangle^{\otimes N}$.

In particular, the gap $\Delta W = \max \langle W \rangle_{QAOA} - \max \langle W \rangle_{sep}$ between the maximum expectation value for $W$ achievable by QAOA states and that achievable by separable states is lower bounded as
\begin{eqnarray}
\Delta W \geq N \frac{\sqrt{2} - 1}{2}
\end{eqnarray}
\end{theorem}
\begin{proof}
For ring graphs of $N$ nodes, we have $d_u = d_v = 1$, $f=0$ for each of $N$ edges. Therefore, the observable has expectation value $\langle W \rangle = N\, W_{e}$, where $W_e$ denotes the expectation value $\langle XX + ZZ \rangle$ over a single edge in the ring graph. From Lemma \ref{lemma:xx-zz-expects}, we have
\begin{eqnarray}\label{eq:ring_expect}
W_e &=& \cos^2{2\gamma} + \frac{\sin{4\gamma} \sin{4\beta}}{2}
\end{eqnarray}
At all local optima, we have $\partial_{\gamma} \langle W_e \rangle = \partial_{\beta} \langle W_e \rangle = 0$. Simultaneously solving these equalities gives us either $\gamma_a = \frac{n \pi}{4}$ or $\gamma_b = \frac{\pi}{8} \left( n - \frac{1}{2}\right)$, and correspondingly either $\sin{4\beta_a} = 0$ or $\sin{4\beta_b} = \pm 1$ ($+1$ for $n$ odd, and $-1$ for $n$ even) respectively.

Let $D(\gamma, \beta) = \partial_{\gamma\gamma} W_e \cdot \partial_{\beta\beta} W_e - \left( \partial_{\beta\gamma} W_e\right)^2$. Then, $D(\gamma_a, \beta_a) < 0$, indicating a saddle point. For $n=3,7,11,\dots$ or $n=2,6,10,\dots$, we have $D(\gamma_b, \beta_b) > 0$ and $\partial_{\gamma\gamma}W_e > 0$, $\partial_{\beta\beta}W_e > 0$ indicating a local minimum. For $n=1,5,9,\dots$ or $n=0,4,8,\dots$, we have $D(\gamma_b, \beta_b) > 0$ and $\partial_{\gamma\gamma}W_e < 0$, $\partial_{\beta\beta}W_e < 0$ indicating a local maximum. At all these local maxima, $W_e(\gamma_b, \beta_b) = \frac{1 + \sqrt{2}}{2}$. Since this is true of all local maxima, this also provides the global maximum.

For separable states, the maximum achievable value for $W_e$ is one. Since all edges contribute the same expectation value in a ring graph, the gap $\Delta W$ at $p=1$ is $N \frac{\sqrt{2} - 1}{2}$. For larger $p$ values, this gap can increase in value and the theorem follows.
\end{proof}

\subsubsection*{Regular triangle-free graphs}

More generally, we can consider regular triangle-free graphs $\mathcal{G}_{\Delta\mathrm{free}}$. This family includes the ring graph as a special case, but also other types such as bipartite graphs. For all such graphs, we have $d_u = d_v = d$ and $f=0$, so that we have $\langle W_{XZ}^{(2,N,\mathcal{G}_{\Delta\mathrm{free}})}\rangle = N_{edges} W_e$ where
\begin{eqnarray}
W_e &=& \cos^{2d}{2\gamma} + \sin{4\beta} \sin{2\gamma} \cos^{d}{2\gamma}
\label{eqn:regular-triangle-free}
\end{eqnarray}
This expression can be optimized over $\gamma$ and $\beta$ to find its maximum value. We find numerically that $max_{\gamma, \beta}  W_e > 1$ and therefore serves as a witness for arbitrary large $d$, as seen in Fig. \ref{fig:regular_triangle_free_qaoa}. However, it is greatest at $d=1$ and smoothly decays to 1 for large values of $d$. This asymptotic decay can be seen by considering Eq. (\ref{eqn:regular-triangle-free}) in the limit $d \rightarrow \infty$. In this limit, 
\begin{eqnarray}
\lim_{d \rightarrow \infty} \cos^{d}{2\gamma} &\rightarrow& \left\lbrace \begin{array}{cc}
+1, & \gamma = n\pi, \\
(-1)^{d}, & \gamma = \left( n + \frac{1}{2}\right) \pi,\\
0, & \text{otherwise}
\end{array}\right.
\label{eqn:large-cos-behavior}
\end{eqnarray}
while $\sin{2\gamma} \rightarrow 0$ at $\gamma = n\pi$ and $\gamma = \left( n + \frac{1}{2}\right) \pi$, for $n \in \mathbb{Z}$. Therefore, one term will asymptotically converges to 0 while the other converges to 1, which is the value saturated by separable states.

\begin{figure}
    \centering
    \includegraphics[width=0.48\textwidth]{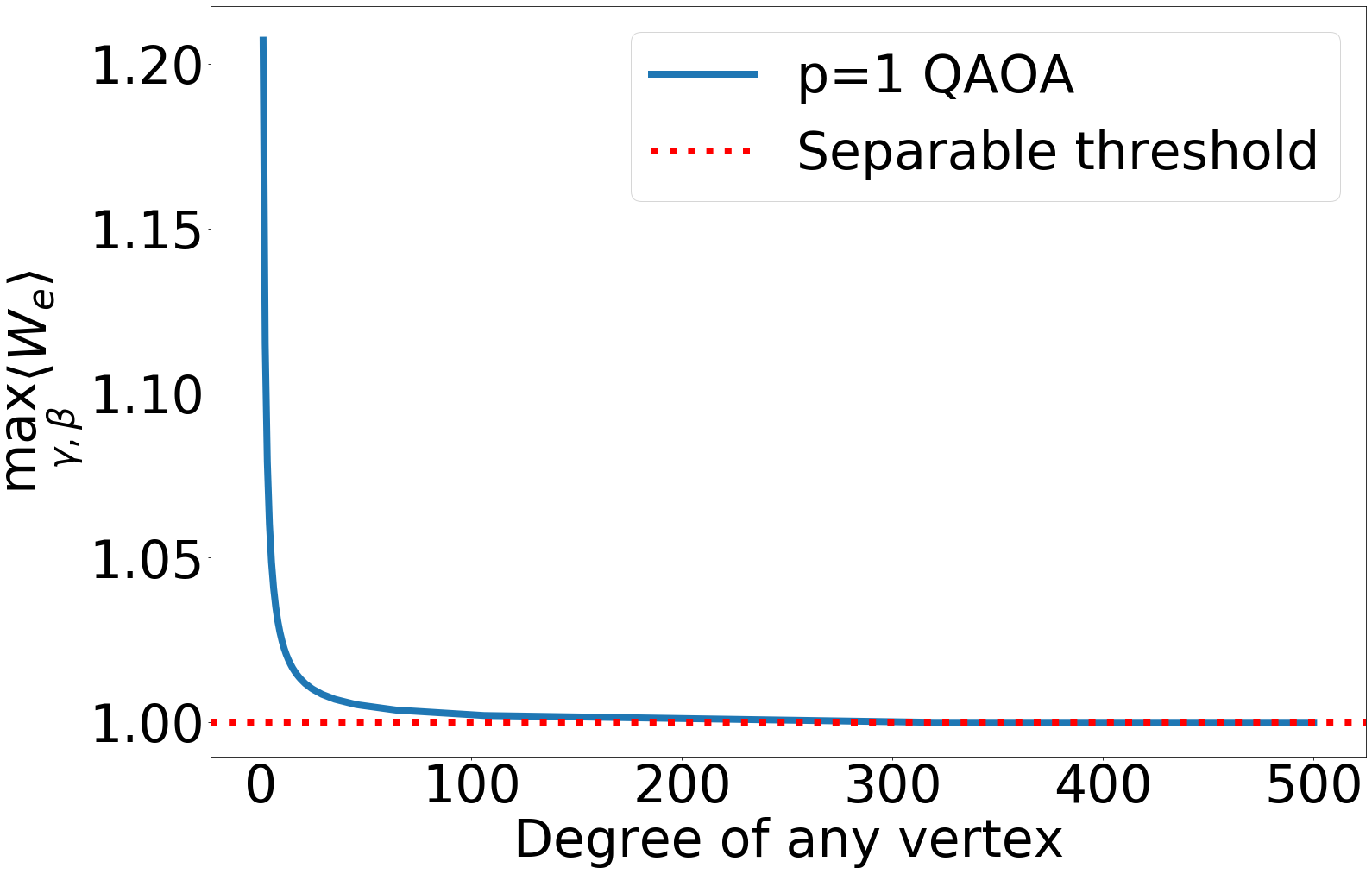}
    \caption{The maximum achievable expectation value, divided by the number of edges, of observables of the form $W_{XZ}^{(2,N,\mathcal{G}_{\Delta-\text{free}})}$ for regular, triangle-free graphs in $p=1$ QAOA. At asymptotically large values of the degree of any vertex in the regular graph, this observable fails to serve as an entanglement witness for $p=1$ QAOA-MaxCut states.}
    \label{fig:regular_triangle_free_qaoa}
\end{figure}

\subsubsection*{Linear chain graph}

A related graph is that of a linear chain, $\mathcal{G}_{\mathrm{line}}$ which unlike the ring, has open boundary conditions. This is the graph we use in our experiments (see sub-section \ref{sec:results}). Here, we can use the fact that $d_u=0,\ d_v=1$ (or vice versa) for the two edges at either end, $d_u=d_v=1$ for all other edges, $f=0$ for every edge, and the fact that there are $N-1$ total edges, where $N$ is the number of nodes in the linear chain graph. After some algebraic simplification, the expectation value of $W_{XZ}^{(2,N,\mathcal{G}_\mathrm{line})}$ becomes for $N \geq 3$
\begin{eqnarray}
\langle W_{XZ}^{(2,N,\mathcal{G}_\mathrm{line})} \rangle &=& \frac{\sin{(4\beta)}}{2} \left[ (N-2) \sin{(4\gamma)} + 2\sin{(2\gamma)} \right] \nonumber \\
&& +\, (N-3) \cos^2{(2\gamma)} + 2\cos{(2\gamma)}
\label{eqn:W_XZ_2_linear_chain}
\end{eqnarray}
while for $N=2$, it is $\langle W_{XZ}^{(2,N,\mathcal{G}_{\mathrm{line}})} \rangle = 1 + \sin{4\beta} \sin{2\gamma}$. It is found numerically that the quantity $\langle W_\mathrm{line} \rangle_{max} \equiv max_{\beta,\gamma} \langle W_{XZ}^{(2,N,\mathcal{G}_{\mathrm{line}})} \rangle$ grows linearly with the number of qubits for $N \geq 3$, roughly as
\begin{eqnarray}
\langle W_\mathrm{line} \rangle_{max} &\approx& 1.207 N - 1.019 \nonumber \\
&>& N - 1
\label{eqn:Wmax-qaoa-linear-chain}
\end{eqnarray}
while for $N=2$, $\langle W_{line} \rangle_{max} = 2$. Therefore, this observable serves as an entanglement witness for $p=1$ QAOA.

\subsubsection*{Fully connected graph}

Lastly, we may also consider the fully connected graph $\mathcal{G}_{\mathrm{full}}$. In this case, we use the fact that $d_u = d_v = f = N-2$ for every edge, and the fact that there are a total of $\vert E_2 \vert = {N \choose 2}$ edges. The expectation value of $W_\mathrm{full} \equiv W_{XZ}^{(2,N,\mathcal{G}_\mathrm{full})}$ becomes
\begin{eqnarray}
\langle W_\mathrm{full} \rangle &=& {N \choose 2} \left( \frac{1 + \cos^{N-2}{4\gamma}}{2} + \sin{4\beta} \sin{2\gamma} \cos^{N-2}{2\gamma}\right. \nonumber \\
&& \left. \qquad\qquad + \frac{\sin^{2}{2\beta} (1 - \cos^{N-2}{4\gamma})}{2} \right)
\end{eqnarray}
Let $W_{e}(\gamma, \beta, N) \equiv \langle W_\mathrm{full} \rangle / \vert E_2 \vert$ denote the term in the paranthesis above. Requiring that $W_{e} > 1$ and that therefore $W_\mathrm{full}$ serves as an entanglement witness for the $p=1$ QAOA state is equivalent to requiring that
\begin{eqnarray}
1 - \cos^{N-2}{4\gamma} &<& 2 \tan{2\beta} \sin{2\gamma} \cos^{N-2}{2\gamma}
\end{eqnarray}
Numerically, we find that this condition is satisfied by some $(\gamma, \beta)$ (e.g. by $\gamma = \pi/499$, $\beta=122\pi/499$) for up to $N=5000$. In the large $N$ limit however, the observable $W_{XZ}^{(2,N,\mathcal{G}_\mathrm{full})}$ asymptotically fails to be a witness as the maximum value it attains is no more than that achieved by separable states. We can see this by noting the asymptotic behavior of $\cos^{n}{x}$ in the limit of large $n$ as in Eq. \eqref{eqn:large-cos-behavior}, and the fact that $\sin{x} = 0$ whenever $\cos{x}=1$ to find
\begin{equation}
    \lim_{N\rightarrow\infty} \max_{\gamma,\beta} W_e (\gamma, \beta,N) \rightarrow 1
\end{equation}
which is precisely the bound for separable states.

\bigskip

In addition to the types of graphs analyzed above, we numerically checked that the observables $W_{XZ}^{(2,5,\mathcal{G})}$ serve as entanglement witnesses for $p=1$, $N=5$ QAOA-MaxCut states defined over any of the $2^{{5 \choose 2}} - 1$ non-trivial graphs, where the observable is defined over the same set of edges as the cost Hamiltonian. We conjecture that this holds true for any finite $N$.

The asymptotic decay we observe at large $N$ for complete graphs and large $d$ for regular triangle-free graphs for the (normalized by the number of edges) maximum expectation values of $W_{XZ}^{(2,N,\mathcal{G})}$ in $p=1$ QAOA states may reflect the fact that it becomes increasingly difficult to solve the corresponding MaxCut problem for larger values of $N$ and $d$ respectively at $p=1$, and that a larger number of QAOA rounds ($p$) may be required to solve it. It is generally difficult to study the performance of QAOA at large values of $p$, but for reasonably small values of fixed $N$ and $d$, we numerically find that $\max_{\gamma,\beta} \langle W_{XZ}^{(2,N,\mathcal{G})} \rangle$ saturates around the same $p$ as when the cost Hamiltonian saturates.

\subsection{XYZ witnesses}

In addition to the $W_{XZ}^{(2,N,\mathcal{G})}$ family of observables, we may also consider observables of the form $W_{XYZ}^{(N,N,\mathcal{G})}$. In this case, we can calculate the expectation values of $\braket{\Xn}$, $\braket{\Yn}$, $\braket{\Zn}$ rather easily. We note that QAOA-MaxCut states are  $\mathbb{Z}_2$ symmetric (see for example \cite{Bravyi_2020,shaydulin2020classical} for additional reference), and we can express all such states as 
\begin{eqnarray}
\ket{\Psi} &=& \sum_{x \in \mathcal{X}_0} c_x (\ket{x}+\ket{\bar{x}})
\end{eqnarray}
Where $\mathcal{X}_0$ is the set of bitstrings that start with $0$, and $\bar{x}$ represents the ones' complement of $x$. We then have
\begin{eqnarray}
\Xn \ket{\Psi} &=& \sum_{x \in \mathcal{X}_0} c_x \Xn(\ket{x}+\ket{\bar{x}}) \nonumber \\
&=& \sum_{x \in \mathcal{X}_0} c_x \Xn(\ket{\bar{x}}+\ket{x}) = \ket{\Psi}
\end{eqnarray}
so that $\braket{\Xn} = \braket{\Psi|\Psi}$ = 1. Moreover,
\begin{equation}
\Zn \ket{\Psi} = \sum_{x \in \mathcal{X}_0} c_x ((-1)^{|x|}\ket{x}+(-1)^{|\bar{x}|}\ket{\bar{x}}) 
\end{equation}
where $|x|$ denotes the hamming weight of $x$.
Note that for even $N$, $(-1)^{|x|}=(-1)^{|\bar{x}|}$ and for odd $N$, $(-1)^{|x|}=(-1)^{1 + |\bar{x}|}$. Therefore,
\begin{eqnarray}
\braket{\Zn} &=& \left\lbrace \begin{array}{cc}
0 & \text{odd N}, \\
\sum_{x \in \mathcal{X}_0} 2 (-1)^{|x|} |c_x|^2 & \text{even N}
\end{array}\right.
\end{eqnarray}
Furthermore,
\begin{eqnarray}
\braket{\Yn} &=& \braket{(-i)^N \Zn \Xn} \nonumber \\
&=& (-i)^N \braket{\Zn}
\end{eqnarray}
for even $N$.
Based on the above observations, we see that necessary conditions for a violation of the separable threshold of 1 are $N=2n$ for $W_{XZ}^{(N,N,\mathcal{G})}$ and $W_{XY}^{(N,N,\mathcal{G})}$, and $N=4n$ for $W_{YZ}^{(N,N,\mathcal{G})}$ and $W_{XYZ}^{(N,N,\mathcal{G})}$, where $n \in \mathbb{Z}^{+}$. However, these constraints alone do not guarantee that the separable threshold is violated, which may not always be possible for certain types of graphs. However, we do observe numerically for a few small allowed values of $N$ that such violations are indeed possible at least for some graphs.

\subsection{\label{subsecn:separability}Separability}
Having established an upper bound for the expectation value of a large family of observables in fully separable states, it is a natural to ask next whether the type of entanglement detected by certain witnesses is genuine $N$-partite entanglement, or could be achieved with entanglement over only a subset of those $N$ qubits. In other words, one can look at the $k$-separability properties of the observables described above. An $N$-partite pure quantum state $\vert \Psi_{k_{sep}} \rangle$ is $k$-separable iff it can be expressed in the form
\begin{eqnarray}
\vert \Psi_{k_{sep}} \rangle &=& \vert \Psi_1 \rangle \otimes \vert \Psi_2 \rangle \otimes \dots \otimes \vert \Psi_k \rangle
\end{eqnarray}
and likewise, a mixed state $\rho_{k_{sep}}$ is $k$-separable iff it can be expressed in the form
\begin{eqnarray}
\rho_{k_{sep}} &=& \sum_{i} p_i \vert \Psi_{k_{sep}}^{(i)} \rangle \langle \Psi_{k_{sep}}^{(i)} \vert
\end{eqnarray}
We know from theorem \ref{thm:sep-W_M_k} that the maximum expectation value achievable by a fully seperable state for any observable of the form $W_{M}^{(k,N,\mathcal{G})}$ is $\vert E_k \vert$. On the other hand, we have shown that this bound is violated by $N$-qubit GHZ states for observables of the form $W_{XYZ}^{(N,N,\mathcal{G})}$ whenever $N=4n$ and for those of the form $W_{XZ}^{(N,N,\mathcal{G})}$ whenever $N=2n$ for some $n \in \mathbb{Z}^{+}$. However, such violations do not necessarily certify genuine $N$-partite entanglement.

To illustrate, consider an $M$-separable product state of the form $\vert \Phi_{n}^{M} \rangle = \otimes_{m=1}^{M} \vert \psi_{n}^{(m)} \rangle$, where $\vert \psi_n \rangle = \frac{1}{\sqrt{2}}\left( \vert 0 \rangle^{\otimes n} + \vert 1 \rangle^{\otimes n} \right)$ is the $n$-qubit GHZ state. Then, the maximal expectation value of 3 for the observable $W_{XYZ}^{(N,N,\mathcal{G})}$ is achieved by both an $N$-qubit GHZ state $\vert \psi_{N} \rangle$, but also any $M$-separable state $\vert \Phi_{n}^{M} \rangle$ whenever $N=4nM$. Likewise, $\langle \psi_{N} \vert W_{XZ}^{(N,N,\mathcal{G})} \vert \psi_{N} \rangle = \langle \Phi_{n}^{M} \vert W_{XZ}^{(N,N,\mathcal{G})} \vert \Phi_{n}^{M} \rangle = 2$, the maximal value, whenever $N=2nM$.

For $W_{XZ}^{(2,N,\mathcal{G})}$, the upper bound in Lemma \ref{lemma:upper_bound} implies that if the edge set is given by $E_2 = \cup_{j=1}^{3} E_2^{(j)}$ where $E_2^{(1)}$ and $E_2^{(2)}$ are respectively the edges contained within each halves of some partition of the graph $\mathcal{G}$, and $E_2^{(3)}$ the edges that cross that partition, then any bi-separable state could achieve an expectation value of at most $2 (\vert E_2^{(1)} \vert + \vert E_2^{(2)} \vert) + \vert E_2^{(3)}\vert$, which can certainly be exceeded if the upper bound of $2(\vert E_2^{(1)} \vert + \vert E_2^{(2)} \vert + \vert E_2^{(3)} \vert)$ is saturated, in which case the observable would certify genuine $N$-partite entanglement\footnote{We could similarly extend this argument to $W_{M}^{(2,N,\mathcal{G})}$ (see Appendix \ref{secn:general-obs}), and also use the same reasoning to $W_{M}^{(k,N,\mathcal{G})}$ to establish the upper bound of $M \left(\sum_{j=1}^{k} \vert E_k^{(j)} \vert\right) + \vert E_k^{(k+1)} \vert$ for $k$-separable states, where the edge sets $E_k^{(1)}, \dots, E_k^{(k)}$ are restricted to each of $k$-partitions, and $E_{k}^{(k+1)}$ consists of $k$-tuples with one index in each of the $k$ partitions.}.
However, such a saturation of the upper bound is not guaranteed. In principle, the generalized Cauchy-Schwarz inequality (Lemma \ref{thm:cauchy-schwarz}) could be employed in conjunction with the inequalities proved earlier, as well as the physicality requirement $\sum_{\alpha=1}^{4^k} \langle P_{\alpha} \rangle^2 \leq 2^k$, where the sum runs over all $k$-qubit Pauli operators in any $k$-qubit state $\rho$, to obtain some upper bounds on the expectation value of such observables in $k$-separable states. However, in practice we can find tighter bounds using numerical techniques. For each $k \leq N$, we compute the largest expectation value achievable by any $k$-separable pure state, setting an upper bound to the largest expectation value achievable by any $k$-separable mixed state.

Concretely, we adopt a trigonometric parametrization of a classical probability vector $\vec{p} = (p_1, \dots, p_{2^N})$ as
\begin{eqnarray}
p_i &=& \sin^{2}{\theta_{i-1}} \prod_{j=i}^{N-1} \cos^{2}{\theta_j}
\end{eqnarray}
with $\theta_0 = \pi/2$, which naturally enforces the normalization $\sum_i p_i = 1$. A pure $N$-qubit quantum state can then be parametrized in terms of this probability vector and $2^N - 1$ relative phases
\begin{eqnarray}
\vert \psi \rangle_i &=& \sqrt{p_i} e^{i\theta_i}
\end{eqnarray}
where $\theta_0 = 0$.

As an example, numerically optimizing for the expectation value of $W_{XZ}^{(2,5,\mathcal{G})}$ in various $k$-separable states ($k \leq 5$), we observe in Fig. \ref{fig:W_XZ_2_sep} that the witness $W_{XZ}^{(2,5,\mathcal{G})}$ serves to certify genuine multi-partite entanglement in the linear chain QAOA-MaxCut state over $N=5$ qubits at $p\geq 2$, while at $p=1$ the witness can certify that the produced QAOA state is not 3-(or higher) separable, but cannot certify that it is not 2-separable. On the other hand, we also observe in Fig. \ref{fig:W_XZ_2_sep} that the corresponding witness $W_{XZ}^{(2,5,\mathcal{G})}$ serves to certify genuine multi-partite entanglement in the complete (all-to-all connected) graph QAOA-MaxCut state over $N=5$ qubits at all $p$ values.

In both these examples, there is a strict hierarchy of the maximum achievable expectation value $\langle W_{XZ}^{(2,N,\mathcal{G})} \rangle$ according to the separability of the class of states in which that expectation is computed, i.e. $\max_{\rho_{\text{$k$-sep}}} \langle W \rangle > \max_{\rho_{\text{$k^{\prime}$-sep}}} \langle W \rangle$ whenever $k < k^{\prime}$. This is not always the case, and the existence of such a strictly ordered hierarchy depends, at least in the case of the $W_{XZ}^{(2,N,\mathcal{G})}$, as well as their cousins in the $W_{PQ}^{(2,N,\mathcal{G})}$ family, on the structure of the graph whose edges are being summed over. In particular, if the graph $\mathcal{G}$ in $W_{PQ}^{(2,N,\mathcal{G})}$ is only non-trivially defined on $m \leq N$ qubits, then the maximum expectation value is the same for all $k$-separable states whenever $k \leq N-m+1$.

\begin{figure}
    \centering
    \includegraphics[width=0.48\textwidth]{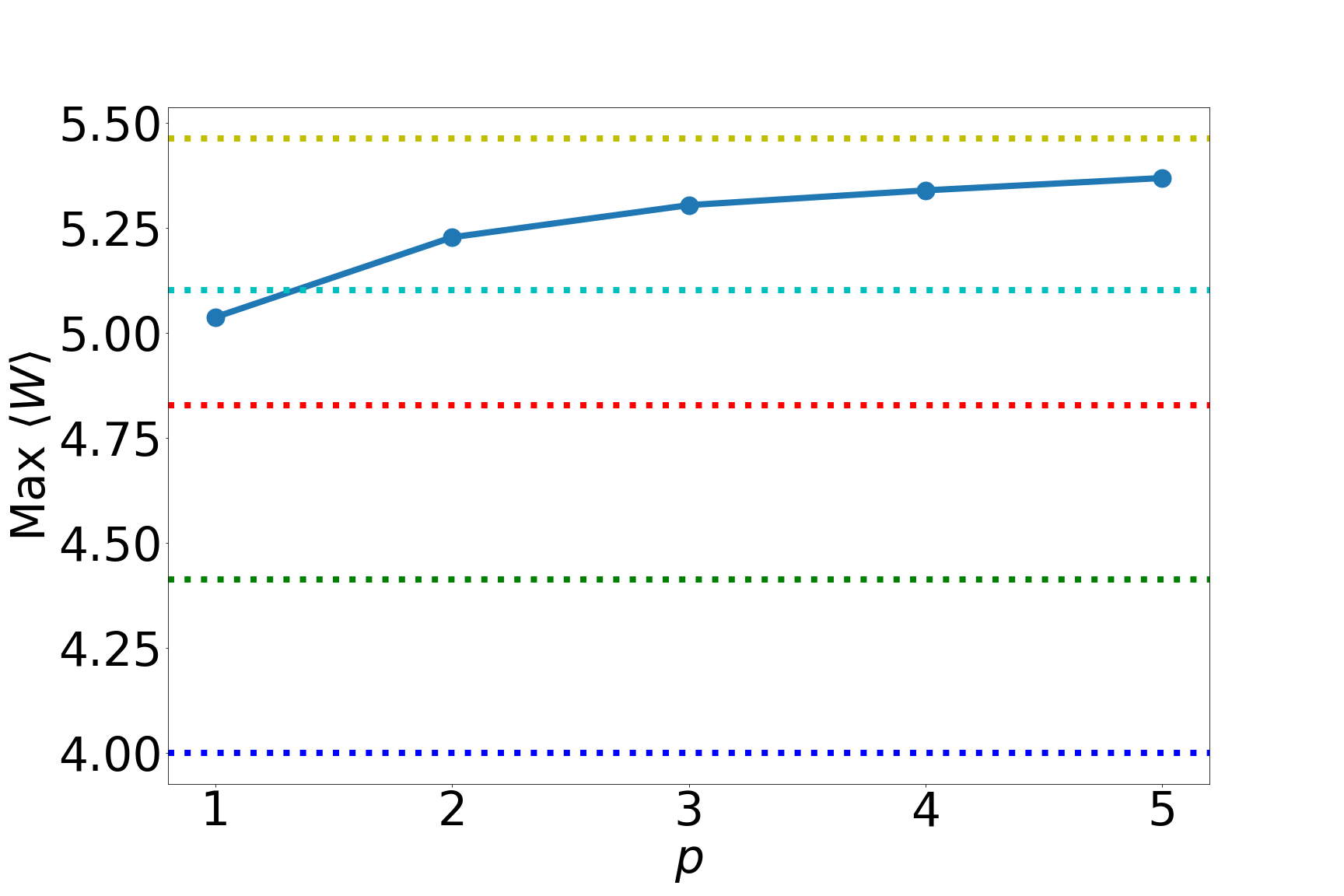}
    \includegraphics[width=0.48\textwidth]{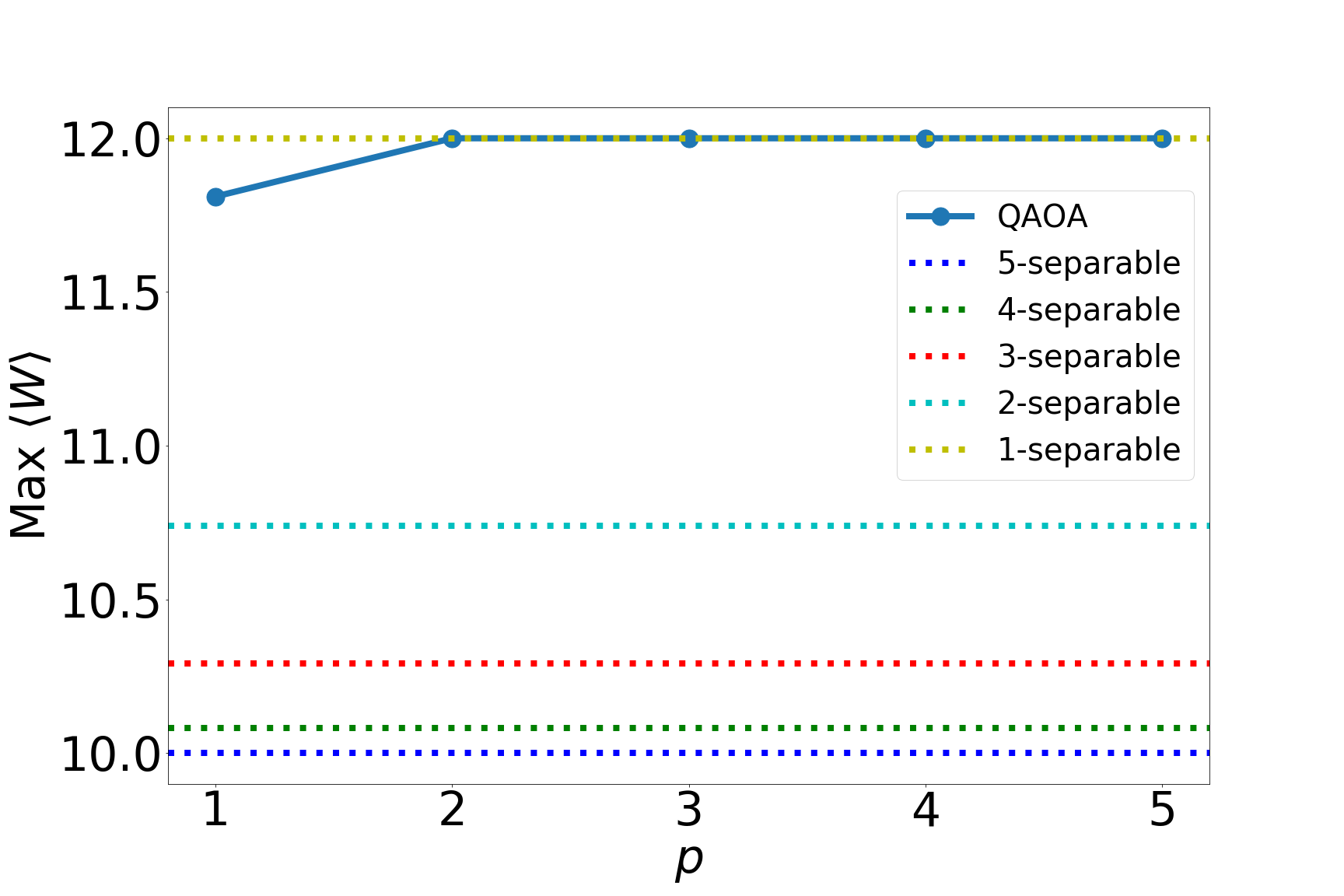}
    \caption{Separability properties of the entanglement witness $W=W_{XZ}^{(2,5,\mathcal{G})}$ for a linear chain graph (top) and complete graph (bottom) MaxCut problem. We numerically compute and report the maximum achievable expectation value of this observable in $k$-separable states for various values of $k\leq N = 5$, and compare it to the maximum expectaion value that QAOA states can achieve for various value of $p$.}
    \label{fig:W_XZ_2_sep}
\end{figure}

\subsection{Entanglement Potency}
Since the entanglement structure of multi-qubit states is very rich \cite{HorodeckiQuantumEntanglement, bengtsson2017geometry}, it is relevant to ask how many entangled states can the discussed Bell-type observables detect? In order to answer this question, we introduce a metric on witnesses, called {\it entanglement potency}, and compute its value on the space of QAOA-MaxCut states (according to a particular distribution) as well as Haar random states. We see that these observables have non-negligible entanglement potency for QAOA-MaxCut states, but close to zero potency for Haar random states, making them a suitable choice for detecting entanglement generated by this particular ansatz.
\begin{definition}
The {\it entanglement potency} of an entanglement witness $W$ with respect to some measurable set of states $\xi$ is given by
\begin{equation}
    P_\xi(W) = \frac{V^\xi_d}{V^\xi_T},
\end{equation}
where $V^\xi_T$ is the total volume of quantum states in the set $\xi$, and $V^\xi_d$ is the volume of such states whose entanglement is detected by $W$.
\end{definition}
The set $\xi$ could be a parametric family of states, such as QAOA or VQE, or it could also be the set of all $N$-qubit states (e.g. with Hurwitz parametrization \cite{Hurwitz1897, Zyczkowski_2001}). Since the set is measurable, the definition allows us the freedom to draw from arbitrary distributions over the chosen set of states. In practice, we would draw states according to some fixed distribution over some collection, and measure the fraction of drawn states that violate the separable threshold for the observable $W$. This provides an unbiased estimate of the entanglement potency defined above, and approaches its ideal value asymptotically as the number of samples becomes large.

One may expect that the potency of any $N$-qubit observable in Haar random states decreases with $N$ and perhaps vanishes at large $N$, since each observable can detect only a small subset of entangled states, while the volume of the entire set of states grows exponentially. This is similar to the vanishingly small volume of separable states at large $N$ \cite{Zyczkowski_1998}. We numerically observe that for $W_{XZ}^{(2,N,\mathcal{G})}$ on fully connected and ring graphs, the potency is negligible while for $W_{XZ}^{(N,N,\mathcal{G})}$ it decreases with $N$ (see Table~\ref{table:Z2_potency}). For QAOA-MaxCut, we observe a non-negligible amount of entangled states detected by the investigated observables (see Fig.~\ref{fig:MaxCut_potency},~\ref{fig:Z2_potency} and Table~\ref{table:Z2_potency}).


Applying this metric to QAOA-MaxCut, we consider states prepared via the ring graph Hamiltonian together with observables of the type $W_{XZ}^{(2,N,\mathcal{G}_\mathrm{ring})}$ defined on the same ring graph. The potency is then given by the volume of $(\gamma,\beta)$ for which the normalized expectation value of this witness given in Eq. \eqref{eq:ring_expect} is greater than one. This can be expressed as the following integral
\begin{eqnarray}
    && P_{\mathrm{QAOA}_1}(W_{XZ}^{(2,N,\mathcal{G}_\mathrm{ring})}) \nonumber \\
    &=& \frac{1}{(2\pi)^2} \int_0^{2\pi}d\beta \int_0^{2\pi}d\gamma \nonumber \\
    && \qquad \Theta\left(\cos^2(2\gamma) +\frac{1}{2}\sin(4\gamma)\sin(4\beta) - 1 \right)
\label{eq:heaviside_potency}
\end{eqnarray}
where $\Theta$ is the Heaviside step function. This integral evaluates to 1 only for angles that violate the separable bound, and is normalized by the total volume of $(2\pi)^2$ (i.e. the area of $[0,2\pi)\times[0,2\pi)$ of $\beta, \gamma$ intervals).

Although we could try to evaluate this integral numerically, in practice we use the Monte Carlo (MC) method to sample QAOA-MaxCut states with uniformly random values of $\gamma$ and $\beta$ from $[0,2\pi)$ for $p=1$ and $p=5$ and measure the fraction of states that violate the separable threshold for the witness, i.e. satisfy $\langle W_{XZ}^{(2,N,\mathcal{G}_\mathrm{ring})}\rangle > N$.
These results are depicted in the bottom figure of Fig.~\ref{fig:MaxCut_potency}. 
\begin{figure}
    \centering
    \includegraphics[width = 0.48\textwidth]{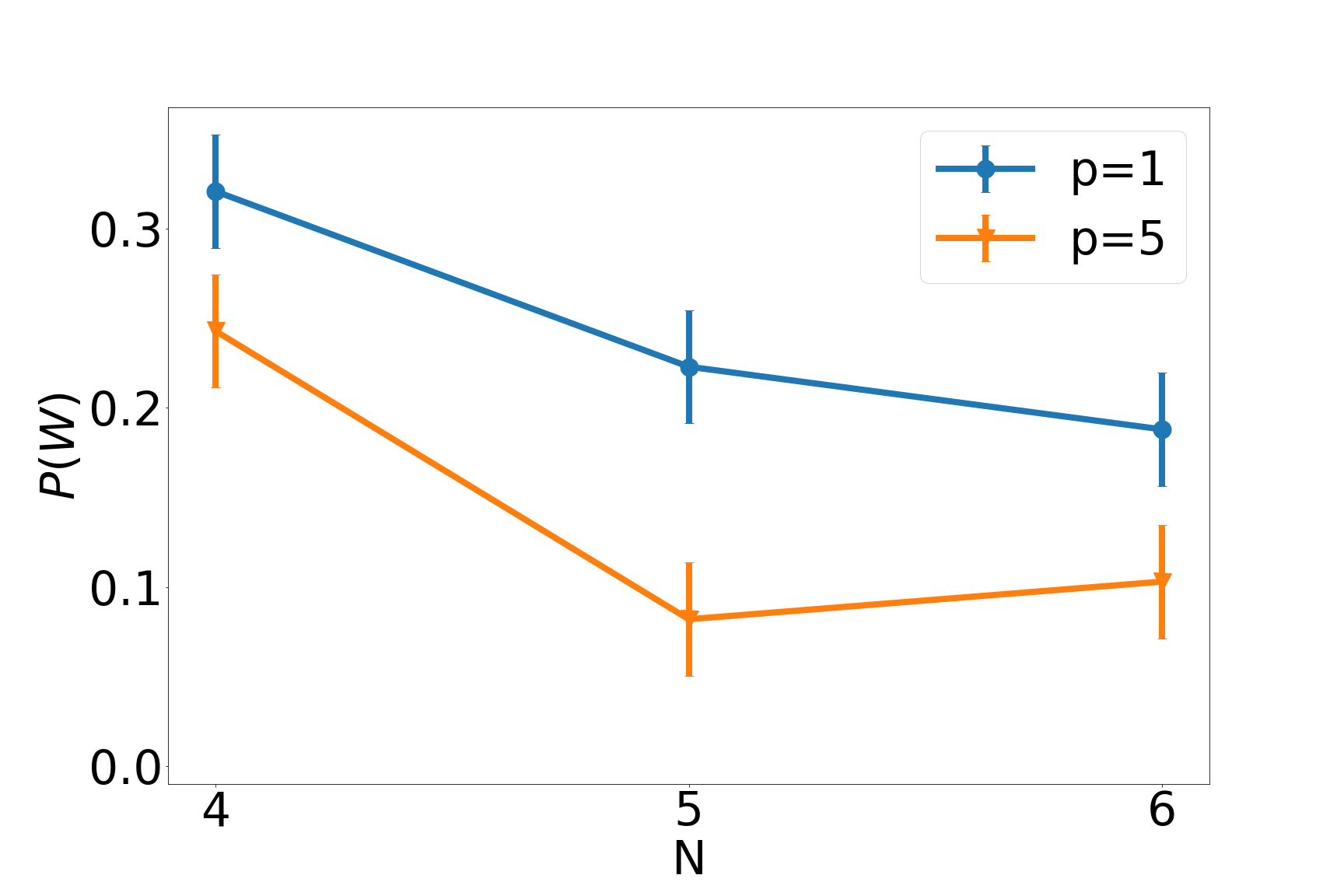}
    \includegraphics[width = 0.48\textwidth]{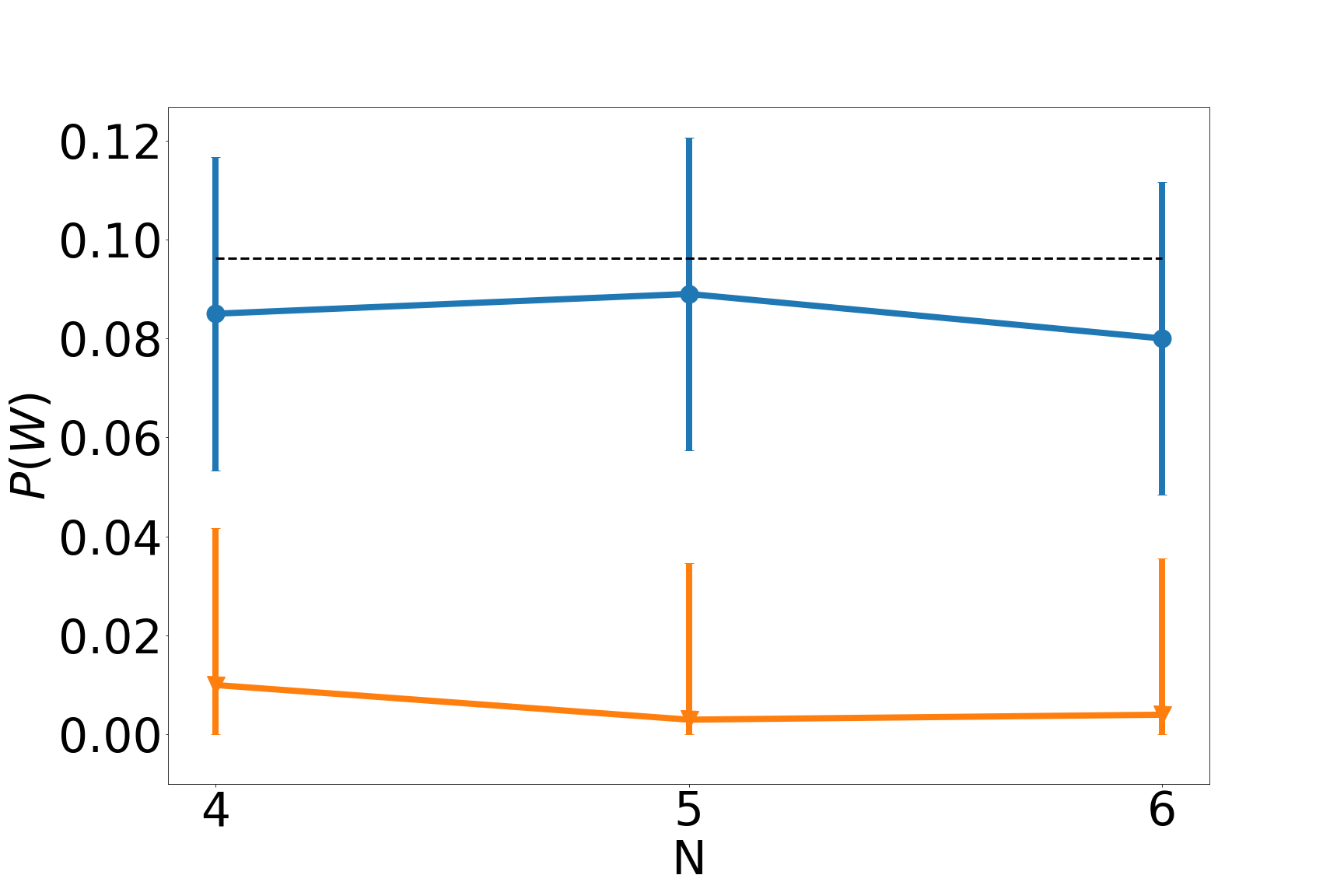}
    \caption{ Entanglement potencies for $W=W_{XZ}^{(2,N,\mathcal{G})}$ with both the witness observable and problem Hamiltonian corresponding to (top) fully connected $N$-qubit graph, for QAOA MaxCut problem, and (bottom) $N$-qubit ring graph, with $p=1$ (blue) and $p=5$ (orange) sampled numerically with 1000 uniform random $(\gamma,\beta)$ angles. Dashed line (bottom) shows numerical integration of Eq. \eqref{eq:heaviside_potency} at $\approx0.096$. Although we do not depict this in the figure, we similarly sampled 10,000 Haar random states (as opposed to only QAOA-MaxCut states), and found no violation of the separable threshold, so we consider the entanglement potency of both these observables for Haar random states to be zero for each of the investigated values of $N$.
    }
    \label{fig:MaxCut_potency}
\end{figure}
We repeat the MC analysis for fully connected MaxCut problem with the witness sharing the fully connected graph structure (see top figure in Fig.~\ref{fig:MaxCut_potency}). One notices that the potency of the witness decreases for $p=5$ QAOA-MaxCut states compared to $p=1$ states. This behavior is  unsurprising, since more layers of QAOA increases the algorithm's expressibility, allowing a larger fraction of non-detectable states to be explored,
so that the fraction of detectable entangled states should decrease.

We also numerically estimate the entanglement potency of $W_{XZ}^{(N,N,\mathcal{G})}$ with respect to $p=1$ QAOA states prepared via fully connected graphs with Hamiltonians
\begin{equation}
    H = \sum_{\langle i,j\rangle \in E_f}J_{i,j} Z_iZ_j,
\end{equation}
with $J_{i,j}\in\{-1,+1\}$.
In particular, we investigate the case for $N=4$, and employ the MC technique to estimate the potency of $W_{XZ}^{(4,4,\mathcal{G})}$ with respect to the set of $p=1$ QAOA states prepared according to each of 64 possible Hamiltonians (from which only 6 are non-isomorphic), with uniformly randomly drawn values of $\gamma$ and $\beta$ (see~Fig.~\ref{fig:Z2_potency}). Additionally, we check the scaling properties of randomly selected Hamiltonians that display $\mathbb{Z}_2$ symmetry (i.e. randomly sampling $J_{i,j}$ coefficients from $\{-1,1\}$) which constitute $\xi$. We report these entanglement potencies in Fig.~\ref{fig:Z2_potency} and in Table~\ref{table:Z2_potency}.
\begin{figure}
    \centering
    \includegraphics[width=0.48\textwidth]{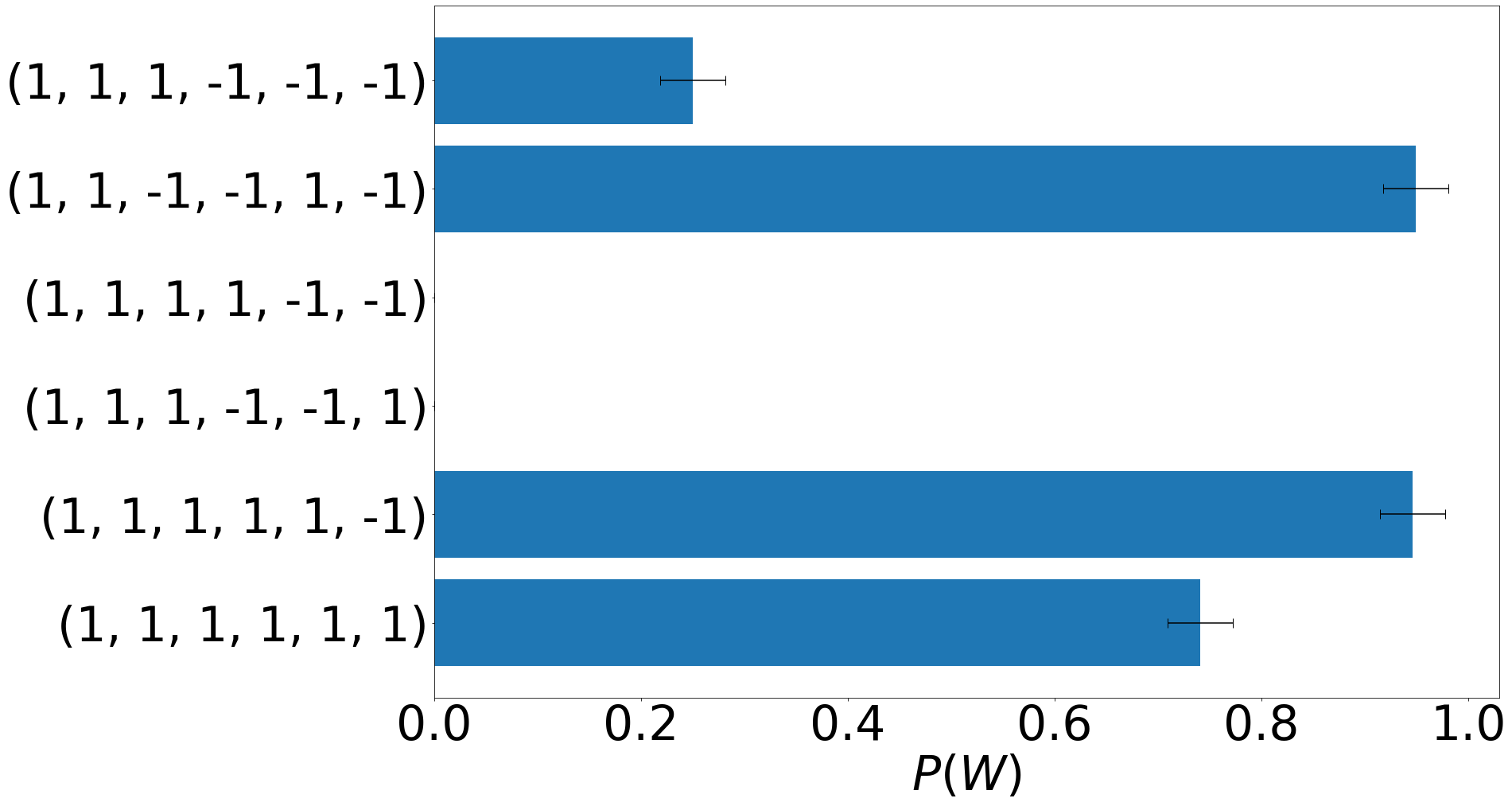}
    \caption{Entanglement potency of witness observables $W=W_{XZ}^{(N,N,\mathcal{G})}$ for random fully connected problem Hamiltonians of the form $\sum_{\langle i,j \rangle} J_{ij} Z_i \otimes Z_j$ with $J_{ij} \in \{-1,1\}$. (top) Depicts potency for all possible non-isomorphic fully connected Hamiltonians on 4 nodes (specified by $J_{ij}$ terms along the y-axis), sampled with 1000 random angles. Error bars represent sampling uncertainty $\sigma = \frac{1}{\sqrt{N}} \approx 0.03$.
    \label{fig:Z2_potency}}
\end{figure}

\begin{table*}[!htpb]
\setlength{\tabcolsep}{12pt}
\begin{tabular}{@{}*{8}{c}@{}*{8}{c}@{}*{10}{c}@{}*{10}{c}@{}*{8}{c}@{}*{8}{c}@{}*{10}{c}@{}*{10}{c}}\toprule
\multirow{2}{*}{N} & \multicolumn{3}{c}{QAOA $p=1$} & \multicolumn{3}{c}{QAOA $p=5$} & Random \\\noalign{\smallskip}
& mean & max & min & mean & max & min & mean\\
\midrule\addlinespace
4 & 0.48 & 0.96 & 0 & 0.48 & 0.85 & 0.16 & 0.12$\pm$0.01 \\\noalign{\smallskip}
5 & 0 & 0 & 0 & 0 & 0 & 0 &  0.09$\pm$0.01 \\\noalign{\smallskip}
6 & 0.47 & 0.52 & 0.42 & 0.50 & 0.53 & 0.46 & 0.05$\pm$0.01\\
\bottomrule
\end{tabular}
\caption{ Entanglement potency for a set of problem Hamiltonians with randomly chosen $J_{ij} \in \{-1,1\}$ (i.e. with $\mathbb{Z}_2$ symmetry), along with the minimum and maximum potency found in the set of problem Hamiltonians. Each potency was calculated using 1000 random angle samples. The data at $N=4,5,6$ were calculated using $60$,$80$, and $100$ random problem Hamiltonians, respectively - note that based on the histogram (Fig.~\ref{fig:Z2_potency} it is likely that for $N=6$ there can exist Hamiltonians for which QAOA $p=1$ has $P(W)=0$. The random state data represents the potency of the observable over 10,000 Haar random states.}\label{table:Z2_potency}
\end{table*}

\subsection{\label{sec:results}Results - Entanglement}

For various numbers of qubits, we report the experimental results of measuring the value of the entanglement witness $W_{XZ}^{(2,N,\mathcal{G}_\mathrm{line})}$ for linear chain $p=1$ QAOA-MaxCut states defined over the same set of edges as the respective cost Hamiltonian. We use randomized compiling \cite{randomized-compiling,beale_3945250} to twirl errors on the physical gates into stochastic errors. We fit a simple depolarizing model onto the results \cite{Xue_2021, local-noise-qaoa}, and obtain a reasonably good fit, as shown in Fig.~\ref{fig:expt_linear}. More details are provided in Appendix \ref{secn:experiment}.
\begin{figure}
    \centering
    \includegraphics[width=0.48\textwidth]{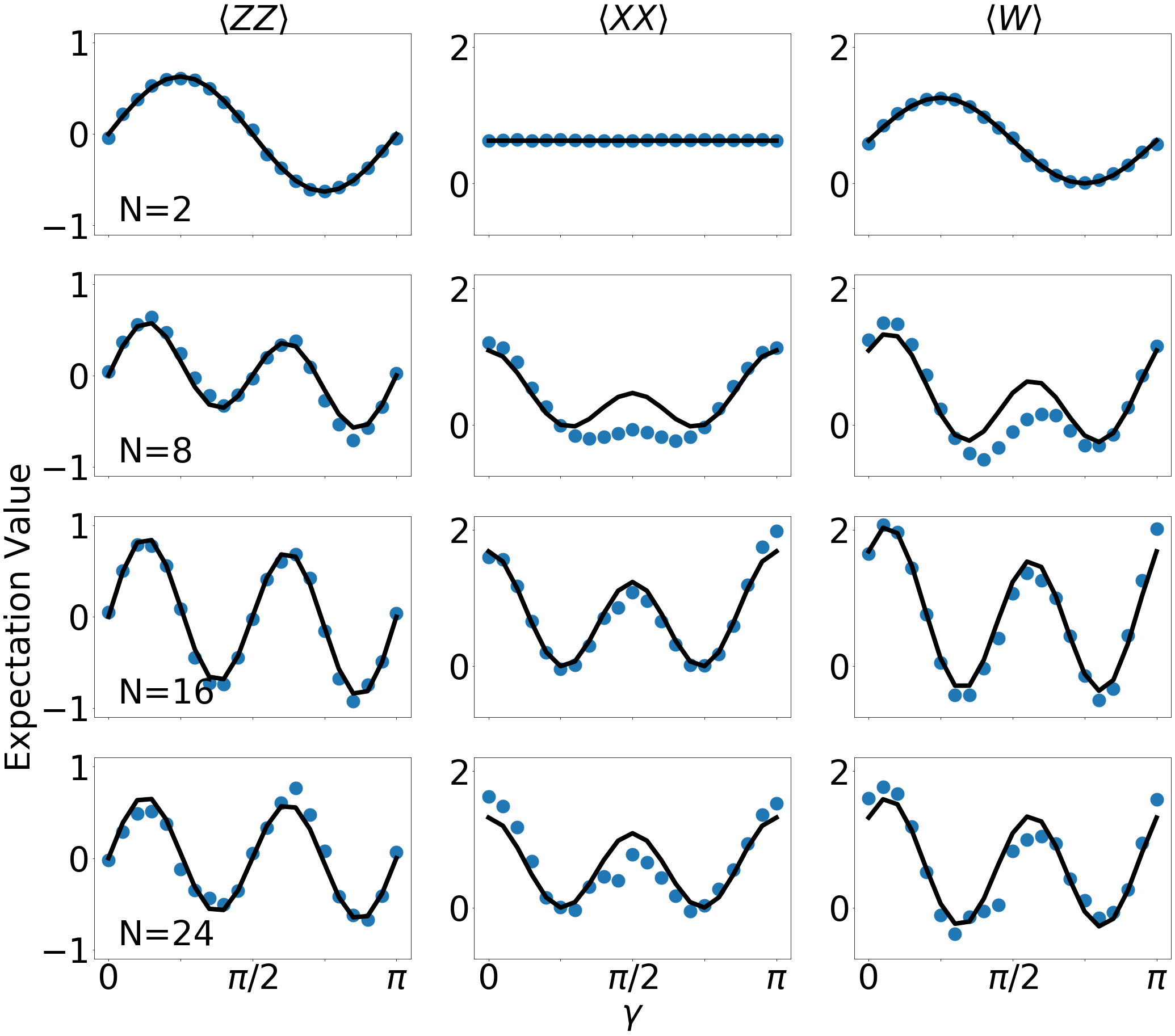}
    \caption{Experimental results for a linear chain $W = W_{XZ}^{(2,N,\mathcal{G}_{linear})} = \sum_{i}^{N-1} X_i X_{i+1} + Z_i Z_{i+1}$ witness observable from the Rigetti Aspen-9 QPU. For $N=2,8,16,24$ qubits (top-bottom), we plot the expectation values of the sum of all $ZZ$ terms (left-most column), $XX$ terms (middle column), and the witness observable $W$ (right-most column). We overlay the experimental data points with a fitted model that is given by some scale factor $(1 - p_{noise})$ times the analytical expression Eq. \eqref{eqn:W_XZ_2_linear_chain}.
    }
    \label{fig:expt_linear}
\end{figure}

In general, noise will reduce the chances of detecting entanglement. In the case of a global depolarizing noise channel
\begin{equation}
    \rho \rightarrow (1-p_{noise})\rho + p_{noise}\frac{\mathbbm{1}}{2^N}
\end{equation}
with noise parameter $p_{noise}$, the expectation value in the noisy mixed state scales down as $(1-p_{noise}) \langle W \rangle$, where $\langle W \rangle$ represents the expectation value in the ideal pure QAOA state.
Experimentally, we find a violation of the separable threshold at $N=2$ for the observables $W_{XZ}^{(2,2,\mathcal{G})}$ and $W_{XY}^{(2,2,\mathcal{G})}$ as depicted in Fig.~\ref{fig:W2_ent} (a similar violation was recently reported experimentally in \cite{gold2021entanglement}), we failed to find any such violations for $N \geq 3$. In Appendix \ref{secn:experiment}, we provide the values of the noise parameters that were obtained from fitting a global depolarizing channel to the data. It is found that for $N \geq 3$, these are above the critical threshold of the depolarizing noise parameter that one obtains from Eq. \eqref{eqn:Wmax-qaoa-linear-chain}
\begin{equation}
    p_{noise} \lesssim 1 - \frac{N-1}{1.207N - 1.019}
\label{eqn:p-critical}
\end{equation}
If the noise level on the hardware is below the threshold defined by Eq. \eqref{eqn:p-critical} for $N \geq 3$, then one would be able to certify entanglement in the linear chain $p=1$ QAOA-MaxCut state. For $N=2$, the noise on the hardware was sufficiently below the critical threshold $p_{noise} < \frac{1}{2}$ to enable the witness to detect entanglement in the QAOA state.

\begin{figure}
    \centering
    \includegraphics[width=0.48\textwidth]{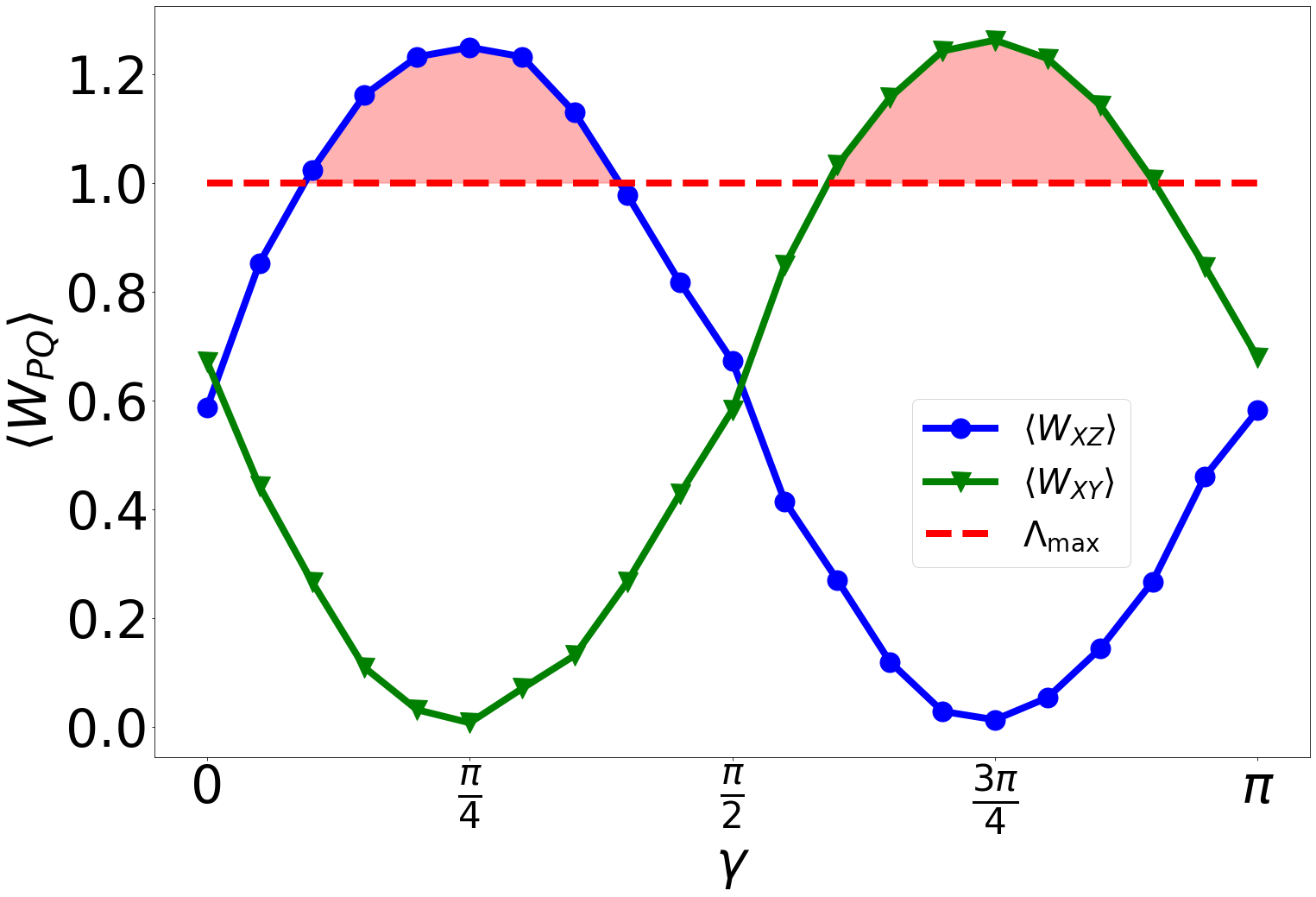}
    \caption{Expectation value of two observables: (green, wedges) $W_{XY}=W_{XY}^{(2,2,\mathcal{G}_{linear})}$ and (blue dots) $W_{XZ}=W_{XZ}^{(2,2,\mathcal{G}_{linear})}$ on Rigetti Aspen-9 chip for QAOA $p=1$ MaxCut problem (see Appendix~\ref{secn:experiment} for details of the experimental setup) as a function of $\gamma$ angle ($\beta=\frac{\pi}{8}$). Red area indicates region of violation of the separable bound $\Lambda_{\mathrm{max}}=1$ (red dashed line).}
    \label{fig:W2_ent}
\end{figure}

\section{\label{sec:Conclusions}Conclusions}

In this paper, we have introduced a practical method to verify key non-classical properties of a quantum algorithm's implementation on a physical device. Using measurements in only the three Pauli bases, we reconstruct the experimental single-qubit reduced density matrices (SQRDMs), and interpret their coherence as a basic measure of non-classicality, or quantumness. We identified a large family of obsevables, that can serve as entanglement witnesses and that could be measured using the same measurement data. Although our work focuses on QAOA-MaxCut, the same procecure could in principle be used to test the entanglement of any other state prepared on a (noisy) quantum device -- since Theorem~\ref{thm:sep-W_M_k} provides both upper and lower bounds for separable states, the experimentalist simply has to collect bitstrings in the $X$, $Y$ and $Z$ bases, compute the expectation value for a suitable observable and check if either the upper or lower bound is violated. Given the measurement data, estimating the expectation value of these observables is at most polynomial in the number of qubits, and therefore efficient.

Our work has also proposed a novel generalization of Bell-type observables, for which we established a non-trivial inequality. The generalized Cauchy-Schwarz inequality we used to prove one of our main results may have independent interest in other areas of quantum information, computer science, and related fields. We showed that entanglement witnesses for variational circuits can be constructed out of the cost Hamiltonian itself, without resorting to decomposing the projector of the entangled state into a possibly exponentially large set of measurable operators. This may inspire the construction of yet more witnesses in future work. In particular, we imagine that the techniques we have outlined in this paper could provide a foundation to the construction of similar witnesses in other parameteric families of circuits, such as those found in Variational Quantum Eigensolvers (VQE) or Quantum Machine Learning (QML), prevalent in the NISQ era. 

For QAOA-MaxCut in particular, we noted that the cost Hamiltonian and the corresponding 2-local witness defined on the same graph saturate in maximum expectation value at around the same $p$. Thus, even though further rounds of QAOA might generate more entanglement, the witness does not capture this extra entanglement. Instead, the observation that it tends to saturate around the same $p$ hints that it may provide a measure of the amount of useful entanglement that the algorithm employs in the optimization problem. In general, it is desirable to obtain measures of and verification procedures for the amount of entanglement that is algorithmically relevant, and our work provides a step in that direction.

As gate fidelities and qubit counts and connectivities on near-term devices improve, it will become increasingly important to verify that the hardware resources are employing genuinely non-classical resources in the execution of some algorithm, without which a quantum advantage would be impossible. Our work takes an important step in the direction of such algorithmic benchmarks. Additionally, identifying the entanglement properties in quantum circuits can help us in understanding the role of entanglement as a resource for quantum computing.

\begin{acknowledgments}
This material is based upon work supported by the Defense Advanced Research Projects Agency (DARPA) under agreement No. HR00112090058 and IAA 8839. F.W., J.S., Z.W., P.A.L. and D.V. also acknowledge USRA NASA Academic Mission Services (contract NNA16BD14C). 
The authors wish to thank Stuart Hadfield and Jeffrey Marshall for useful comments and feedback, as well as the Quantum Benchmark Inc. team for support on error modeling and mitigation. The experimental results were made possible by contributions across the Rigetti enterprise, and we especially acknowledge effort from Alex Hill, Nicholas Didier, Joseph Valery, and Greg Stiehl.
\end{acknowledgments}

\appendix
\section{Experimental details}\label{secn:experiment}
Experimental results were obtained on the 32 qubit Rigetti Aspen-9 system using a QAOA $p=1$ circuit ansatz shown in Fig.~\ref{fig-circuit}. Our implementation of the linear chain topology achieved constant circuit-depth by construction with parallel gates. Moreover, this approach enabled a systematic investigation, since adding new elements to the chain did not disrupt the prior circuit design. Median average gate fidelity for the CZ gate across the 24 qubit linear array under study was $F_{CZ}=94\%$ at the time of experiments, predominantly limited by decoherence. To mitigate residual coherent error, all circuits were randomly compiled 100 times under Pauli twirling \cite{randomized-compiling,beale_3945250}. Each circuit instance compiled to a random bit-flip pattern over the qubit register before measurement (later undone in post-processing) to remove readout bias. Each unique circuit was executed 300 times.

After fitting a global depolarizing model as described in Section \ref{sec:results}, we find that to 3 decimal places, $p_{noise} \approx 0.370$ for $N=2$, which is less than the critical threshold of $p_{noise} < \frac{1}{2}$, so that we observe a violation of the separable threshold at $N=2$. At $N=8$, $16$ and $24$, to 3 decimal places we obtain $p_{noise} \approx 0.844$, $0.888$ and $0.943$ respectively. These are all above the critical thresholds of $p_{critical} \approx 0.190$, $0.180$ and $0.177$ (to 3 decimal places, for $N=8$, $16$ and $24$ respectivelty) defined by Eq. \eqref{eqn:p-critical} so that for these values of $N$, we do not observe a violation of the separable threshold.

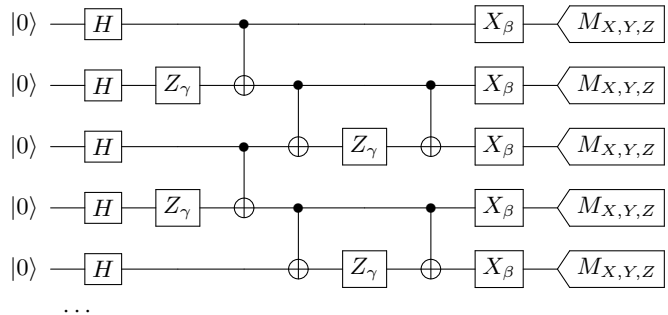
\begin{figure}[ht]
\centerline{
\Qcircuit @C=1.4em @R=1.0em {
\lstick{\ket{0}} & \gate{H} & \qw & \ctrl{1} & \qw & \qw & \qw & \gate{X_{\beta}} & \measuretab{M_{X,Y,Z}} \\
\lstick{\ket 0} & \gate{H}  & \gate{Z_{\gamma}} & \targ & \ctrl{1} & \qw & \ctrl{1} & \gate{X_{\beta}} & \measuretab{M_{X,Y,Z}} \\
\lstick{\ket 0} & \gate{H}  & \qw & \ctrl{1} & \targ & \gate{Z_{\gamma}} & \targ & \gate{X_{\beta}} & \measuretab{M_{X,Y,Z}} \\
\lstick{\ket 0} & \gate{H}  &  \gate{Z_{\gamma}} & \targ & \ctrl{1} & \qw & \ctrl{1} & \gate{X_{\beta}} & \measuretab{M_{X,Y,Z}} \\
\lstick{\ket 0} & \gate{H}  & \qw & \qw & \targ & \gate{Z_{\gamma}} & \targ & \gate{X_{\beta}} & \measuretab{M_{X,Y,Z}} \\
\rstick{\cdots}
}
}
\caption{Example circuit for QAOA, $p=1$ circuit ansatz (N=5), optimized for parallel execution on a linear chain topology at arbitrary size. CNOT gates are compiled to native CZ gates. Single-qubit rotations are compiled to continuous-angle $Z_{\phi}$ and fixed-angle $X_{90}$ gates. Measurements ($M_{X,Y,Z}$) are collected for three axes ($X,Y,Z$) using tomographic pre-rotations and $Z$-basis measurements.}\label{fig-circuit}
\end{figure}

In Fig. \ref{fig:coherences-rainbow}, we plot the coherences for some fixed value of $\beta=\frac{\pi}{8}$ {\it vs.} $\gamma$ ($\gamma = k\frac{\pi}{20}$ for $k=0 ,1,\ldots,20$). In the same figure, we additionally report average (middle semicircle), minimal (inner semicircle), and maximal (outer semicircle) coherences over all active qubits  for each $\gamma$.
\begin{figure*}
    \centering
\includegraphics[width = 0.98\textwidth]{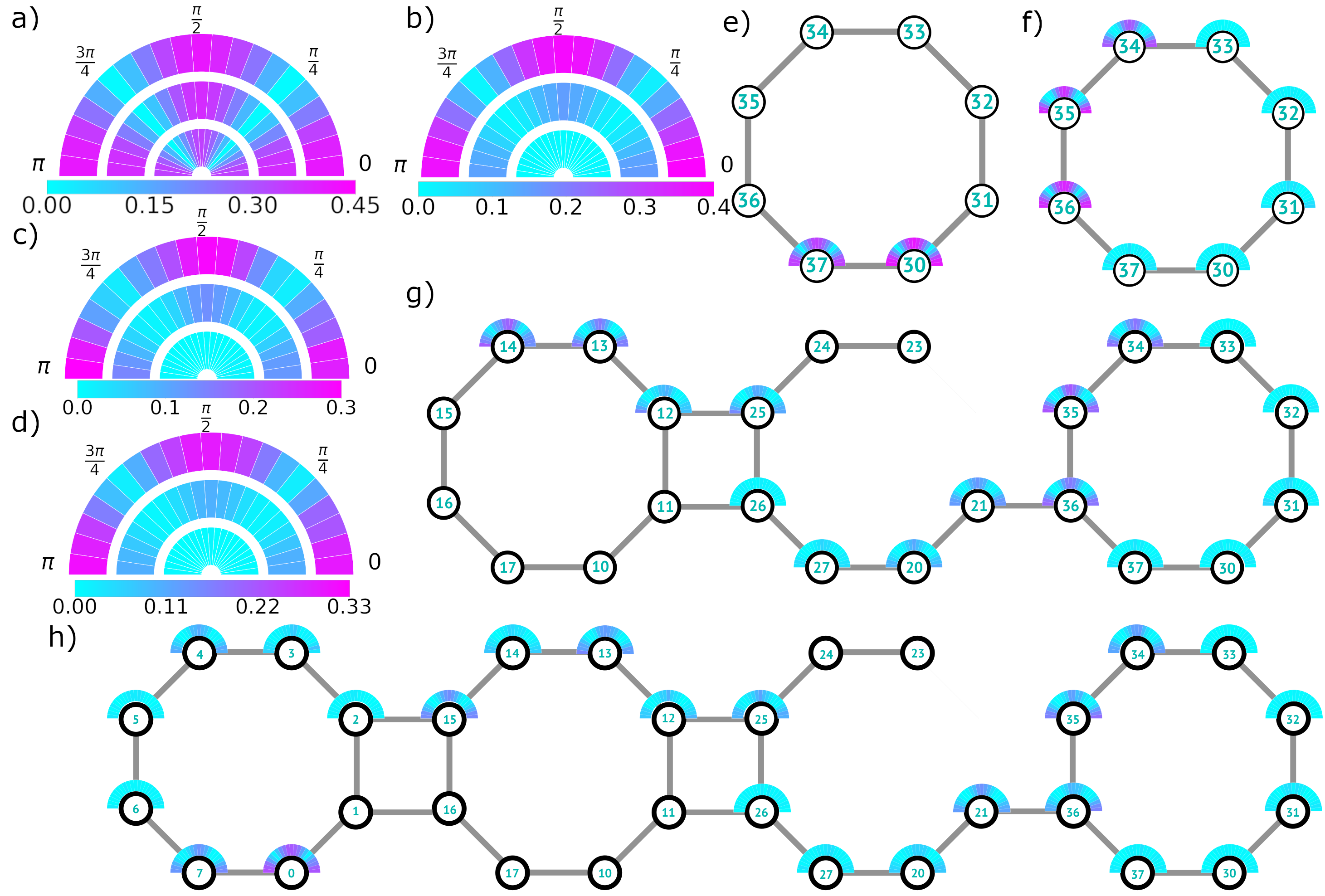}
    \caption{Coherences of SQRDMs for MaxCut $p=1$ QAOA run on Rigetti Aspen-9 chip (see text in Appendix~\ref{secn:experiment} for more details on the experiment) - a) and e): 2-qubit problem; b) and f): 8-qubit; c) and g): 16-qubit, d) and h): 24-qubit . Semicircles are divided into segments, each represents $\gamma$ angle ($\gamma = k\frac{\pi}{20}$ for $k=0,1,2\ldots,20$) with the coherence values color-coded (below a)-d) charts) for that QAOA setup (for all figures we used $\beta =\frac{\pi}{8}$). Figures a)-d) report maximal (outer semicircle), mean (middle semicircle) and minimal (inner semicircle) values of coherences across all used qubits. e)-h) specify active qubit's coherences and their position in Aspen-9 chip. }
    \label{fig:coherences-rainbow}
\end{figure*}
One can notice, that Rigetti Aspen-9 chip can build non-zero signatures of quantumness over all qubits across various problems of different size, subjected to non-negligible noise.

\section{General observables}
\label{secn:general-obs}
Here, we describe generalizations of the observables $W_{PQ}^{(k,N,\mathcal{G})}$ and $W_{XYZ}^{(k,N,\mathcal{G})}$ defined in Eqs. (\ref{eqn:W_PQ_k}) and (\ref{eqn:W_XYZ_k}) respectively, to which Theorem \ref{thm:sep-W_M_k} applies. We consider observables of the form
\begin{eqnarray}
W_{M}^{(k,N,\mathcal{G})} &=& \sum_{\langle i_1, \dots, i_k \rangle} \sum_{m=1}^{M} \alpha_m \bigotimes_{j=1}^{k}\sigma_{i_j}^{(a_{j,m})}
\label{eqn:W_M_kN}
\end{eqnarray}
where we restrict to $M \leq 3$, and
\begin{enumerate}
    \item $\alpha_m \in \{0, 1\}$ for all $1\leq m \leq M=3$,
    \item the indices $a_{j,m} \in \{1, 2, 3\}$ specify which single-qubit (non-identity) Pauli operator acts on the $j$-th qubit $i_j$ in the $m$-th term,
    \item $a_{j,1} \neq a_{j,2} \neq a_{j,3}$ for each $j$,
\end{enumerate}
and the superscript $\mathcal{G}$ refers to the generalized graph, i.e. the set of $k$-tuples $\langle i_1, \dots, i_k \rangle$ being summed over.
All observables of the form Eqs. (\ref{eqn:W_PQ_k}) and \eqref{eqn:W_XYZ_k} belong to this family. Note that in those expressions, we use the subscript to denote a Pauli string, not a numerical value for $M$. We choose this convention to simplify our expressions whenever the indices $a_{j,m}$ are fixed for all values of $j$ once $m$ is specified. The number of characters in the subscript string then specify the number of non-zero coefficients $\alpha_m$. Thus, $W_{XZ}^{(k,N,\mathcal{G})}$ is really an observable of the form $W_{2}^{(k,N,\mathcal{G})}$ where $\alpha_1 = \alpha_2 = 1$, $\alpha_3=0$ and we have fixed $a_{j,1} = 1$ and $a_{j,2} = 3$. Similarly, $W_{XYZ}^{(k,N,\mathcal{G})}$ is an observable of the form $W_{3}^{(k,N,\mathcal{G})}$ where $\alpha_m = 1$ and we have fixed $a_{j,m} = m$.

Using the same convention, we can also build other observables such as
\begin{equation}
    W_{XY}^{(k,N,\mathcal{G})} = \sum_{\langle i_1, \dots, i_k \rangle} \left( \bigotimes_{j=i_1}^{i_k} X_j + \bigotimes_{j=i_1}^{i_k} Y_j \right),
\label{eqn:W_XY_k}
\end{equation}
which is an observable of the form $W_{2}^{(k,N,\mathcal{G})}$ with $\alpha_1 = \alpha_2 = 1$, $\alpha_3 = 0$ and $a_{j,1}=1$, $a_{j,2}=2$ for all $j$. Another such example is
\begin{equation}
    W_{YZ}^{(k,N,\mathcal{G})} = \sum_{\langle i_1, \dots, i_k \rangle} \left( \otimes_{j=i_1}^{i_k} Y_j + \otimes_{j=i_1}^{i_k} Z_j \right)
\label{eqn:W_YZ_k}
\end{equation}
which is again an observable of the form $W_{2}^{(k,N,\mathcal{G})}$ with $\alpha_1 = \alpha_2 = 1$, $\alpha_3 = 0$ and $a_{j,1}=2$, $a_{j,2}=3$ for all $j$. We refer to any observable of the form $W_{XZ}^{(k,N,\mathcal{G})}$, $W_{XY}^{(k,N,\mathcal{G})}$ or $W_{YZ}^{(k,N,\mathcal{G})}$ as that of the form $W_{PQ}^{(k,N,\mathcal{G})}$. Note that $W_{ZYX}^{(k,N,\mathcal{G})}$ or $W_{YZX}^{(k,N,\mathcal{G})}$ etc. are operationally the same as $W_{XYZ}^{(k,N,\mathcal{G})}$.

However, none of these observables exploit the additional freedom allowed by Theorem \ref{thm:sep-W_M_k} to choose different $a_{j,m}$ for different values of $j$ as well as $m$. We can also build more complicated observables by allowing the choice of the indices $a_{j,m}$ to vary with the choice of $j$ as well as $m$, as long as condition 3 above is met. For example, with $M=2$, $k=6$ we could construct $a_{1,1}=1$, $a_{2,1}=2$, $a_{3,1}=3$, $a_{4,1}=3$, $a_{5,1}=2$, $a_{6,1}=1$ and $a_{1,2}=3$, $a_{2,2}=1$, $a_{3,2}=2$, $a_{4,2}=1$, $a_{5,2}=3$, $a_{6,2}=2$ to give
\begin{equation}
W = \sum_{\langle i,j,k,l,m,n \rangle} \left( X_i Y_j Z_k Z_l Y_m X_n + Z_i X_j Y_k X_l Z_m Y_n \right)
\end{equation}
As another example, with $M=3$, $k=5$ we could construct
\begin{eqnarray}
W &=& \sum_{\langle i, j, k, l, m \rangle} \left( X_i Y_j Z_k X_l Y_m + Z_i X_j Y_k Y_l X_m \right.\nonumber \\
&&\qquad\qquad\qquad\qquad \left. + Y_i Z_j X_k Z_l Z_m \right)
\end{eqnarray}
and so on. Theorem \ref{thm:sep-W_M_k} applies to this entire family of observables, which is much larger than those of the form in Eqs (\ref{eqn:W_PQ_k}) and  (\ref{eqn:W_XYZ_k}) alone.

\section{\label{secn:proofs}Proofs}

\textbf{Proof of Lemma \ref{thm:cauchy-schwarz}}
\begin{proof}
We will first prove the following identity involving the Hadamard product using induction
\begin{equation}
    \norm{\odot_{j=1}^{k} \vec{x}^{(j)}} \leq \prod_{j=1}^{k} \norm{\vec{x}^{(j)}}
\label{eqn:lemma-cs-hadamard-identity}
\end{equation}
where the Hadamard product is defined as
\begin{equation}
    \left( \odot_{j=1}^{k} \vec{x}^{(j)} \right)_i = \prod_{j=1}^{k} x_{i}^{(j)}.
\end{equation}
First, we prove the base case involving two vectors $\vec{a}$, $\vec{b} \in \mathbb{R}^{n}$ and their Hadamard product $\vec{v} = \vec{a} \odot \vec{b}$. Then,
\begin{eqnarray}
\norm{\vec{v}} &=& \sqrt{\left( a_1 b_1 \right)^2 + \dots + \left( a_n b_n \right)^2} \nonumber \\
&\leq& \sqrt{a_1^{2} \norm{\vec{b}}^2 + \dots + a_n^{2} \norm{\vec{b}}^2} \nonumber \\
&=& \norm{\vec{a}} \norm{\vec{b}}
\end{eqnarray}
where the inequality follows from noting that the absolute value of any component of a vector is bounded above by the (Euclidean) norm of that vector by definition, $\vert b_i \vert \leq \norm{\vec{b}} = \sqrt{\sum_{j=1}^{n} b_j^{2}}$ for any $\vec{b} \in \mathbb{R}^{n}$ and $i \in [n]$.
Next, we assume that Eq. \eqref{eqn:lemma-cs-hadamard-identity} holds true for some collection of $k-1$ vectors $\vec{x}^{(1)}, \dots, \vec{x}^{(k-1)} \in \mathbb{R}^{n}$, i.e. defining $\vec{v} = \odot_{j=1}^{k-1} \vec{x}^{(j)}$, we have $\norm{\odot_{j=1}^{(k-1)} \vec{x}^{(j)}} = \norm{\vec{v}} \leq \prod_{j=1}^{k-1} \norm{\vec{x}^{(j)}}$.
Then,
\begin{eqnarray}
\norm{\odot_{j=1}^{k} \vec{x}^{(j)}} &=& \norm{\vec{v} \odot \vec{x}^{(k)}} \nonumber \\
&=& \sqrt{v_1^{2} \left( x_1^{(k)}\right)^2 + \dots + v_n^{2} \left( x_n^{(k)}\right)^2} \nonumber \\
&\leq& \sqrt{\norm{\vec{v}}^2 \left( x_1^{(k)}\right)^2 + \dots + \norm{\vec{v}}^2 \left( x_n^{(k)}\right)^2} \nonumber \\
&=& \norm{\vec{v}} \norm{\vec{x}^{(k)}} = \prod_{j=1}^{k} \norm{\vec{x}^{(j)}}
\end{eqnarray}
and so Eq. \eqref{eqn:lemma-cs-hadamard-identity} follows by induction. We can now prove the main statement of our Lemma. Defining $\vec{w} = \odot_{j=1}^{m-1} \vec{x}^{(j)}$, we now have
\begin{eqnarray}
\left\vert \sum_{i=1}^{n} \left( \odot_{j=1}^{m} \vec{x}^{(j)}\right)_i \right\vert &=& \left\vert \sum_{i=1}^{n} \left( \vec{w}\odot \vec{x}^{(m)} \right)_i \right\vert \nonumber \\
&=& \vert \vec{w} \cdot \vec{x}^{(m)} \vert \nonumber \\
&\leq& \norm{\vec{w}} \norm{\vec{x}^{(m)}} \nonumber \\
&\leq& \prod_{j=1}^{m-1} \norm{\vec{x}^{(j)}} \norm{\vec{x}^{(m)}} \nonumber \\
&=& \prod_{j=1}^{m} \norm{\vec{x}^{(j)}}
\end{eqnarray}
which is the statement of the Lemma. In the proof above, the first inequality follows from the standard Cauchy-Schwarz inequality, while the second inequality follows from Eq. \eqref{eqn:lemma-cs-hadamard-identity}.

\end{proof}

\textbf{Proof of Theorem \ref{thm:sep-W_M_k}}
\begin{proof}
Given a pure state $\vert \psi \rangle = \cos{\frac{\theta}{2}} \vert 0 \rangle + e^{i\phi} \sin{\frac{\theta}{2}} \vert 1 \rangle$, we have
\begin{eqnarray}
\langle X \rangle_{\psi} &=& \sin{\theta} \cos{\phi} \nonumber \\
\langle Y \rangle_{\psi} &=& \sin{\theta} \sin{\phi} \nonumber \\
\langle Z \rangle_{\psi} &=& \cos{\theta}
\end{eqnarray}
A fully separable pure state is given by $\vert \Psi \rangle = \otimes_{j=1}^{N} \left( \cos{\frac{\theta_j}{2}} \vert 0 \rangle + e^{i\phi_j} \sin{\frac{\theta_j}{2}} \vert 1 \rangle\right)$. For $M=1$ alone, it is easy to prove the the upper bound of the absolute expectation values in such states, as follows
\begin{eqnarray}
\left\vert \langle W_{1}^{(k,N,\mathcal{G})} \rangle\right\vert &=& \left\vert \sum_{\langle i_1, \dots, i_k \rangle} \prod_{j=1}^{k} \langle \sigma_{i_j}^{(a_{j,1})} \rangle\right\vert \nonumber \\
&\leq& \left\vert \sum_{\langle i_1, \dots, i_k \rangle} \prod_{j=1}^{k} (1) \right\vert = \vert E_k \vert
\end{eqnarray}
For $1\leq M\leq 3$, we note that $\langle W_{M}^{(k,N,\mathcal{G})} \rangle$ can be written as a Hadamard product of $k$ vectors
\begin{eqnarray}
\langle W_{M}^{(k,N,\mathcal{G})} \rangle &=& \sum_{\langle i_1, \dots, i_k \rangle} \sum_{m=1}^{M} \prod_{j=1}^{k} \langle \sigma_{i_j}^{a_{j,m}} \rangle \nonumber \\
&=& \sum_{\langle i_1, \dots, i_k \rangle} \sum_{m=1}^{M} \left( \odot_{j=1}^{k} \vec{v}_{i_j}^{(M)} \right)_m
\end{eqnarray}
where $\vec{v}_{i_j}^{(M)} = \left( \langle \sigma_{i_j}^{(a_{j,1})} \rangle, \dots, \langle \sigma_{i_j}^{(a_{j,M})} \rangle\right) \in \mathbb{R}^M$ (since expectations of Pauli operators are real) where $j$ labels some index in a $k$-tuple.
Note that $\norm{\vec{v}_{i_j}^{(1)}} \leq 1$ since the maximum (absolute) expectation value of any Pauli operator is 1. Next, since $\forall j$, $a_{j,1} \neq a_{j,2} \neq a_{j,3}$, we have that $\norm{\vec{v}_{i_j}^{(2)}}$ is the square root of the sum of squares of $\langle X\rangle$ and $\langle Y \rangle$, $\langle X \rangle$ and $\langle Z\rangle$, or $\langle Y\rangle$ and $\langle Z\rangle$. All of these cases are bounded above by $\norm{\vec{v}_{i_j}^{(3)}}$.
Furthermore,
\begin{eqnarray}
\norm{\vec{v}_{i_j}^{(3)}} &=& \sqrt{\sin^{2}{\theta_{i_j}}\cos^{2}{\phi_{i_j}} + \sin^{2}{\theta_{i_j}}\sin^{2}{\phi_{i_j}} + \cos^{2}{\theta_{i_j}}} \nonumber \\
&=& 1
\end{eqnarray}
so that we have $\norm{\vec{v}_{i_j}^{(M)}} \leq 1$ for any $M \in \{1,2,3\}$. Therefore, we obtain
\begin{eqnarray}
\left\vert \langle W_{M}^{(k,N,\mathcal{G})} \rangle \right\vert &=& \left\vert \sum_{\langle i_1, \dots, i_k \rangle} \sum_{m=1}^{M} \left( \odot_{j=1}^{k} \vec{v}_{i_j}^{(M)} \right)_m \right\vert \nonumber \\
&\leq& \sum_{\langle i_1, \dots, i_k \rangle} \left\vert \sum_{m=1}^{M} \left( \odot_{j=1}^{k} \vec{v}_{i_j}^{(M)} \right)_m \right\vert \nonumber \\
&\leq& \sum_{\langle i_1, \dots, i_k \rangle} \prod_{j=1}^{k} \norm{\vec{v}_{i_j}^{(M)}} \nonumber \\
&\leq& \sum_{\langle i_1, \dots, i_k \rangle} \prod_{j=1}^{k} (1) = \vert E_k \vert
\label{eqn:thm1-pure-states-bound}
\end{eqnarray}
where the first inequality follows from repeated use of the triangle inequality $\vert a + b\vert \leq \vert a \vert + \vert b \vert$ for any $a,b \in \mathbb{R}$, the second inequality follows from Lemma \ref{thm:cauchy-schwarz} and the last inequality follows from the bound $\norm{\vec{v}_{i_j}^{(M)}} \leq 1$ as noted above.
Next, suppose we are given some fully separable mixed state
\begin{equation}
\rho = \sum_j \alpha_j \rho_{j_1} \otimes \dots \otimes \rho_{j_N}
\end{equation}
where each $\alpha_j > 0$ and $\sum_j \alpha_j = 1$. Each of the density operators in the tensor product can be expressed as
\begin{equation}
\rho_{j_{k}} = \sum_{l} c_{j_{k},l} \vert \psi_{j_{k}, l} \rangle \langle \psi_{j_{k}, l} \vert 
\end{equation}
with $c_{j_{k},l} > 0$ and $\sum_{l} c_{j_{k},l} = 1$. Therefore, by absorbing these constants into the definition of the $\alpha_j$'s, we have
\begin{equation}
\rho = \sum_j \alpha_j \vert \psi_{j_1} \rangle \langle \psi_{j_1} \vert \otimes \dots \otimes \vert \psi_{j_N} \rangle \langle \psi_{j_N} \vert
\end{equation}
and defining the separable pure states $\vert \Psi_j \rangle \langle \Psi_j \vert = \otimes_{m=j_1}^{j_n} \vert \psi_{m} \rangle \langle \psi_{m} \vert$, we have
\begin{eqnarray}
\left\vert \tr{\left( \rho W_{M}^{(k,N,\mathcal{G})}\right)} \right\vert &=& \left\vert \tr{\left( \sum_j \alpha_j \vert \Psi_j \rangle \langle \Psi_j \vert \cdot W_M^{(k,N,\mathcal{G})}\right)} \right\vert \nonumber \\
&=& \left\vert \sum_j \alpha_j \tr{\left( \ket{\Psi_j} \bra{\Psi_j} \cdot W_M^{(k,N,\mathcal{G})}\right)}\right\vert \nonumber \\
&\leq& \sum_j \alpha_j \left\vert \tr{\left( \ket{\Psi_j} \bra{\Psi_j} \cdot W_M^{(k,N,\mathcal{G})} \right)} \right\vert \nonumber \\
&\leq& \sum_j \alpha_j \cdot \vert E_k \vert = \vert E_k \vert
\end{eqnarray}
where the first inequality follows from repeated use of the triangle inequality, and the last inequality follows from Eq. \eqref{eqn:thm1-pure-states-bound} for pure states.
\end{proof}

\textbf{Proof of Theorem \ref{thm:max_ev}}
\begin{proof}
We will first prove the theorem for the case of $W_{XZ}^{(k,N,\mathcal{G})}$. One way to lower bound the largest eigenvalue of $W = \sum_{\langle i_1, \dots, i_k \rangle} \bigotimes_{j=i_1}^{i_k} X_j + \bigotimes_{j=i_1}^{i_k} Z_j$ is to simply compute its expectation value in some state. A convenient choice is $\vert + \rangle^{\otimes N}$. In this state, we find $\langle W \rangle = \vert E_k \vert$, therefore $\lambda_{max} \geq \vert E_k \vert$. However, we would like to increase this lower bound, since $\vert E_k \vert$ is met by separable states. In \cite{van_meighem1}, a sharper lower bound for the largest eigenvalue of any $m \times m$  symmetric matrix is provided as
\begin{equation}
\lambda_{max} \geq \frac{N_1}{m} + 2 \left( \frac{N_3}{2m} - \frac{N_1 N_2}{m^2} + \frac{N_{1}^{3}}{2m^3} \right) \lambda_{0}^{-2} + O(t^{-4})
\label{eqn:lower_bound_symm_matrix}
\end{equation}
where $u = [ 1\, 1 \dots 1]^{T}$ is the $m \times 1$ all-ones vector, $N_{k} = u^{T} A^{k} u$, $B = \frac{1}{\sqrt{m}} \text{max}_{1\leq j \leq m} (W_{jj} + \sum_{i=1;i\neq j}^{m} \vert W_{ij} \vert)$, $\lambda_0 = t \sqrt{m}$ and $t$ is any real number satisfying $t \geq B$. In our case, we have $m = 2^N$.

We already noted that $N_1/m = \vert E_k \vert$. Now,
\begin{eqnarray}
W^2 &=& \left( \sum_{\langle i_1, \dots, i_k \rangle} ( \bigotimes_{m=i_1}^{i_k} X_m + \bigotimes_{m=i_1}^{i_k} Z_m)  \right)^2 \nonumber \\
&=& \sum_{\langle i_1, \dots, i_k \rangle} \sum_{\langle j_1, \dots, j_k \rangle} \left( X_{i_1} \dots X_{i_k} X_{j_1} \dots X_{j_k}  + \right. \nonumber \\
&& \qquad X_{i_1} \dots X_{i_k} Z_{j_1} \dots Z_{j_k} + Z_{i_1} \dots Z_{i_k} X_{j_1} \dots X_{j_k}\nonumber \\
&& \left.\qquad + Z_{i_1} \dots Z_{i_k} Z_{j_1} \dots Z_{j_k}\right)
\end{eqnarray}
Under expectation in the $\vert + \rangle^{\otimes N}$ state, the first term becomes $\left( \sum_{\langle i_1,\dots,i_k \rangle} \langle X_{i_1} \rangle \dots \langle X_{i_k} \rangle \right)^2 = \vert E_k \vert^2$. Both terms involving a product of $X\dots X$ and $Z\dots Z$ terms are zero under expectation, since either the $Z_i$ terms will survive under the product, or multiply with an $X_i$ to produce a $Y_i$, and since $\langle + \vert Z \vert + \rangle = \langle + \vert Y \vert + \rangle = 0$, the entire term vanishes in expectation. The last term $Z_{i_1}\dots Z_{i_k} Z_{j_1}\dots Z_{j_k}$ can only contribute whenever $\langle i_1,\dots,i_k \rangle$ and $\langle j_1, \dots, j_k \rangle$ specify the same $k$-tuple, so that $\left\langle \sum_{\langle i_1, \dots, i_k \rangle} \sum_{\langle j_1, \dots, j_k \rangle} Z_{i_1}\dots Z_{i_k} Z_{j_1} \dots Z_{j_k} \right\rangle = \vert E_k \vert$. Altogether, this gives us
\begin{eqnarray}
\frac{N_2}{m} &=& \langle W^2 \rangle = \vert E_k \vert^{2} + \vert E_k \vert
\end{eqnarray}
To compute the $N_3$ term, we observe that $W^3$ consists of four types of terms:
\begin{enumerate}
    \item One term of the form $X_{i_1} \dots X_{i_k} X_{j_1}\dots X_{j_k} X_{l_1}\dots X_{l_k}$, which under expectation give a total contribution of $\vert E_k \vert^2$, following the same reasoning as in the computation of $N_2 / m$ above,
    \item Three terms of the form $\left( X_{i_1}\dots X_{i_k}\right)^2 \left( Z_{i_1}\dots Z_{i_k}\right)$, each of which give a contribution of 0 under expectation, since the $Z_m$'s have no choice but to either persist or multiply with an $X_m$ to give $Y_m$, and $\langle Z_m \rangle = \langle Y_m \rangle = 0$,
    \item Three terms of the form $\left( X_{i_1}\dots X_{i_k}\right)\left(Z_{i_1}\dots Z_{i_k} \right)^2$, each of which give a contribution of $\vert E_k \vert^2$ since $\sum_{i_1, \dots, i_k} \langle X_{i_1} \dots X_{i_k} \rangle = \vert E_k \vert$, and $\langle \left( \sum_{\langle i_1, \dots, i_k} Z_{i_1} \dots Z_{i_k} \right)^2 \rangle = \vert E_k \vert$ for the same reason we noted earlier in the computation of $N_2 / m$,
    \item One term of the form $Z_{i_1} \dots Z_{i_k} Z_{j_1}\dots Z_{j_k} Z_{l_1}\dots Z_{l_k}$, which has non-trivial contributions from triangles within a graph, and is generally tedious to compute.
\end{enumerate}
Even without explicitly calculating this last term, we have
\begin{eqnarray}
\frac{N_3}{m} &=& \langle W^3 \rangle \geq \vert E_k \vert^{3} + 3 \vert E_k\vert^{2}
\end{eqnarray}
Combining these results, we have
\begin{eqnarray}
&& \frac{N_3}{2m} - \frac{N_1 N_2}{m^2} + \frac{N_{1}^{3}}{2m^3} \nonumber \\
&\geq& \frac{\vert E_k\vert^{3} + 3 \vert E_k\vert^2}{2} - \vert E_k\vert^3 - \vert E_k\vert^2 + \frac{\vert E_k\vert^3}{2} \nonumber \\
&=& \frac{\vert E_k\vert^2}{2} > 0
\end{eqnarray}
Letting $t=\frac{\sqrt{m}}{2}B = \vert E_k \vert$, we have $t \geq B$ for all $N \geq 2$, and $\lambda_0 = t \sqrt{m} = \vert E_k \vert \sqrt{2^N}$. Plugging these values into Eq. (\ref{eqn:lower_bound_symm_matrix}), we find that
\begin{eqnarray}
2 \left( \frac{N_3}{2m} - \frac{N_1 N_2}{m^2} + \frac{N_{1}^{3}}{2m^3} \right) \lambda_{0}^{-2} \nonumber &\geq& \frac{1}{2^N} > 0
\label{eqn:lower_bound_max_eigval}
\end{eqnarray}
for any finite $N$, so that
\begin{eqnarray}
\lambda_{max} &\geq& \vert E_k \vert + \frac{1}{2^N} + O\left( \frac{1}{\vert E_k \vert^4}\right) \nonumber \\
&>& \vert E_k \vert + O\left( \frac{1}{\vert E_k \vert^4} \right)
\end{eqnarray}
which shows that the theorem holds for $W_{XZ}^{(k,N,\mathcal{G})}$. Noting that the eigenspectrum of any observable $A$ remains invariant under a unitary transformation $U^{\dagger} A U$, we can perform single-qubit rotations $U = \bigotimes_{q=1}^{N} R_X^{(q)}(\pi/2)$ (where $R_X^{(q)}(\alpha) = \exp(-i\alpha X_q)$) at every qubit to have Eq. (\ref{eqn:lower_bound_max_eigval}) apply to observables of the form $W_{XY}^{(k,N,\mathcal{G})}$ as well. Similarly, choosing $U=\bigotimes_{q=1}^{N} R_Z^{(q)}(\pi/2)$ yields the same bound for observables of the form $W_{YZ}^{(k,N,\mathcal{G})}$. This exhausts all the possibilities for $W_{PQ}^{(k,N,\mathcal{G})}$, and the theorem holds.

\end{proof}

\textbf{Proof of Lemma \ref{lemma:xx-zz-expects}}
\begin{proof}
These expressions can be derived using essentially the procedure described in \cite{jw-qaoa}. In that paper, the authors derive the expression for $\langle C_{uv} \rangle$, where $C_{uv} = \frac{1}{2} \left( \mathbbm{1} - Z_u Z_v \right) \rightarrow \langle Z_u Z_v \rangle = 1 - 2 \langle C_{uv} \rangle$, and where the QAOA state is built from the cost Hamiltonian $C^{\prime} = \frac{1}{2} \sum_{\langle u,v \rangle} \left( \mathbbm{1} - Z_u Z_v \right)$. The expression they derive for $\langle Z_u Z_v \rangle$ matches our expression for $\langle Z_u Z_v \rangle$ as long as one adjusts for the difference in convention by replacing $\gamma \rightarrow 2\gamma$ and $\langle Z_u Z_v \rangle \rightarrow 2 \langle Z_u Z_v \rangle$ in going from the notation of \cite{jw-qaoa} to our expression.

The expression for $\langle X_u X_v \rangle$ is derived similarly in the Heisenberg picture. We first note that $e^{i\beta B} X_u X_v e^{-i\beta B} = X_u X_v$ for $B = \sum_{i=1}^{N} X_i$, so that
\begin{eqnarray}
\langle X_u X_v \rangle &=& \tr{\left( e^{i\gamma C} e^{i\beta B} X_u X_v e^{-i\beta B} e^{-i\gamma C} \rho_0 \right)} \nonumber \\
&=& \tr{\left( e^{i\gamma C} X_u X_v e^{-i\gamma C} \rho_0 \right)}
\end{eqnarray}
where $\rho_0 = \prod_{q=1}^{N} \frac{1}{2} \left( \mathbbm{1} + X_q\right)$ is the initial state. Further, since $[Z_u Z_v, X_u X_v] = 0$ and $[Z_i Z_j, X_u X_v] = 0$ whenever $i \neq u$, $j \neq v$, the only terms that survive are those involving $u$ ($v$) and its neighbors excluding $v$ ($u$). Denoting the number of such neighbors as $d_u$ and $d_v$ respectively for $u$ and $v$, and defining $C_u = \sum_{j=1}^{d_u} Z_u Z_{a_j}$ and $C_v = \sum_{i=1}^{d_v} Z_v Z_{b_i}$, we have
\begin{eqnarray}
&& \tr{\left( e^{i\gamma C} X_u X_v e^{-i\gamma C} \rho_0 \right)} \nonumber \\
&=& \tr{\left( e^{i\gamma C_v} e^{i\gamma C_u} X_u X_v e^{-i\gamma C_u} e^{-i\gamma C_v} \rho_0 \right)} \nonumber \\
&=& \tr{\left( e^{i 2\gamma C_v} e^{i2\gamma C_u} X_u X_v\rho_0 \right)}
\end{eqnarray}
where in the last line we have used the fact that $X_m Z_m = -Z_m X_m$ to move the $e^{-i\gamma C_u}$ and $e^{-i\gamma C_v}$ terms past $X_u X_v$. Now, $\tr{\left( A_u B_v \rho_0 \right)} = 1$ only when $A, B \in \{\mathbbm{1}, X\}$, so that the only terms in the product $e^{i 2\gamma C_v} e^{i2\gamma C_u}$ that can contribute non-trivially are those that are proportional to the identity. Therefore, using $e^{i\alpha Z_a Z_b} = \cos{\alpha} \mathbbm{1} + i\sin{\alpha} Z_a Z_b$, we see that for triangle-free edges,
\begin{equation}
\langle X_u X_v \rangle = \cos^{d_u + d_v}{2\gamma}
\end{equation}
which agrees with the expression in Lemma \ref{lemma:xx-zz-expects} for the case $f=0$. In the more general case, only terms in the expansion that are proportional to even powers of $Z_u$ and $Z_v$ (and therefore the identity) contribute non-trivially. We then have
\begin{eqnarray}
&& \langle X_u X_v \rangle \nonumber \\
&=& \sum_{j=0,2,4,\dots}^{f} {f \choose j} \left( \cos{2\gamma}\right)^{d_u + d_v - 2j} (i\sin{2\gamma})^{2j} \nonumber \\
&=& \cos^{d_u + d_v - 2f}{2\gamma} \nonumber \\
&& \; \times \sum_{j=0,2,4,\dots}^{f} {f \choose j} (\cos^{2}{2\gamma})^{f-j} (-\sin^{2}{2\gamma})^{j}
\end{eqnarray}
and using the identity
\begin{equation}
    \sum_{j=0,2,4,\dots}^{f} {f \choose j} a^{f-j} b^{j} = \frac{1}{2} \left[ (a + b)^{f} + (a-b)^{f} \right]
\label{eqn:identity-xx-zz-lemma-sum-even}
\end{equation}
we finally obtain
\begin{equation}
    \langle X_u X_v \rangle = \frac{1}{2} \cos^{d_u + d_v - 2f}{2\gamma} \left(1 + \cos^{f}{4\gamma} \right)
\end{equation}
as stated in the Lemma.

To derive the expression for $\langle Y_uY_v \rangle$, we make use of the relation
$SWAP_{uv} = \frac{1}{2}\left(X_uX_v+Y_uY_v+Z_uZ_v+\mathbbm{1}\right)$ and derive the expectation value for $SWAP_{uv}$ first.
Notice that $[Y_u Y_v, B] = - [Z_u Z_v, B]$ and $[X_u X_v, B] = [\mathbbm{I}, B] = 0$, so that $[SWAP_{uv},B]=0$, and therefore the mixing unitaries do not contribute to the expectation value:
\begin{eqnarray}
\langle SWAP_{uv}\rangle &=& 
\tr{\left( e^{i\gamma C} e^{i\beta B} X_u X_v e^{-i\beta B} e^{-i\gamma C} \rho_0 \right)} \nonumber \\
&=& \tr{\left( e^{i\gamma C} SWAP_{uv} e^{-i\gamma C} \rho_0 \right)}
\end{eqnarray}
Inside the trace, only terms involving $u$ or $v$ contribute, and by making use of $SWAP_{uv} e^{i\gamma Z_uZ_{u'}} SWAP_{uv}= e^{i\gamma Z_vZ_{u'}}$, we have
\begin{eqnarray}
&& \langle SWAP_{uv}\rangle \nonumber \\
&=& 
\tr\Big( 
e^{i\gamma \sum_{u'\in\mathcal{N}(u)}Z_uZ_{u'}} 
e^{i\gamma \sum_{v'\in\mathcal{N}(v)}Z_vZ_{v'}}
\nonumber\\
&& \;SWAP_{uv}
e^{-i\gamma \sum_{u'\in\mathcal{N}(u)}Z_uZ_{u'}} 
e^{-i\gamma \sum_{v'\in\mathcal{N}(v)}Z_vZ_{v'}}
\rho_0 \Big)\nonumber\\
&=&
\tr{\left( 
e^{i\gamma (Z_u-Z_v)\left(\sum_{u'\in\mathcal{N}'(u)}Z_{u'}-\sum_{v'\in\mathcal{N}'(v)}Z_{v'}\right)} 
SWAP_{uv} \rho_0 \right)}\nonumber\\
&=&
\tr{\left( 
e^{i\gamma (Z_u-Z_v)\left(\sum_{u'\in\mathcal{N}'(u)}Z_{u'}-\sum_{v'\in\mathcal{N}'(v)}Z_{v'}\right)} 
\rho_0 \right)}
\end{eqnarray}
where 
$\mathcal{N}(u)$ is the set of neighbors of $u$, $\mathcal{N}'(u)$ is the set of nodes that are neighbors of $u$ but not $v$ or neighbors of $v$, with $\mathcal{N}'(v)$ defined similarly,
and the last step is due to $SWAP_{uv}\rho_0=\rho_0$. In the second equality above, we have made use of the fact that $\sum_{u'\in\mathcal{N}(u)}Z_{u'}-\sum_{v'\in\mathcal{N}(v)}Z_{v'} = \sum_{u'\in\mathcal{N}'(u)}Z_{u'}-\sum_{v'\in\mathcal{N}'(v)}Z_{v'}$. Since the exponential contributes either identity or $Z$'s, terms in $\rho_0$ that have $X$ trace to zero, and we further have
\begin{eqnarray}
&& \langle SWAP_{uv}\rangle \nonumber \\
&=& \frac{1}{2^N}\tr{\left( e^{i\gamma (Z_u-Z_v)\left(\sum_{u'\in\mathcal{N}'(u)}Z_{u'}-\sum_{v'\in\mathcal{N}'(v)}Z_{v'}\right)}\right)}\nonumber\\
&=& \frac{1}{2^N} \times \nonumber \\
&& \tr\left(
\prod_{u'\in\mathcal{N}'(u)}
\left(\cos\gamma+i\sin\gamma Z_uZ_{u'}\right)
\left(\cos\gamma-i\sin\gamma Z_vZ_{u'}\right) \right.
\nonumber\\
&&
\left.\prod_{v'\in\mathcal{N}'(v)}
\left(\cos\gamma-i\sin\gamma Z_uZ_{v'}\right)
\left(\cos\gamma+i\sin\gamma Z_vZ_{v'}\right)
\right) \nonumber \\
\end{eqnarray}
There are two types of products that come from the above expansion. The first kind are those in which terms from both the $\mathcal{N}'(u)$ and $\mathcal{N}'(v)$ products are proportional to the identity, which give the following contribution
\begin{eqnarray}
&& \sum_{i=0,2,4,\dots}^{d'_{u}} {d'_u \choose i} \left( \cos^{2}{\gamma}\right)^{d'_u - i} \left( \sin^{2}{\gamma}\right)^{i} \nonumber \\
&& \; \times \sum_{j=0,2,4,\dots}^{d'_{v}} {d'_v \choose v} \left( \cos^{2}{\gamma}\right)^{d'_v - j} \left( \sin^{2}{\gamma}\right)^{j} \nonumber \\
&=& \frac{1}{4} \left( 1 + \cos^{d'_u}{2\gamma}\right) \left( 1 + \cos^{d'_v}{2\gamma}\right)
\label{eqn:lemma-xx-zz-swap-even-terms}
\end{eqnarray}
where $d'_u = \vert \mathcal{N}'(u) \vert = d_u - f$, and $d'_v$ is defined similarly, and we have used Eq. \eqref{eqn:identity-xx-zz-lemma-sum-even}.
The second kind of terms are those in which terms from both the $\mathcal{N}'(u)$ and $\mathcal{N}'(v)$ products are proportional to $Z_u Z_v$, so that the overall product is still proportional to the identity. Terms like these give the following overall contribution
\begin{eqnarray}
&& \sum_{i=1,3,5,\dots}^{d'_{u}} {d'_u \choose i} \left( \cos^{2}{\gamma}\right)^{d'_u - i} \left( \sin^{2}{\gamma}\right)^{i} \nonumber \\
&& \; \times \sum_{j=1,3,5,\dots}^{d'_{v}} {d'_v \choose v} \left( \cos^{2}{\gamma}\right)^{d'_v - j} \left( \sin^{2}{\gamma}\right)^{j} \nonumber \\
&=& \frac{1}{4} \left( 1 - \cos^{d'_u}{2\gamma}\right) \left( 1 - \cos^{d'_v}{2\gamma}\right)
\label{eqn:lemma-xx-zz-swap-odd-terms}
\end{eqnarray}
where we have used the identity
\begin{equation}
    \sum_{j=1,3,5,\dots}^{f} {f \choose j} a^{f-j} b^{j} = \frac{1}{2} \left[ (a + b)^{f} - (a-b)^{f} \right]
\label{eqn:identity-xx-zz-lemma-sum-odd}
\end{equation}
Adding Eqs. \eqref{eqn:lemma-xx-zz-swap-even-terms} and \eqref{eqn:lemma-xx-zz-swap-odd-terms}, we get the total contribution
\begin{equation}
    \langle SWAP_{uv}\rangle = \frac{1}{2}\left(1+\cos^{d_u+d_v-2f}{2\gamma} \right)
\end{equation}
Using the earlier expressions for $\langle X_u X_v \rangle$ and $\langle Z_u Z_v \rangle$, it follows that the expression for $\langle Y_u Y_v\rangle$ is as given in the Lemma.

\end{proof}

\section{Properties of $W_{XYZ}^{(N,N,\mathcal{G})}$ and $W_{PQ}^{(N,N,\mathcal{G})}$}\label{secn:properties-W_XYZ_W_PQ}
Here we present properties of the $W_{XYZ}^{(N,N,\mathcal{G})}$ observable (which can later be used to infer those of $W_{PQ}^{(N,N,\mathcal{G})}$). In particular, we show that some of the GHZ-type states yield expectation values of $\pm3$ for an even number of qubits, and that they are eigenstates of the considered observable. Additionally, it follows from Theorem~\ref{thm:sep-W_M_k} that separable states yield expectation values that are bounded by $\pm 1$.

First, let's calculate the expectation value with respect to the {\it standard} incarnation of an $N$-qubit GHZ state
\begin{equation}
    \ket{\psi} = \frac{\ket{00\ldots0} + \ket{11\ldots 1}}{\sqrt{2}} = \frac{\ket{\bar{0}}+\ket{\bar{1}}}{\sqrt{2}}, 
\end{equation}
where we used $\ket{\bar{k}}=\ket{kk\ldots k}$ for brevity of notation. The expectation value of $W_{XYZ}^{(N,N,\mathcal{G})}$ with respect to $\ket{\psi}$ is
\begin{equation}
\bra{\psi}W_{XYZ}^{(N,N,\mathcal{G})} \ket{\psi} = \frac{3}{2} +\frac{(-1)^N}{2} + \frac{1}{2}\big((-i)^N+(i)^N\big),
\label{eq:exp_Wxyz}
\end{equation}
where we used
\begin{align}
&\bra{\bar{0}}\Zn\ket{\bar{0}}  =  1,\quad\bra{\bar{1}}\Zn\ket{\bar{1}}=  (-1)^N,\quad \bra{\bar{0}}\Zn\ket{\bar{1}} =  0,\nonumber\\
&\bra{\bar{1}}\Zn\ket{\bar{0}}  =  0,\quad \bra{\bar{0}}\Xn\ket{\bar{0}}  =  0,\quad \bra{\bar{1}}\Xn\ket{\bar{1}}  =  0,\nonumber\\
&\bra{\bar{0}}\Xn\ket{\bar{1}}  =  1,\quad\bra{\bar{1}}\Xn\ket{\bar{0}}  =  1,\quad\bra{\bar{0}}\Yn\ket{\bar{0}}  =  0,\nonumber\\
&\bra{\bar{1}}\Yn\ket{\bar{1}}  =  0,\quad\bra{\bar{0}}\Yn\ket{\bar{1}}= (-i)^N,\quad \bra{\bar{1}}\bar{Y}\ket{\bar{0}}  =  (i)^N.\nonumber
\end{align}
For $N=4k$, Eq. \eqref{eq:exp_Wxyz} gives $\bra{\psi}W_{XYZ}^{(N,N,\mathcal{G})} \ket{\psi}=+3$. It is easy to generate other $N$-qubit GHZ states that also give an expectation value of $+3$ by appropriately flipping an even number of qubits with the $X$ operator. Such a transformation leaves the expectation value unchanged, since $XY X = -Y$ and $XZX = -Z$. A further odd number of transformations will cancel the negative sign. 
There are a total of $2^{N-2}$ such GHZ-like states that produce this expectation value, which can be obtained from the following counting formula
\begin{equation}
 \sum_{i=1}^{k-1} {4k \choose 2i} + \frac{1}{2}{4k \choose 2k}+1 =    2^{N-2}.
\end{equation}
The first term counts the number of allowed flips $2,4,\ldots N/2-1$ on all possible combinations of qubits, noting that by the $\mathbb{Z}_2$ symmetry of the GHZ state, no more than half of the qubits need be flipped.
The next term is the number of ways that we can flip exactly half of the qubits, with the factor $\frac{1}{2}$ eliminating $\mathbb{Z}_2$ symmetric duplicates, and the final addition of 1 just counts the original $\ket{\psi}$ state.

Similarly, one can construct GHZ-type states for systems composed of $N=4k+2$ qubits to give the expectation value $-3$.
Starting with the state
\begin{equation}
    \ket{\phi} = \frac{\ket{00\ldots 01} - \ket{11\ldots 10}}{\sqrt{2}},
\end{equation}
and by applying similar reasoning as for the case above, one can show that $\bra{\phi}W_{XYZ}^{(N,N,\mathcal{G})}\ket{\phi} = -3$.
The number of states obtainable from $\ket{\phi}$ by an even number of flips, and that therefore produce the same expectation value $-3$ is again $2^{N-2}$ as above.

Now, we show that the states identified above are also eigenstates of $W_{XYZ}^{(N,N,\mathcal{G})}$. Knowing that 
\begin{eqnarray}
X\ket{0} =  \ket{1} & \qquad & X\ket{1}  =  \ket{0},\\
Y\ket{0} =  i\ket{1} & \qquad & Y\ket{1}  =  (-i)\ket{0},\\
Z\ket{0} =  \ket{0} & \qquad & Z\ket{1}  =  -\ket{1},
\end{eqnarray}
we see that 
\begin{eqnarray}
X^{\otimes N}\left(\ket{\bar{0}} + \ket{\bar{1}} \right) & = & \left(\ket{\bar{1}} + \ket{\bar{0}} \right),\\
Y^{\otimes N}\left(\ket{\bar{0}} + \ket{\bar{1}} \right) & = & \left((i)^N\ket{\bar{1}} +(-i)^N \ket{\bar{0}} \right),\\
Z^{\otimes N}\left(\ket{\bar{0}} + \ket{\bar{1}} \right) & = & \left(\ket{\bar{0}} + (-1)^N \ket{\bar{1}} \right),
\end{eqnarray}
so that $\ket{\psi}$ is an eigenstate of $W_{XYZ}^{(N,N,\mathcal{G})}$ for $N=4k$, and the sum of the above terms gives its corresponding eigenvalue of $+3$. Similarly, other GHZ states that are generated by flipping an even number of qubits in $\ket{\psi}$ are also eigenvectors with the same $+3$ eigenvalue. One can incorporate the bit flip ($X$) operator to all $X^{\otimes N}, Y^{\otimes N}$ or $Z^{\otimes N}$ terms and use Pauli algebra, e.g. for flipping the first two qubits we will have terms like: $(XX)\otimes (XX) = \mathbbm{1}\otimes \mathbbm{1}$, $(YX)\otimes (YX) = -Z\otimes Z$ and $ZX\otimes ZX = - Y\otimes Y$. Thus, for example, the state
\begin{equation}
    \ket{\psi_{1,2}}= \ket{1100\ldots0} + \ket{0011\ldots1} = (X_1X_2)(\ket{\bar{0}} + \ket{\bar{1}}
\end{equation}
is an eigenstate of $W_{XYZ}^{(N,N,\mathcal{G})}$, since $X^{\otimes N} \ket{\psi_{1,2}} = Y^{\otimes N}\ket{\psi_{1,2}} = \ket{\psi_{1,2}}$, and perhaps less obviously
\begin{align}
   & \Yn \ket{\psi_{1,2}} = -Z_1Z_2 Y^{\otimes (N-2)}\left(\ket{\bar{0}} + \ket{\bar{1}} \right) \nonumber\\
   &= - (i)^{N-2}\ket{0011\ldots1}- (-1)^{N-2}\ket{1100\ldots0}\nonumber \\
   &= - (i)^{N-2}X_1X_2 \ket{1} - (-i)^{N-2} X_1X_2\ket{0} = \ket{\psi_{1,2}},
\end{align}
and similarly for all other states obtained by flipping an even number of qubits in $\ket{\psi}$.
For eigenstates with the $-3$ eigenvalue, the reasoning follows analogous steps.

Based on Theorem~\ref{thm:sep-W_M_k}, the upper and lower bounds of $W_{XYZ}^{(N,N,\mathcal{G})}$ are $+1$ and $-1$ respectively, since we have only a single generalized edge. Therefore, this observable serves as a provable entanglement witness \cite{Bae_2020} for an even number of qubits. For odd numbers of qubits, we numerically found the expectation value to be $\pm\sqrt{3}$. Similarly, if we restrict to only two terms (e.g. $XZ$, $XY$ and $YZ$), one can follow essentially the same logic and demonstrate that an even number of qubits give an expectation value $\pm2$ in GHZ states, while an odd number of qubits are numerically found to give an expectation value of $\pm\sqrt{2}$. Once again, these expectation values violate the separable thresholds of $\pm |E|=\pm 1$ from Theorem~\ref{thm:sep-W_M_k}.
\bibliography{refs}

\begin{thebibliography}{55}
\providecommand{\natexlab}[1]{#1}
\providecommand{\url}[1]{\texttt{#1}}
\expandafter\ifx\csname urlstyle\endcsname\relax
  \providecommand{\doi}[1]{doi: #1}\else
  \providecommand{\doi}{doi: \begingroup \urlstyle{rm}\Url}\fi

\bibitem[Mol(2018)]{Moll_2018}
2018.
\newblock \doi{10.1088/2058-9565/aab822}.
\newblock URL \url{https://doi.org/10.1088/2058-9565/aab822}.

\bibitem[loc(2020)]{local-noise-qaoa}
2020.
\newblock \doi{10.1088/2633-1357/abb0d7}.
\newblock URL \url{https://doi.org/10.1088/2633-1357/abb0d7}.

\bibitem[Xue(2021)]{Xue_2021}
2021.
\newblock \doi{10.1088/0256-307x/38/3/030302}.
\newblock URL \url{https://doi.org/10.1088/0256-307x/38/3/030302}.

\bibitem[Aspect et~al.(1982{\natexlab{a}})Aspect, Dalibard, and
  Roger]{Aspect1982}
Alain Aspect, Jean Dalibard, and G\'erard Roger.
\newblock Experimental test of bell's inequalities using time-varying
  analyzers.
\newblock \emph{Phys. Rev. Lett.}, 49:\penalty0 1804--1807, 1982{\natexlab{a}}.

\bibitem[Aspect et~al.(1982{\natexlab{b}})Aspect, Grangier, and
  Roger]{Aspect1982_2}
Alain Aspect, Philippe Grangier, and G\'erard Roger.
\newblock Experimental realization of einstein-podolsky-rosen-bohm
  gedankenexperiment: A new violation of bell's inequalities.
\newblock \emph{Phys. Rev. Lett.}, 49:\penalty0 91--94, 1982{\natexlab{b}}.

\bibitem[Bae et~al.(2020)Bae, Chruściński, and Hiesmayr]{Bae_2020}
Joonwoo Bae, Dariusz Chruściński, and Beatrix~C. Hiesmayr.
\newblock Mirrored entanglement witnesses.
\newblock \emph{npj Quantum Information}, 6\penalty0 (1), 2020.
\newblock ISSN 2056-6387.
\newblock \doi{10.1038/s41534-020-0242-z}.
\newblock URL \url{http://dx.doi.org/10.1038/s41534-020-0242-z}.

\bibitem[Beale et~al.()Beale, Carignan-Dugas, Dahlen, Emerson, Hincks, Iyer,
  Jain, Hufnagel, Ospadov, Saunders, Stasiuk, Wallman, and
  Winick]{beale_3945250}
Stefanie~J. Beale, Arnaud Carignan-Dugas, Dar Dahlen, Joseph Emerson, Ian
  Hincks, Pavithran Iyer, Aditya Jain, David Hufnagel, Egor Ospadov, Jordan
  Saunders, Andrew Stasiuk, Joel~J. Wallman, and Adam Winick.
\newblock True-q.

\bibitem[Bell(2004)]{Bell2004-BELSAU}
John~Stewart Bell.
\newblock \emph{Speakable and Unspeakable in Quantum Mechanics: Collected
  Papers on Quantum Philosophy}.
\newblock Cambridge University Press, 2004.

\bibitem[Bengtsson and {\.Z}yczkowski(2017)]{bengtsson2017geometry}
Ingemar Bengtsson and Karol {\.Z}yczkowski.
\newblock \emph{Geometry of quantum states: an introduction to quantum
  entanglement}.
\newblock Cambridge university press, 2017.

\bibitem[Bera et~al.(2017)Bera, Das, Sadhukhan, Singha~Roy, Sen(De), and
  Sen]{Bera_2017}
Anindita Bera, Tamoghna Das, Debasis Sadhukhan, Sudipto Singha~Roy, Aditi
  Sen(De), and Ujjwal Sen.
\newblock Quantum discord and its allies: a review of recent progress.
\newblock \emph{Reports on Progress in Physics}, 81\penalty0 (2), 2017.
\newblock ISSN 1361-6633.
\newblock \doi{10.1088/1361-6633/aa872f}.
\newblock URL \url{http://dx.doi.org/10.1088/1361-6633/aa872f}.

\bibitem[Biamonte et~al.(2017)Biamonte, Wittek, Pancotti, Rebentrost, Wiebe,
  and Lloyd]{biamonte2017quantum}
Jacob Biamonte, Peter Wittek, Nicola Pancotti, Patrick Rebentrost, Nathan
  Wiebe, and Seth Lloyd.
\newblock Quantum machine learning.
\newblock \emph{Nature}, 549\penalty0 (7671):\penalty0 195--202, 2017.

\bibitem[Blume-Kohout(2010)]{Blume_Kohout_2010}
Robin Blume-Kohout.
\newblock Optimal, reliable estimation of quantum states.
\newblock \emph{New Journal of Physics}, 12\penalty0 (4), 2010.
\newblock ISSN 1367-2630.
\newblock \doi{10.1088/1367-2630/12/4/043034}.
\newblock URL \url{http://dx.doi.org/10.1088/1367-2630/12/4/043034}.

\bibitem[Bravyi et~al.(2020)Bravyi, Kliesch, Koenig, and Tang]{Bravyi_2020}
Sergey Bravyi, Alexander Kliesch, Robert Koenig, and Eugene Tang.
\newblock Obstacles to variational quantum optimization from symmetry
  protection.
\newblock \emph{Physical Review Letters}, 125\penalty0 (26), 2020.
\newblock ISSN 1079-7114.
\newblock \doi{10.1103/physrevlett.125.260505}.
\newblock URL \url{http://dx.doi.org/10.1103/PhysRevLett.125.260505}.

\bibitem[Caldwell et~al.(2018)Caldwell, Didier, Ryan, Sete, Hudson, Karalekas,
  Manenti, da~Silva, Sinclair, Acala, Alidoust, Angeles, Bestwick, Block,
  Bloom, Bradley, Bui, Capelluto, Chilcott, Cordova, Crossman, Curtis,
  Deshpande, Bouayadi, Girshovich, Hong, Kuang, Lenihan, Manning, Marchenkov,
  Marshall, Maydra, Mohan, O'Brien, Osborn, Otterbach, Papageorge, Paquette,
  Pelstring, Polloreno, Prawiroatmodjo, Rawat, Reagor, Renzas, Rubin, Russell,
  Rust, Scarabelli, Scheer, Selvanayagam, Smith, Staley, Suska, Tezak,
  Thompson, To, Vahidpour, Vodrahalli, Whyland, Yadav, Zeng, and
  Rigetti]{Caldwell2018}
S.~A. Caldwell, N.~Didier, C.~A. Ryan, E.~A. Sete, A.~Hudson, P.~Karalekas,
  R.~Manenti, M.~P. da~Silva, R.~Sinclair, E.~Acala, N.~Alidoust, J.~Angeles,
  A.~Bestwick, M.~Block, B.~Bloom, A.~Bradley, C.~Bui, L.~Capelluto,
  R.~Chilcott, J.~Cordova, G.~Crossman, M.~Curtis, S.~Deshpande, T.~El
  Bouayadi, D.~Girshovich, S.~Hong, K.~Kuang, M.~Lenihan, T.~Manning,
  A.~Marchenkov, J.~Marshall, R.~Maydra, Y.~Mohan, W.~O'Brien, C.~Osborn,
  J.~Otterbach, A.~Papageorge, J.-P. Paquette, M.~Pelstring, A.~Polloreno,
  G.~Prawiroatmodjo, V.~Rawat, M.~Reagor, R.~Renzas, N.~Rubin, D.~Russell,
  M.~Rust, D.~Scarabelli, M.~Scheer, M.~Selvanayagam, R.~Smith, A.~Staley,
  M.~Suska, N.~Tezak, D.~C. Thompson, T.-W. To, M.~Vahidpour, N.~Vodrahalli,
  T.~Whyland, K.~Yadav, W.~Zeng, and C.~Rigetti.
\newblock Parametrically activated entangling gates using transmon qubits.
\newblock \emph{Phys. Rev. Applied}, 10:\penalty0 034050, 2018.

\bibitem[Carvalho et~al.(2004)Carvalho, Mintert, and Buchleitner]{conc}
Andr\'e R.~R. Carvalho, Florian Mintert, and Andreas Buchleitner.
\newblock Decoherence and multipartite entanglement.
\newblock \emph{Phys. Rev. Lett.}, 93:\penalty0 230501, 2004.

\bibitem[Chen and Chen(2007)]{PhysRevA.76.022330}
Lin Chen and Yi-Xin Chen.
\newblock Multiqubit entanglement witness.
\newblock \emph{Phys. Rev. A}, 76:\penalty0 022330, 2007.

\bibitem[Chru{\'s}ci{\'n}ski and Sarbicki(2014)]{chruscinski2014entanglement}
Dariusz Chru{\'s}ci{\'n}ski and Gniewomir Sarbicki.
\newblock Entanglement witnesses: construction, analysis and classification.
\newblock \emph{Journal of Physics A: Mathematical and Theoretical},
  47\penalty0 (48):\penalty0 483001, 2014.

\bibitem[Cohen-Tannoudji(1962)]{cohen1962theorie}
Claude Cohen-Tannoudji.
\newblock \emph{Th{\'e}orie quantique du cycle de pompage optique.
  V{\'e}rification exp{\'e}rimentale des nouveaux effets pr{\'e}vus}.
\newblock PhD thesis, Universit{\'e} Paris, 1962.

\bibitem[Cornelissen et~al.(2021)Cornelissen, Bausch, and
  Gilyén]{cornelissen2021scalable}
Arjan Cornelissen, Johannes Bausch, and András Gilyén.
\newblock Scalable benchmarks for gate-based quantum computers, 2021.

\bibitem[Cross et~al.(2019)Cross, Bishop, Sheldon, Nation, and
  Gambetta]{PhysRevA.100.032328}
Andrew~W. Cross, Lev~S. Bishop, Sarah Sheldon, Paul~D. Nation, and Jay~M.
  Gambetta.
\newblock Validating quantum computers using randomized model circuits.
\newblock \emph{Phys. Rev. A}, 100:\penalty0 032328, 2019.

\bibitem[Datta et~al.(2008)Datta, Shaji, and Caves]{datta2008quantum}
Animesh Datta, Anil Shaji, and Carlton~M Caves.
\newblock Quantum discord and the power of one qubit.
\newblock \emph{Physical review letters}, 100\penalty0 (5):\penalty0 050502,
  2008.

\bibitem[DiVincenzo and Loss(1999)]{DiVincenzo_1999}
David~P DiVincenzo and Daniel Loss.
\newblock Quantum computers and quantum coherence.
\newblock \emph{Journal of Magnetism and Magnetic Materials}, 200\penalty0
  (1-3), 1999.
\newblock ISSN 0304-8853.
\newblock \doi{10.1016/s0304-8853(99)00315-7}.
\newblock URL \url{http://dx.doi.org/10.1016/S0304-8853(99)00315-7}.

\bibitem[Díez-Valle et~al.(2021)Díez-Valle, Porras, and
  García-Ripoll]{diezvalle2021quantum}
Pablo Díez-Valle, Diego Porras, and Juan~José García-Ripoll.
\newblock Quantum variational optimization: the role of entanglement and
  problem hardness, 2021.

\bibitem[D’Ariano et~al.(2002)D’Ariano, Laurentis, Paris, Porzio, and
  Solimeno]{D_Ariano_2002}
G~Mauro D’Ariano, Martina~De Laurentis, Matteo G~A Paris, Alberto Porzio, and
  Salvatore Solimeno.
\newblock Quantum tomography as a tool for the characterization of optical
  devices.
\newblock \emph{Journal of Optics B: Quantum and Semiclassical Optics},
  4\penalty0 (3), 2002.
\newblock ISSN 1741-3575.
\newblock \doi{10.1088/1464-4266/4/3/366}.
\newblock URL \url{http://dx.doi.org/10.1088/1464-4266/4/3/366}.

\bibitem[Farhi and Neven(2018)]{farhi2018classification}
Edward Farhi and Hartmut Neven.
\newblock Classification with quantum neural networks on near term processors,
  2018.

\bibitem[Farhi et~al.(2014)Farhi, Goldstone, and Gutmann]{farhi2014quantum}
Edward Farhi, Jeffrey Goldstone, and Sam Gutmann.
\newblock A quantum approximate optimization algorithm, 2014.

\bibitem[Gold et~al.(2021)Gold, Paquette, Stockklauser, Reagor, Alam, Bestwick,
  Didier, Nersisyan, Oruc, Razavi, Scharmann, Sete, Sur, Venturelli,
  Winkleblack, Wudarski, Harburn, and Rigetti]{gold2021entanglement}
Alysson Gold, JP~Paquette, Anna Stockklauser, Matthew~J. Reagor, M.~Sohaib
  Alam, Andrew Bestwick, Nicolas Didier, Ani Nersisyan, Feyza Oruc, Armin
  Razavi, Ben Scharmann, Eyob~A. Sete, Biswajit Sur, Davide Venturelli,
  Cody~James Winkleblack, Filip Wudarski, Mike Harburn, and Chad Rigetti.
\newblock Entanglement across separate silicon dies in a modular
  superconducting qubit device, 2021.

\bibitem[G\"uhne et~al.(2007)G\"uhne, Lu, Gao, and Pan]{PhysRevA.76.030305}
Otfried G\"uhne, Chao-Yang Lu, Wei-Bo Gao, and Jian-Wei Pan.
\newblock Toolbox for entanglement detection and fidelity estimation.
\newblock \emph{Phys. Rev. A}, 76:\penalty0 030305, 2007.

\bibitem[Gühne and Tóth(2009)]{entanglement-detection-review}
Otfried Gühne and Géza Tóth.
\newblock Entanglement detection.
\newblock \emph{Physics Reports}, 474\penalty0 (1):\penalty0 1--75, 2009.
\newblock ISSN 0370-1573.
\newblock \doi{https://doi.org/10.1016/j.physrep.2009.02.004}.
\newblock URL
  \url{https://www.sciencedirect.com/science/article/pii/S0370157309000623}.

\bibitem[Hill and Wootters(1997)]{PhysRevLett.78.5022}
Scott Hill and William~K. Wootters.
\newblock Entanglement of a pair of quantum bits.
\newblock \emph{Phys. Rev. Lett.}, 78:\penalty0 5022--5025, 1997.

\bibitem[Horodecki(2001)]{horodecki2001entanglement}
Michal Horodecki.
\newblock Entanglement measures.
\newblock \emph{Quantum Inf. Comput.}, 1\penalty0 (1):\penalty0 3--26, 2001.

\bibitem[Horodecki et~al.(2009)Horodecki, Horodecki, Horodecki, and
  Horodecki]{HorodeckiQuantumEntanglement}
Ryszard Horodecki, Pawe\l{} Horodecki, Micha\l{} Horodecki, and Karol
  Horodecki.
\newblock Quantum entanglement.
\newblock \emph{Rev. Mod. Phys.}, 81:\penalty0 865--942, 2009.

\bibitem[Hurwitz(1897)]{Hurwitz1897}
A.~Hurwitz.
\newblock Über die erzeugung der invarianten durch integration.
\newblock \emph{Nachrichten von der Gesellschaft der Wissenschaften zu
  Göttingen, Mathematisch-Physikalische Klasse}, 1897:\penalty0 71--2, 1897.
\newblock URL \url{http://eudml.org/doc/58378}.

\bibitem[Job and Lidar(2018)]{job2018test}
Joshua Job and Daniel Lidar.
\newblock Test-driving 1000 qubits.
\newblock \emph{Quantum Science and Technology}, 3\penalty0 (3):\penalty0
  030501, 2018.

\bibitem[Jozsa(1994)]{JozsaFidelity}
Richard Jozsa.
\newblock Fidelity for mixed quantum states.
\newblock \emph{Journal of Modern Optics}, 41\penalty0 (12):\penalty0
  2315--2323, 1994.
\newblock \doi{10.1080/09500349414552171}.
\newblock URL \url{https://doi.org/10.1080/09500349414552171}.

\bibitem[Jozsa and Linden(2003)]{Jozsa_2003}
Richard Jozsa and Noah Linden.
\newblock On the role of entanglement in quantum-computational speed-up.
\newblock \emph{Proceedings of the Royal Society of London. Series A:
  Mathematical, Physical and Engineering Sciences}, 459\penalty0 (2036), 2003.
\newblock ISSN 1471-2946.
\newblock \doi{10.1098/rspa.2002.1097}.
\newblock URL \url{http://dx.doi.org/10.1098/rspa.2002.1097}.

\bibitem[McClean et~al.(2016)McClean, Romero, Babbush, and
  Aspuru-Guzik]{mcclean2016theory}
Jarrod~R McClean, Jonathan Romero, Ryan Babbush, and Al{\'a}n Aspuru-Guzik.
\newblock The theory of variational hybrid quantum-classical algorithms.
\newblock \emph{New Journal of Physics}, 18\penalty0 (2):\penalty0 023023,
  2016.

\bibitem[Mieghem()]{van_meighem1}
Piet~Van Mieghem.
\newblock A new type of lower bound for the largest eigenvalue of a symmetric
  matrix.
\newblock \emph{Elsevier: Linear Algebra and its Applications 427 (2007)
  119-129}.
\newblock URL
  \url{https://www.nas.ewi.tudelft.nl/people/Piet/papers/LAA_2007_lower_bound_largest_eigenv_symm_matrix.pdf}.

\bibitem[Paris and Sacchi(2003)]{d2003quantum}
Matteo~GA Paris and Massimiliano~F Sacchi.
\newblock Quantum tomography.
\newblock \emph{Advances in Imaging and Electron Physics}, 128:\penalty0
  206--309, 2003.

\bibitem[Peruzzo et~al.(2014)Peruzzo, McClean, Shadbolt, Yung, Zhou, Love,
  Aspuru-Guzik, and O’brien]{peruzzo2014variational}
Alberto Peruzzo, Jarrod McClean, Peter Shadbolt, Man-Hong Yung, Xiao-Qi Zhou,
  Peter~J Love, Al{\'a}n Aspuru-Guzik, and Jeremy~L O’brien.
\newblock A variational eigenvalue solver on a photonic quantum processor.
\newblock \emph{Nature communications}, 5\penalty0 (1):\penalty0 1--7, 2014.

\bibitem[Plenio and Virmani(2014)]{plenio2014introduction}
Martin~B Plenio and Shashank~S Virmani.
\newblock An introduction to entanglement theory.
\newblock \emph{Quantum information and coherence}, pages 173--209, 2014.

\bibitem[Reagor et~al.(2018)Reagor, Osborn, Tezak, Staley, Prawiroatmodjo,
  Scheer, Alidoust, Sete, Didier, da~Silva, Acala, Angeles, Bestwick, Block,
  Bloom, Bradley, Bui, Caldwell, Capelluto, Chilcott, Cordova, Crossman,
  Curtis, Deshpande, El~Bouayadi, Girshovich, Hong, Hudson, Karalekas, Kuang,
  Lenihan, Manenti, Manning, Marshall, Mohan, O{\textquoteright}Brien,
  Otterbach, Papageorge, Paquette, Pelstring, Polloreno, Rawat, Ryan, Renzas,
  Rubin, Russel, Rust, Scarabelli, Selvanayagam, Sinclair, Smith, Suska, To,
  Vahidpour, Vodrahalli, Whyland, Yadav, Zeng, and Rigetti]{Reagor2018}
Matthew Reagor, Christopher~B. Osborn, Nikolas Tezak, Alexa Staley, Guenevere
  Prawiroatmodjo, Michael Scheer, Nasser Alidoust, Eyob~A. Sete, Nicolas
  Didier, Marcus~P. da~Silva, Ezer Acala, Joel Angeles, Andrew Bestwick,
  Maxwell Block, Benjamin Bloom, Adam Bradley, Catvu Bui, Shane Caldwell,
  Lauren Capelluto, Rick Chilcott, Jeff Cordova, Genya Crossman, Michael
  Curtis, Saniya Deshpande, Tristan El~Bouayadi, Daniel Girshovich, Sabrina
  Hong, Alex Hudson, Peter Karalekas, Kat Kuang, Michael Lenihan, Riccardo
  Manenti, Thomas Manning, Jayss Marshall, Yuvraj Mohan, William
  O{\textquoteright}Brien, Johannes Otterbach, Alexander Papageorge,
  Jean-Philip Paquette, Michael Pelstring, Anthony Polloreno, Vijay Rawat,
  Colm~A. Ryan, Russ Renzas, Nick Rubin, Damon Russel, Michael Rust, Diego
  Scarabelli, Michael Selvanayagam, Rodney Sinclair, Robert Smith, Mark Suska,
  Ting-Wai To, Mehrnoosh Vahidpour, Nagesh Vodrahalli, Tyler Whyland, Kamal
  Yadav, William Zeng, and Chad~T. Rigetti.
\newblock Demonstration of universal parametric entangling gates on a
  multi-qubit lattice.
\newblock \emph{Science Advances}, 4\penalty0 (2), 2018.
\newblock \doi{10.1126/sciadv.aao3603}.
\newblock URL \url{https://advances.sciencemag.org/content/4/2/eaao3603}.

\bibitem[R{\o}nnow et~al.(2014)R{\o}nnow, Wang, Job, Boixo, Isakov, Wecker,
  Martinis, Lidar, and Troyer]{ronnow2014defining}
Troels~F R{\o}nnow, Zhihui Wang, Joshua Job, Sergio Boixo, Sergei~V Isakov,
  David Wecker, John~M Martinis, Daniel~A Lidar, and Matthias Troyer.
\newblock Defining and detecting quantum speedup.
\newblock \emph{science}, 345\penalty0 (6195):\penalty0 420--424, 2014.

\bibitem[Shaydulin et~al.(2020)Shaydulin, Hadfield, Hogg, and
  Safro]{shaydulin2020classical}
Ruslan Shaydulin, Stuart Hadfield, Tad Hogg, and Ilya Safro.
\newblock Classical symmetries and qaoa, 2020.

\bibitem[Vasconcelos et~al.()Vasconcelos, Kjaergaard, Menke, Gustavsson,
  Orlando, and Oliver]{vasconcelosextending}
Francisca Vasconcelos, Morten Kjaergaard, Tim Menke, Simon Gustavsson, Terry~P
  Orlando, and William~D Oliver.
\newblock Extending quantum state tomography for superconducting quantum
  processors.

\bibitem[Wallman and Emerson(2016)]{randomized-compiling}
Joel~J. Wallman and Joseph Emerson.
\newblock Noise tailoring for scalable quantum computation via randomized
  compiling.
\newblock \emph{Phys. Rev. A}, 94:\penalty0 052325, 2016.

\bibitem[Wang et~al.(2020)Wang, Song, Zhao, Wang, and Wang]{wang2020detecting}
Kun Wang, Zhixin Song, Xuanqiang Zhao, Zihe Wang, and Xin Wang.
\newblock Detecting and quantifying entanglement on near-term quantum devices,
  2020.

\bibitem[Wang et~al.(2018)Wang, Hadfield, Jiang, and Rieffel]{jw-qaoa}
Zhihui Wang, Stuart Hadfield, Zhang Jiang, and Eleanor~G. Rieffel.
\newblock Quantum approximate optimization algorithm for maxcut: A fermionic
  view.
\newblock \emph{Phys. Rev. A}, 97:\penalty0 022304, 2018.

\bibitem[Wecker et~al.(2015)Wecker, Hastings, and Troyer]{Wecker_2015}
Dave Wecker, Matthew~B. Hastings, and Matthias Troyer.
\newblock Progress towards practical quantum variational algorithms.
\newblock \emph{Physical Review A}, 92\penalty0 (4), 2015.
\newblock ISSN 1094-1622.
\newblock \doi{10.1103/physreva.92.042303}.
\newblock URL \url{http://dx.doi.org/10.1103/PhysRevA.92.042303}.

\bibitem[Wiersema et~al.(2020)Wiersema, Zhou, de~Sereville, Carrasquilla, Kim,
  and Yuen]{Wiersema_2020}
Roeland Wiersema, Cunlu Zhou, Yvette de~Sereville, Juan~Felipe Carrasquilla,
  Yong~Baek Kim, and Henry Yuen.
\newblock Exploring entanglement and optimization within the hamiltonian
  variational ansatz.
\newblock \emph{PRX Quantum}, 1\penalty0 (2), 2020.
\newblock ISSN 2691-3399.
\newblock \doi{10.1103/prxquantum.1.020319}.
\newblock URL \url{http://dx.doi.org/10.1103/PRXQuantum.1.020319}.

\bibitem[Woitzik et~al.(2020)Woitzik, Barkoutsos, Wudarski, Buchleitner, and
  Tavernelli]{Woitzik_2020}
Andreas J.~C. Woitzik, Panagiotis~Kl. Barkoutsos, Filip Wudarski, Andreas
  Buchleitner, and Ivano Tavernelli.
\newblock Entanglement production and convergence properties of the variational
  quantum eigensolver.
\newblock \emph{Physical Review A}, 102\penalty0 (4), 2020.
\newblock ISSN 2469-9934.
\newblock \doi{10.1103/physreva.102.042402}.
\newblock URL \url{http://dx.doi.org/10.1103/PhysRevA.102.042402}.

\bibitem[Wootters(1998)]{Wootters_1998}
William~K. Wootters.
\newblock Entanglement of formation of an arbitrary state of two qubits.
\newblock \emph{Physical Review Letters}, 80\penalty0 (10), 1998.
\newblock ISSN 1079-7114.
\newblock \doi{10.1103/physrevlett.80.2245}.
\newblock URL \url{http://dx.doi.org/10.1103/PhysRevLett.80.2245}.

\bibitem[Wu and Wu(2009)]{cauchy-schwarz}
Hui-Hua Wu and Shanhe Wu.
\newblock Various proofs of the cauchy-schwarz inequality.
\newblock \emph{Octogon Mathematical Magazine}, 17\penalty0 (1):\penalty0
  221--229, 2009.

\bibitem[Zyczkowski and Sommers(2001)]{Zyczkowski_2001}
Karol Zyczkowski and Hans-Jürgen Sommers.
\newblock Induced measures in the space of mixed quantum states.
\newblock \emph{Journal of Physics A: Mathematical and General}, 34\penalty0
  (35), 2001.
\newblock ISSN 1361-6447.
\newblock \doi{10.1088/0305-4470/34/35/335}.
\newblock URL \url{http://dx.doi.org/10.1088/0305-4470/34/35/335}.

\bibitem[Życzkowski et~al.(1998)Życzkowski, Horodecki, Sanpera, and
  Lewenstein]{Zyczkowski_1998}
Karol Życzkowski, Paweł Horodecki, Anna Sanpera, and Maciej Lewenstein.
\newblock Volume of the set of separable states.
\newblock \emph{Physical Review A}, 58\penalty0 (2), 1998.
\newblock ISSN 1094-1622.
\newblock \doi{10.1103/physreva.58.883}.
\newblock URL \url{http://dx.doi.org/10.1103/PhysRevA.58.883}.

\end{thebibliography}
\end{document}